\documentclass[11pt,english]{article}
\usepackage{amsmath,amssymb,amsthm}
\usepackage{multicol}
\usepackage[margin=1in]{geometry}
\usepackage{graphicx,color}
\usepackage{enumitem}
\usepackage{fullpage}
\usepackage[noblocks]{authblk}
\usepackage{tcolorbox}
\usepackage{babel}
\usepackage{wrapfig}
\usepackage{MnSymbol}
\usepackage{mdframed}

\usepackage{algorithm}
\usepackage{algpseudocode}

\usepackage{tablefootnote}
\usepackage{thm-restate}
\usepackage{bbding}
\usepackage{pifont}
\usepackage{nicefrac}

\definecolor{Darkblue}{rgb}{0,0,0.4}
\definecolor{Brown}{cmyk}{0,0.61,1.,0.60}
\definecolor{Purple}{cmyk}{0.45,0.86,0,0}
\definecolor{Darkgreen}{rgb}{0.133,0.543,0.133}

\usepackage[colorlinks,linkcolor=Darkblue,filecolor=blue,citecolor=blue,urlcolor=Darkblue,pagebackref]{hyperref}
\usepackage[nameinlink]{cleveref}

\usepackage[colorinlistoftodos,prependcaption,textsize=tiny]{todonotes}
\newif\ifdraft 
\draftfalse

\newtheorem{theorem}{Theorem}
\newtheorem{lemma}{Lemma}
\newtheorem{Lem}{Lemma}
\newtheorem{definition}{Definition}
\newtheorem{claim}{Claim}
\newtheorem{observation}{Observation}
\newtheorem{corollary}{Corollary}

\newtheorem{question}{Question}

\let\int\undefined

\newcommand{\port}{\mathsf{Port}}

\newcommand{\polylog}{\mathrm{polylog}}
\newcommand{\poly}{\mathrm{poly}}
\newcommand{\depth}{\mathsf{depth}}

\newcommand{\parent}{\mathsf{parent}}
\newcommand{\im}{\mathsf{cpr}}

\newcommand{\ext}{\mathsf{Ext}}
\newcommand{\int}{\mathsf{Int}}
\newcommand{\leavs}{\mathsf{L}}
\newcommand{\home}{\mathsf{home}}
\newcommand{\ansc}{\mathsf{Ansc}}

\newcommand{\SP}{\mathsf{SP}}
\newcommand{\cpr}{\mathsf{cpr}}
\newcommand{\spc}{\mathsf{Space}}
\newcommand{\md}{\mathcal{D}}
\newcommand{\mt}{\mathcal{T}}
\newcommand{\mx}{\mathcal{X}}
\newcommand{\my}{\mathcal{Y}}
\newcommand{\mr}{\mathcal{R}}
\newcommand{\lcads}{\mathcal{LCA}}
\newcommand{\mword}{\overline{\omega}}

\newcommand{\lport}{\mathsf{LPort}}
\newcommand{\portrev}{\overline{\mathsf{Port}}}
\newcommand{\lportrev}{\overline{\mathsf{LPort}}}

\newcommand{\lcadepth}{\mathsf{LCADepth}}

 \def\rank{\operatorname{rank}}
  \def\head{\operatorname{head}}

 \newcommand{\eps}{\epsilon}

\DeclareMathOperator*{\argmax}{arg\,max}

\def\eps{\epsilon}

\newcommand{\initOneLiners}{%
	\setlength{\itemsep}{0pt}
	\setlength{\parsep }{0pt}
	\setlength{\topsep }{0pt}
}

\title{Optimal Approximate Distance Oracle for Planar Graphs}
\author{Hung Le}
\affil{University of Massachusetts at Amherst}
\author{Christian Wulff-Nilsen}
\affil{University of Copenhagen}

\date{}
\begin{document}
\maketitle
\begin{abstract}
	A $(1+\epsilon)$-approximate distance oracle of an edge-weighted graph is a data structure that returns an approximate shortest path distance between any two query vertices up to a $(1+\epsilon)$ factor. Thorup (FOCS 2001, JACM 2004) and Klein (SODA 2002) independently constructed a $(1+\epsilon)$-approximate distance oracle with $O(n\log n)$ space, measured in number of words, and $O(1)$ query time when $G$ is an undirected planar graph with $n$ vertices and $\epsilon$ is a fixed constant. Many follow-up works gave $(1+\epsilon)$-approximate distance oracles with various trade-offs between space and query time. However, improving $O(n\log n)$ space bound without sacrificing query time remains an open problem for almost two decades. In this work, we resolve this problem affirmatively by constructing a $(1+\epsilon)$-approximate distance oracle with optimal $O(n)$ space and $O(1)$ query time for undirected planar graphs and fixed $\epsilon$.

We also make substantial progress for planar digraphs with non-negative edge weights. For fixed $\epsilon > 0$, we give a $(1+\epsilon)$-approximate distance oracle with space $o(n\log(Nn))$ and $O(\log\log(Nn)$ query time; here $N$ is the ratio between the largest and smallest positive edge weight. This improves Thorup's (FOCS 2001, JACM 2004) $O(n\log(Nn)\log n)$ space bound by more than a logarithmic factor while matching the query time of his structure. This is the first improvement for planar digraphs in two decades, both in the weighted and unweighted setting.

\end{abstract}
\pagebreak
{\small \setcounter{tocdepth}{2} \tableofcontents}
\newpage
\pagenumbering{arabic}

\section{Introduction}

Computing shortest path distances in edge-weighted planar graphs\footnote{Graphs considered in this paper are edge-weighted unless specified otherwise.} is a fundamental problem with lots of practical applications, such as planning, logistics, and traffic simulation~\cite{Sommer14}. The celebrated algorithm of Henzinger et al.~\cite{HKRS97} can compute the shortest path distance in a planar graph between any given pair of vertices in linear time. However, when we are given a large number of distance queries and the network is huge, a linear time algorithm for answering \emph{each} distance query becomes unsatisfactory. This motivates the development of \emph{(exact or approximate) distance oracles}: a data structure that can quickly answer any (exact or approximate) distance query.  

Trees are simplest planar graphs that admit exact distance oracles with linear space and constant query time; the construction is by a simple reduction to constructing a lowest common ancestor (LCA) data structure. This basic fact leads to the following fundamental question that has been driving the field:

\begin{question}\label{quest:motivating}
	Is it possible to construct  a distance oracle, exact or approximate, for edge-weighted  planar graphs, directed or undirected, with guarantees matching those for trees: linear space and constant query time?
\end{question}

Despite significant research efforts spanning more than two decades~\cite{Klein02,Thorup04,MZ07,WulffNilsen16,KKS11,GX19,CS19,KST13,ACG12,MS12,GMWW18,CDW17,CGMW19,LP20}, \Cref{quest:motivating} remains largely open in all four basic settings: exact/approximate oracles for  directed/undirected planar graphs.

Cohen-Addad, Dahlgaard and Wulff-Nilsen~\cite{CDW17} constructed the first exact distance oracle with truly subquadratic space and $O(\log(n))$\footnote{In this paper, $n$ is the number of vertices of the graph.} query time for \emph{directed} planar graphs; previous exact oracles with truly subquadratic space had polynomial query time. Several follow-up works~\cite{GMWW18,CGMW19,LP20} improved the oracle of Cohen-Addad, Dahlgaard and Wulff-Nilsen~\cite{CDW17} in several ways, getting either space $\tilde O(n)$ and $n^{o(1)}$ query time or space $n^{1+o(1)}$ and $\tilde O(1)$ query time~\cite{LP20}. ($\tilde{O}$ notation hides a polylogarithmic factor of  $n$.) However, no exact distance oracle with \emph{constant query time} and \emph{truly subquadratic} space is known, even for undirected weighted or unweighted directed planar graphs. For a special case of \emph{unweighted}, undirected planar graphs,  Fredslund-Hansen, Mozes and Wulff-Nilsen~\cite{FMW20} recently constructed an exact distance oracle with $O(n^{\frac{5}{3}+\eps})$ space and $O(\log\frac{1}{\eps})$ query time for any choice of parameter $\epsilon > 0$. If we only want to query exact distances of value at most a \emph{constant} $k$ in an unweighted, directed planar graph, Kowalik and Kurowski~\cite{KK03} showed that an oracle with linear space and constant query time exists. For a more thorough review of exact distance oracles in planar graphs, see \Cref{subsec:related-work}.

Given the remote prospect of answering \Cref{quest:motivating} for exact distance oracles, we focus on  \emph{$(1+\eps)$-approximate} distance oracles where the query output is never smaller than the true shortest path distance and not larger by more than a factor of $(1+\epsilon)$, for any given $\epsilon > 0$. Note that \Cref{quest:motivating} is only relevant when $\eps$ is a fixed constant, and this is the most basic regime that we are interested in.

Over the past 20 years, significant progress has been made in constructing $(1+\eps)$-approximate distance oracle for \emph{undirected} planar graphs.  In their seminal papers, Thorup~\cite{Thorup04} and Klein~\cite{Klein02} indpendently constructed distance oracles with $O(n(\log n)\epsilon^{-1})$ space\footnote{Unless otherwise stated, each space bound in this paper is in the number of \emph{words} that the oracle uses.} and $O(\epsilon^{-1})$ query time. When $\epsilon = \Theta(1)$, their oracles have $O(n\log n)$ space and $O(1)$ query time. (Henceforth, we do not spell out the dependency on $\eps$ unless it is important to do so.) Kawarabayashi, Klein and Sommer~\cite{KKS11} reduced the space to $O(n)$ at the cost of increasing the query time to $O(\log^2 n)$. Kawarabayashi, Sommer and Thorup~\cite{KST13} designed a distance oracle with $\overline{O}(n)$ space and $\overline{O}(1)$ query time when \emph{every edge has weight at most} $\log^{O(1)}(n)$ (and at least $1$); $\overline{O}(.)$ notation hides a $\log\log n$ factor.  Wullf-Nilsen~\cite{WulffNilsen16} was the first to break the $\Omega(n\log n)$ bound on the trade-off between space and query time  by giving  an oracle with $O(n(\log\log n)^2)$ space and $O((\log\log n)^3)$ query time\footnote{The dependency on $\epsilon$ of the oracles in~\cite{KKS11,WulffNilsen16,KST13} is  polynomial in $\frac{1}{\epsilon}$.}.   While the oracle of Wulff-Nilsen implies that the $O(n\log n)$ space-time tradeoff is not the best possible, it is suboptimal in both space and query time.  In this paper, for the first time, we provide an affirmative answer to~\ref{quest:motivating} in one of the four basic settings: approximate distance oracles for undirected planar graphs.

\begin{restatable}{theorem}{Main}
	\label{thm:main}
	Given an edge-weighted $n$-vertex planar undirected graph $G$ and a fixed parameter $\epsilon < 1$, there is a $(1+\epsilon)$-approximate distance oracle with $O(n)$ space and $O(1)$ query time. Furthermore, the oracle can be constructed in worst-case time $O(n\polylog(n))$ time. 
\end{restatable}

The precise dependencies of the space, query time, and construction time on $\eps$ are $O(n\epsilon^{-2})$, $O(\epsilon^{-2})$ and  $O(n\epsilon^{-3}\log^6(n))$, respectively. For the simplicity of the presentation, we do not try to optimize the dependency on $\eps$  as well as the logarithmic factor in the construction time.

Our result is optimal in the sense that any $(1+\epsilon)$-approximate distance oracle for $n$-vertex weighted undirected planar graphs must use $\Omega(n)$ space. This lower bound holds regardless of query time and holds even for $n$-vertex simple paths with unique integer edge weights in $\{2^0, 2^1,\ldots,2^{n-1}\}$ and with $\epsilon < 1$: the number of such paths is $n!$ and since $1 + \epsilon < 2$, a $(1+\epsilon)$-approximate distance oracle for any such path $P$ can be queried to derive the weight of each edge of $P$. Hence, any $(1+\epsilon)$-approximate oracle needs $\Omega(n\lg n)$ bits of space. Although an integer edge weight like $2^{n-1}$ cannot be stored in a $\Theta(\lg n)$ bit word using the standard binary representation of an integer, it can be stored in $1$ (or possibly $O(1)$) words when represented as a floating point since the exponent can be represented using $O(\lg n)$ bits. Also note that for unweighted undirected planar graphs, a $(1+\epsilon)$-approximate oracle with $o(n)$ space would not allow the user to freely assign unique labels to the vertices since there are $n!$ such assignments.

For planar digraphs\footnote{Digraphs is a shorthand for directed graphs.}, Thorup~\cite{Thorup04} designed the first $(1+\epsilon)$-approximate distance oracle with $O(n \log n \log(nN))$ space and $O(\log\log(nN))$ query time where edge weights are non-negative and $N$ is the ratio between the largest and smallest positive edge weight.\footnote{Thorup in fact assumes that edge weights are integers in $\{0,1,\ldots,N\}$ but this is only needed to get a small preprocessing time since it allows the use of fast priority queues. Since we do not focus on preprocessing for planar digraphs, we avoid the integer weight assumption.} (The precise dependency on $\eps$ of Thorup's oracle is $O(n \log n \log(nN)/\eps)$ for space and $O(\log\log(nN) + 1/\eps)$ for query time.) Since its introduction two decades ago, Thorup's oracle has not been improved. An added challenge in the directed setting is that \emph{portals} -- an important concept in approximate distance oracles -- in digraphs are less well-behaved  than their undirected counterparts. Specifically, for a given vertex $v$ and a shortest path $P$, the number of vertices on $P$, called portals, through which we need to re-route the shortest paths from $v$ to all  vertices on $P$ (with $1+\eps$ multiplicative error) is $O(\frac{1}{\eps})$ in the undirected case, while the number of such portals could be up to $\Theta(|P|)$ in the directed case.  Even for the special case of unweighted digraphs, the bounds of Thorup's oracle (by taking $N = 1$) has remained state-of-the-art: $O(n\log^2 n)$ space and $O(\log\log n)$ query time for constant $\eps$.

In this paper, we give the first improvement over the space bound Thorup's oracle by more than a logarithmic factor, while keeping the same query time for a constant $\eps$. For $x > 0$, define $\log^{(1)}x = \log x$ and for integer $k > 1$, $\log^{(k)}x = \log(\log^{(k-1)}x)$. We show the following.

\begin{restatable}{theorem}{Digraphs}
	\label{thm:Digraphs}
	Given a planar $n$-vertex digraph $G$ with non-negative edge weights and given $\epsilon > 0$ and any integer $k = \Theta(1)$, there is a $(1+\epsilon)$-approximate distance oracle for $G$ with space\linebreak $O(n\log(Nn)\log^{(k)}(n)/(\epsilon\log\log\log(Nn)))$ and query time $O(\log\log(Nn) + 1/\epsilon^{5.01})$ where $N$ is the ratio between the largest and smallest positive edge weight.
\end{restatable}
Note that for $k\geq 4$, the space bound is $o(n\log(Nn)/\epsilon)$ in Theorem~\ref{thm:Digraphs}.

By setting $N = 1$, we obtain the first approximate distance oracle for unweighted planar digraphs with $o(n\log^2 n)$ space and $O(\log\log n)$ query time (\Cref{cor:Digraphs} below). In fact, we get a space bound of only $o(n\log n)$. Our result might suggest that it is possible to construct an oracle with $O(n)$ space and $O(1)$ query time for unweighted planar digraphs. We believe that constructing such an oracle is an important step towards resolving \Cref{quest:motivating} for approximate distance oracles in edge-weighted planar digraphs.

\begin{corollary}\label{cor:Digraphs}
	Given an unweighted planar $n$-vertex digraph $G$ and given $\epsilon > 0$ and any integer $k = \Theta(1)$, there is a $(1+\epsilon)$-approximate distance oracle for $G$ with  $O(n\log(n)\log^{(k)}(n)/(\epsilon\log\log\log n))$ space and  $O(\log\log n + 1/\epsilon^{5.01})$ query time.
\end{corollary}

For unweighted planar digraphs, there is a space lower bound of $\Omega(n\lg n)$ bits, i.e., $\Omega(n)$ words with word size $\Theta(\lg n)$, and this lower bound holds even for a data structure that answers reachability queries regardless of query time~\cite{HolmRT15}. Thus the space bound in Corollary~\ref{cor:Digraphs} is only a factor of $o(\log(n)/\epsilon)$ away from optimal.

We have not focused on preprocessing time of our oracle for planar digraphs but it is easy to see from the description of this oracle that preprocessing time is bounded by a polynomial in $n$ times $\log N$.

\paragraph{Model of computation.~} The model of computation considered in our paper is the standard WORD RAM with word size $\mword = \Omega(\log n)$. In this model, arithmetic operations  ($+,-,*,/,\%$), comparisons ($ <,  >,  =, \le, \ge$), and bitwise operations (AND, OR, XOR, SHIFT)  on words take constant time each. 

\subsection{Techniques}\label{subsec:technique}

\paragraph{Approximate distance oracles for undirected planar graphs.~} Our technique for undirected planar graphs is inspired by that of Kawarabayashi, Sommer, and Thorup~\cite{KST13} who constructed a distance oracle with $\overline{O}(n\log n)$ space and  $\overline{O}(1/\epsilon)$ query time. Their idea is to construct an oracle with \emph{multiplicative} stretch $(1+\epsilon)$ from (a collection of) distance oracles with \emph{additive stretch} via a clever use of \emph{sparse covers}\footnote{See Section~\ref{sec:prelim} for a formal definition of a sparse cover.}.  Specifically, sparse covers are used to construct a set of subgraphs of $G$, called \emph{clusters}, and then a  distance oracle with additive stretch $\epsilon D$ is constructed for each cluster; here $D$ is the diameter of the cluster. However, the space bound of the oracle by Kawarabayashi, Sommer, and Thorup is superlinear for two reasons: (1) the additive distance oracle has space that is superlinear in the number of vertices of each cluster and (2)  the total size (the number of vertices) of all clusters is $\Omega(n\log n)$.  Nevertheless, the technique of Kawarabayashi, Sommer, and Thorup~\cite{KST13} suggests an interesting connection to geometry since sparse covers have a very natural geometric interpretation\footnote{A sparse cover of a Euclidean space is simply a tiling of the space by overlapping hypercubes.}.  

In doubling metrics, it was shown how to construct a distance oracle with $O(n)$ space and $O(1)$ query time~\cite{HM05,HM06,BGKLR11} for  constant values of $\epsilon$ and constant dimensions. Thus, it is natural to ask: Can we exploit techniques developed for doubling metrics to construct a $(1+\eps)$-approximate distance oracle for planar graphs with $O(n)$ space and $O(1)$ query time? To be able to answer this question positively, there are several technical barriers that we need to overcome. The most fundamental one is that space bounds of all known oracles for doubling spaces have an \emph{exponential dependency} on the dimension~\footnote{Interestingly, the query time can be made independent of $\epsilon$ and $d$~\cite{BGKLR11}}, while very simple planar graphs, such as star graphs, have doubling dimension $\Omega(\log n)$. Thus, it is somewhat counter-intuitive that we are able to exploit the geometric techniques in the construction of our apporixmate oracle for planar graphs.

We overcome all the technical barriers, as detailed below, to resolve two obstacles in the construction of Kawarabayashi, Sommer, and Thorup~\cite{KST13} by capitalizing on the techniques developed in the geometric context. Like the previous distance oracle constructions in doubling metrics~\cite{HM06,HM05,BGKLR11}, we start with a \emph{net tree} in which each level $i$ of the net tree is a $2^{i}$-net\footnote{An $r$-net of  a metric $(V,d_G)$ is a subset of points $N$ such that $d_G(x,y) > r$ for every $x\not=y \in N$ and for every $z \not\in N$, there exists $x\in N$ such that $d_G(x,z) \leq r$.} of the shortest path metric of $G(V,E)$. (There are $O(\log \Delta)$ levels where $\Delta$ is the spread\footnote{Spread of a metric is the ratio of the maximum pairwise distance to the minimum pairwise distance.} of the metric.)  In doubling metrics, for each level-$i$ of the net tree, we could store all distances from a net point $p$ to other net points in the same level within radius $O(\frac{2^i}{\epsilon})$ from $p$; there are only $O(\epsilon^{-d}) = O(1)$ such points when the dimension is a constant. However, in planar graphs, for each net point\footnote{We use points and vertices interchangeably.} $p$, there could be $\Omega(n)$ net points at the same level $i$ within radius  $O(\frac{2^i}{\epsilon})$ from $p$. We cannot afford to store all such distances. An important observation that we rely on in our construction is that we only need to (approximately) preserve distances from $p$ to net points  at level $i$ within in radius  $\Theta(\frac{2^i}{\epsilon})$ (rather than $O(\frac{2^i}{\epsilon})$) from $p$. Note that there could still be $\Omega(n)$ such net points.

Our first contribution is a technique to construct a \emph{$(d,\alpha,S)$-restricted } distance oracle with $O(|S|)$ space and $O(1)$ query time for every parameter $d$ and a \emph{constant} $\alpha$.  The oracle guarantees that,  for every pair $(u,v)\in S\times S$ such that $d_G(u,v)\in [d,\alpha d]$, the returned distance is within $[d_G(u,v), (1+\eps)d_G(u,v)]$. An important property of a  $(d,\alpha,S)$-restricted distance oracle, beside having $O(|S|)$ space, is that its space does not depend on $n$, the number of vertices of the graph. Essentially, we overcome the first obstacle in the construction of Kawarabayashi, Sommer, and Thorup~\cite{KST13}. We use $(d,\alpha,S)$-restricted  oracles in the following way: for each level $i$ of the net tree, we construct a $(d,\alpha,S)$-restricted distance oracle  with $d = \frac{2^i}{\epsilon}$, $\alpha = O(1)$ and $S = N_i$ where $N_i$ is the set of net points at level $i$. This oracle will guarantee multiplicative stretch $(1+\epsilon)$ for  every pair of net points at level $i$ whose distance from each other is $\Theta(\frac{2^i}{\epsilon})$ -- for other pairs, the distance error could be arbitrarily large. The idea behind this construction is that the space incurred \emph{per net point on average} is $O(1)$. We remark that the construction of restricted oracles heavily relies on planarity.

However, the net tree has its own problem: the number of vertices of $T$ could be $\Omega(n\log \Delta)$ where $\Delta$ could be exponential in $n$. A simple idea to reduce the number of vertices of $T$ is to \emph{compress} degree-$2$ vertices. But the compression of degree-2 vertices introduces two new problems. First, each point still ``participates'' in  the construction of  the restricted distance oracles of  up to $O(\log \Delta)$ different levels, and hence, the compression of degree-2 vertices does not really help. Second, compressing the net-tree makes it harder to \emph{navigate}. That is, given a leaf point $p$, we want to find an ancestor of $p$ at a given level $i$ in $O(1)$ time. In doubling metrics, to efficiently navigate the net-tree, previous constructions  heavily rely on the fact that each vertex of the net tree has $O(1)$ children, which is not the case in our setting.   Resolving both problems can be seen as overcoming the second obstacle in the construction of Kawarabayashi, Sommer, and Thorup~\cite{KST13}. 

We resolve the first problem by distinguishing two types of degree-2 vertces in $T$: those that are required in the construction of restricted oracles and those that are not. For the later type, it is safe to compress. For the former type, we are able to show that there are only $O(n \log(\frac{1}{\eps}))$ such degree-2 vertices; this linear bound is crucial to our construction. We resolve the second problem  -- navigating the net tree -- by designing a new \emph{weighted level ancestor} (WLA) data structure. The WLA problem, introduced by Farach and Muthukrishnan~\cite{FM96}, is a generalization of the predecessor problem~\cite{KL07} where various  \emph{super-constant} lower bounds in query time when the space usage is restricted to $O(n)$ have been established~\cite{Ajtai88,Miltersen94,MNSW98,BF02,SV08,PT06}. Here we need a data structure with linear space and constant query time. Our key insight is that for trees with \emph{polylogarithmic depth}, we can design such a data structure\footnote{Alstrup and Holm\cite{AH00} mentioned the construction of a data structure than can handle WLA in trees of polylogarithmic weights, which could be applicable in our work. However, the details were not given. On the other hand, our data structure instead exploits the polylogarithmic depth.}. Observe that when $\log(\Delta) = \polylog(n)$, the net tree has a polylogarithmic depth. Hence, for planar graphs with \emph{quasi-polynomial spread}, we can design a distance oracle with $O(n)$ space and $O(1)$ query time with all the ideas we have discussed. This turns out to be the hardest case: we adapt the contraction trick of Kawarabayashi, Sommer, and Thorup~\cite{KST13} and devise the bit packing technique to reduce the general problem to the case where the spread is quasi-polynomial.
 
The final tool we need is a data structure that allows us to quickly determine the level of the ancestors of $u$ and $v$ in the net tree for each query pair $(u,v)$. Once the ancestors and their level are found, we can perform a distance query from the restricted oracle at that level. We observe that the level of the ancestors can be approximated within a constant additive error if we can determine the distance between $u$ and $v$ within \emph{any constant factor}. To that end, we design an oracle with $O(1)$ multiplicative stretch, linear space, and constant query time. Our technique is simple and based on $r$-division, a basic tool to design algorithms for planar graphs, and our construction can be implemented in nearly linear time. Furthermore, our results apply to any graph in a hereditary class with sublinear separators.  A similar distance oracle for minor-free graphs can be obtained from the tree cover with $O(1)$ trees and $O(1)$ distortion by Bartal, Fandina, and Neiman~\cite{BFN19B}. However, it is unclear that the tree cover can be constructed in nearly linear time.
   
To obtain a nearly linear time preprocessing time, the major obstacle in our construction is to compute the net tree. Indeed, to the best of our knowledge, it is not known how to  construct an $r$-net of the shortest path metric of planar graphs in \emph{sub-quadratic time}\footnote{In doubling metrics, it is possible to construct a net tree in nearly linear time~\cite{HM05,HM06}. The construction relies heavily on the fact that the metric has a constant doubling dimension.}. Instead, we show that a weaker version of $r$-nets (see \Cref{sec:prelim} for a precise definition) can be computed in $O(n)$ time, and that we can use weak $r$-nets in place of $r$-nets in the net tree. A corollary of our weak net construction is a linear time algorithm to find a sparse cover, which improves upon the $O(n\log n)$ time algorithm of Busch, LaFortune, and Tirthapura~\cite{BLT07}.

\paragraph{Approximate distance oracles for planar digraphs.~} 
The techniques for our $(1+\epsilon)$-approximate distance oracle for planar digraphs build to some extent on those of Thorup~\cite{Thorup04}. Recall that $N$ is the ratio between the largest and smallest edge weight. After normalizing, the shortest path distances can thus be partitioned into $O(\log(Nn))$ distance scales of the form $[\alpha,2\alpha)$ where $\alpha$ is a power of $2$. For each distance scale, Thorup uses an oracle with an additive error of at most $\epsilon\alpha$. A query is answered by applying binary search on the distance scales, querying one of the oracles in each step, and the final oracle queried then gives the desired multiplicative $(1+\epsilon)$-approximation. Thus, $O(\log\log(Nn))$ queries to oracles are needed. Thorup shows how each oracle can answer a query in constant time (for fixed $\epsilon$).

Our first idea for improving space is to not have every oracle answer a query in constant time. In fact, it suffices for every $\Theta(\log\log(Nn))$ distance scale to have an oracle with $O(1)$ query time since binary search on these oracles brings us down to only $\Theta(\log\log(Nn))$ distance scales; since the total query time should be $O(\log\log(Nn))$, we can thus afford slower but more space-efficient oracles for the remaining $O(\log\log\log(Nn))$ steps of the binary search. Hence, only a $1/\Theta(\log\log(Nn))$ fraction of our oracles have $O(1)$ query time.

To improve space further, we make use of a recursive decomposition of the input digraph $G$ into more and more refined $r$-divisions. Now, instead of storing portals for each vertex of each oracle (as in Thorup's paper), we instead store portals only for boundary vertices of pieces of $r$-divisions at each level of the recursive decomposition. Furthermore, each such boundary vertex stores only local portals belonging to the same parent piece, allowing us to use labels of length much smaller than $\lg n$ for each such local portal.

A query for a vertex pair $(u,v)$ is then answered by starting at the lowest level of the recursive decomposition and obtaining local portals for $u$ and $v$. For each combination of a local portal $p(u)$ of $u$ and local portal $p(v)$ of $v$, a query for an approximate distance from $p(u)$ to $p(v)$ is then answered recursively by going one level up in the recursive decomposition. To avoid an exponential explosion in the number of local portals during the recursion, the recursive decomposition is set to have only $k = \Theta(1)$ height where $k$ is the parameter in Theorem~\ref{thm:Digraphs}.

Since our final space bound is $o(n\log(Nn))$, we are only allowed space sublinear in $n$ on each distance scale. This creates some additional obstacles not addressed by the above techniques. A main obstacle is to answer $\lcads$-queries in $O(1)$ time using $o(n)$ space. There are static tree data structures that can do this using only $O(n/\lg n)$ space (i.e., $O(n)$ bits of space). Unfortunately, these structures require labels of query vertices to be of a special form; for instance, the data structure in~\cite{JanssonSS07} that we rely on requires them to be preorder numbers in the tree. We give a new data structure that can convert in $O(1)$ time vertex labels in the input graph $G$ to preorder numbers in such a tree using $o(n)$ space plus $O(n)$ additional space independent of the current distance scale. Here, we again make use of our recursive decomposition and show how it allows for a compact representation of preorder numbers for each tree.

\subsection{Related Work} \label{subsec:related-work}

\paragraph{Optimizing the dependency on $\eps$ for planar undirected graphs.~} A closely related and somewhat orthogonal line of work initiated by Kawarabayashi, Sommer and Thorup~\cite{KST13} is to treat $\epsilon$ as a part of the input and optimize for the dependency on $\epsilon$ in the trade-off between space and query time. They showed that it is possible to  achieve a nearly linear dependency on $\epsilon$ in the space and query time trade-off. Specifically, they constructed an oracle of $\overline{O}(n\log n)$ space and  $\overline{O}(1/\epsilon)$ query time where $\overline{O}(.)$ notation hides $\poly(\log \log(n)) $ and $\polylog(\frac{1}{\epsilon})$ factors. The dependency of space and query time product  on $\epsilon$ in previous work~\cite{Klein02,Thorup04,KKS11} is at least quadratic. Other recent developments~\cite{GX19,CS19} focus on improving the dependency on $\epsilon$ in the query time: Gu and Xu~\cite{GX19} constructed a distance oracle with $O(1)$ query time and $O(n\log n (\log n/\epsilon + 2^{O(\frac{1}{\epsilon})}))$ space; Chan and Skerpetos constructed an oracle with $O(\log \frac{1}{\epsilon})$ query time and $O(n\poly(\frac{1}{\epsilon})\polylog(n))$ space.  

\paragraph{Exact distance oracles for directed planar graphs.~} Arikati et al. \cite{ACCDSZ96} constructed the first distance oracle for directed planar graphs with $O(S)$ space and $O(\frac{n^2}{S})$ query time for $S \in [n^{3/2},n^2]$. Independently from the work of Arikati et al. \cite{ACCDSZ96},  Djidiev~\cite{Djidjev96} constructed two oracles with different space-query time trade-offs: (1) an oracle with $O(S)$ space and $O(\frac{n^2}{S})$ query time for $S \in [n,n^2]$ and (2) an oracle with $O(S)$ space and $\tilde{O}(\frac{n}{\sqrt{S}})$ query time for $S \in [n^{4/3}, n^{3/2}]$. ($\tilde{O}$ notation hides a $\polylog(n)$ factor.)  Subsequent works aimed to widen the range of $S$ in the space-query time tradeoff in the second oracle of Djidiev~\cite{Djidjev96}. Specifically, Chen and Xu~\cite{CX00} pushed the range of $S$ to $[n^{4/3}, n^{2}]$; Fakcharoenphol and Rao~\cite{FR06} constructed an oracle with $S = n\log n$; Mozes and Sommer~\cite{MS12} widened the range of $S$ to $[n\log\log n, n^{2}]$. 

The work of Cabello~\cite{Cabello10} focused on improve the preprocessing time; specifically,  Cabello~\cite{Cabello10} constructed an oracle with $O(S)$ space, $\tilde{O}(\frac{n}{\sqrt{S}})$ query time and $O(S)$ construction time for any $S\in [n^{4/3}\log^{1/3}(n), n^2]$.  Wulff-Nilsen~\cite{Wulff-Nilsen10} designed an oracle with constant query time and $o(n^2)$ space; this space bound has not been improved for oracles with constant query time. The first \emph{linear space} oracle with $O(n^{\frac{1}{2}+\eps})$ query time for any constant $\epsilon > 0$ was obtained independently by Mozes and Sommer~\cite{MS12} and Nussbaum~\cite{Nussbaum11}; the result remains the state-of-the-art if we insist on having an oracle with linear space.

All exact distance oracles mentioned so far were constructed based on (variants) of \emph{$r$-division}, a technique introduced by Frederickson~\cite{Frederickson87}. None of them achieves \emph{truly subquaratic} space and polylogarithmic query time. In a breakthrough work, Cohen-Addad, Dahlgaard, and Wulff-Nilsen~\cite{CDW17} broke this barrier by constructing an  exact distance oracle with  $O(n^{\frac{5}{3}})$ space and $O(\log n)$ query time. Indeed, they obtained a more general trade-off: $O(S)$ space and $\tilde{O}(\frac{n^{5/3}}{S^{3/2}})$ query time. Their construction is based on planar Voronoi diagrams introduced by Cabello~\cite{Cabello18}.  Gawrychowski et al.~\cite{GMWW18} improved the result Cohen-Addad, Dahlgaard, and Wulff-Nilsen~\cite{CDW17} to obtain an oracle with $O(S)$ space and $\tilde{O}(\max\{1, \frac{n^{3/2}}{S}\})$ query time. Recently, Charalampopoulos et al.~\cite{CGMW19} obtained three exact distance oracles with \emph{almost optimal} space-query time trade-offs (ignoring low order terms): (1) $\tilde{O}(n^{1+\eps})$ space and $\tilde{O}(1)$ query time for any constant $\epsilon > 0$, (2) $\tilde{O}(n)$ space and $O(n^{\epsilon})$ query time, and (3) $n^{1 + o(1)}$ space and $n^{o(1)}$ query time. Their trade-offs were furthered improved by Long and Pettie~\cite{LP20} in two regimes: $n^{1 +o(1)}$ space with $\tilde{O}(1)$ query time and $\tilde{O}(n)$ space with $n^{o(1)}$ query time.

\paragraph{Distance oracles for low dimensional metrics} The idea of using net-tree in the construction of $(1+\epsilon)$-approximate distance oracles, which we also use in this paper, was introduced by Har-Peled and Mendel~\cite{HM06} in metrics of constant doubling dimension. Their oracle  has $O(n\log n)$ space, $O(1)$ query time, and $O(n\log n)$ construction time with $\epsilon$ and $d$ being fixed constants. Their result improved an earlier result by  Gudmundsson et al.~\cite{GLNS08} who constructed  a $(1+\epsilon)$-approximate distance oracle for $t-$spanners of point sets in $\mathbb{R}^d$. In the same paper, Har-Peled and Mendel~\cite{HM06} presented a $(1+\epsilon)$-approximate distance oracle for doubling metrics of dimension $d$ with $\epsilon^{-O(d)}n$ space, $O(d)$ query time and $\poly(n)$ construction time. That is, the query time depends linearly, instead of exponentially, on the dimension. Bartal et al.~\cite{BGKLR11} improved the result of Hard-Peled and Mendel by designing  an oracle with $(\epsilon^{-O(d)} + 2^{O(d\log d)})n$ space, $O(1)$  query time and nearly linear \emph{ expected} construction time.

\section{Proof Overview}\label{sec:high-level-ideas}
In this section, we give an overview of the proof of our main result for planar undirected graphs. We omit a proof overview for the part of our paper on planar digraphs since it is self-contained and short compared to the undirected result.

\subsection{An Approximate Distance Oracle for Undirected Planar Graphs}

We refer readers to~\Cref{subsec:technique} for a bird's-eye view of our construction.  We construct our oracles by viewing graph $G(V,E)$ as a metric space $(V,d_{G})$ induced by shortest path distances. There are two major steps in our construction. In Step 1, we construct  an approximate oracle with linear space, constant query time, and \emph{constant stretch}.  We use this oracle later as a tool to navigate the net tree.  In Step 2, we construct a \emph{$(d,\alpha, S)$-restricted} distance oracle with $O(|S|)$ space and $O(1)$ query time for any parameter $d$ and any constant $\alpha$. The oracle returns a distance estimate in $[d_G(u,v), (1+\epsilon)d_G(u,v)]$ for any query pair $(u,v)\in S\times S$ such that $d_{G}(u,v) \in [d,\alpha d]$; the oracle may make an arbitrarily big error on the distance estimate if  $d_{G}(u,v) \not\in [d,\alpha d]$.

\paragraph{Step 1: Constant stretch approximate oracle.~} For a planar graph $H$ and a parameter $r$, an $r$-division of $H$ is a partition of the edges of $H$ into $O(|H|/r)$ subgraphs, each of size $O(r)$ and with $O(\sqrt r)$ boundary vertices shared with other pieces.

For our constant stretch approximate oracle, we use a hierarchical decomposition $\mathcal{T}$ with 3 levels, where each node of $\mathcal{T}$ is associated with a subgraph, called a \emph{piece}, of $G$. The root of $\mathcal{T}$ is $G$, and children of an internal node $P$ at level $i$, $i \in \{0,1\}$, are pieces of an $(r_{i+1})$-division of $P$. Here $r_0 = n, r_{1} = \Theta((\log n)^2)$, and $r_2 = \Theta((\log r_1)^2) = \Theta((\log \log n)^2)$.

Let $P$ be a non-leaf piece. Let $B_P$ be the set of all boundary vertices of \emph{child pieces} of $P$. For each vertex $u\in P$,  we store the distance from $u$ to a nearest vertex in $B_P$. In addition, we store with $P$ an oracle that allows us to query $(1+\eps)$-approximate distances between pairs of vertices in $B_P$; we call this oracle a \emph{$B_P$-restricted oracle}. For a query pair $(u,v)$ whose shortest path is fully contained in $P$ but not in a child piece, the weight of this path can now be approximated up to a constant factor by summing up (1) the distances from $u$ and from $v$ to their respective nearest boundary vertices in $B_P$, and (2) the approximate distance between these two boundary vertices obtained by querying the $B_P$-restricted oracle. We show that the restricted oracle requires only $O(|P|)$ space.

If $P$ is a leaf piece, we show how to encode distances between vertices inside $P$ into a compact lookup table with $O(|P|)$ words. A key insight is that, to encode a constant approximation of the shortest path distance between two vertices $u$ and $v$ in $P$, it suffices to store \emph{the encoding} of the pair $(e_{uv},\rho_{uv})$ where  $e_{uv}$ is the heaviest edge on a shortest path between $u$ and $v$ and $\rho_{u,v} = \lceil \frac{d_P(u,v)}{w(e_{uv})}\rceil$. We show that for each vertex $u \in P$, one can pack $\{(e_{uv}, \rho(u,v))\}_{v\in P}$ into a single word, and that the approximate distance (up to a factor of $2$) can be retrieved in constant time. Our result is summarized in~\Cref{thm:constant-stretch} whose proof is given in Section~\ref{sec:const-stretch}.

\begin{restatable}{theorem}{ConstantStretch}
	\label{thm:constant-stretch}
	Given an edge-weighted $n$-vertex planar graph $G$, there is an $8$-approximate distance oracle with space $O(n)$ and query time $O(1)$ that can be constructed in worst-case time $O(n\log^3n)$.
\end{restatable}

We remark that \Cref{thm:constant-stretch} applies to any graph in a hereditary class with sublinear separators, such as minor-free graphs and graphs with polynomial expansion~\cite{DN16}, assuming that the separator can be found in nearly linear time.  

\paragraph{Step 2: An oracle with $(1+\epsilon)$ stretch.~} 
In this step,  we treat $G$ as a metric space $(V,d_G)$.  We construct a hierarchy of nets $V = N_0 \supseteq N_1 \supseteq \ldots, N_{\lceil \log(\Delta) \rceil}$ where $N_i$ is a $2^i$-net of  $(V,d_G)$ and $\Delta$ is the spread of $(V,d_G)$. (In reality, we use weak nets instead of nets for fast preprocessing time; for the simplicity of the presentation in this section, we describe the construction in terms of nets.) This hierarchy induces a net-tree $T$. For each $2^{i}$-net $N_i$, we identify a subset $N'_{i}\subseteq N_i$ and construct a $(2^{i}, O(\frac{1}{\epsilon}), N'_i)$-restricted distance oracle with $O(|N'_i|)$ space and $O(1)$ query time, as guaranteed by~\Cref{thm:restricted-oracle} below.

\begin{restatable}{theorem}{RestrctedOracle}
	\label{thm:restricted-oracle}
	Given parameters $d,\epsilon > 0, \alpha > 1$ and a subset of vertices $S$ of an edge-weighted $n$-vertex planar graph $G$, there exists an approximate distance oracle $\mathcal{D}$ with $O(|S| \alpha^2 \epsilon^{-2})$ space and $O(\alpha^2\epsilon^{-2})$ query time such that the distance returned by $\mathcal{D}$ for a given query pair $(u,v)\in S\times S$, denoted by $d_{\mathcal{D}}(u,v)$, is always at least $d_G(u,v)$, and:
	\begin{equation*}
		\qquad d_{\mathcal{D}}(u,v)\leq (1+\eps)d_G(u,v) \mbox{ if } d_G(u,v) \in [d, \alpha d].
	\end{equation*}
	Furthermore, $\md$ can be constructed in time $O(\epsilon^{-3}\alpha^3 n\log^3 n)$. 
\end{restatable}

Our construction of the oracle in \Cref{thm:restricted-oracle} uses shortest path separators (see \Cref{sec:prelim} for a formal definition), whose existence follows from planarity.  This is the only place where we need planarity, aside from \Cref{thm:constant-stretch} that only requires sublinear separator.

Next, we show that our chosen family of sets $\mathcal{S} = \{N_i^{'}\}_{i=1}^{\lceil \log(\Delta) \rceil}$ has linear size: $\sum_{S \in \mathcal{S}} |S| = O(n \log \frac{1}{\epsilon}) = O(n)$ for a constant $\eps$. Thus, the total size of all the oracles restricted to sets in $\mathcal{S}$ is $O(n)$.  Note that  $\sum_{i=1}^{\lceil \log (\Delta)\rceil} |N_i|$ could be $\Theta(n\log \Delta)$.

To query for the distance between $u$ and $v$, we find two points in $N_i$, say $x$ and $y$ that cover $u$ and $v$, respectively, where $i \approx \log_2(\epsilon d_G(u,v))$. That is, $d(u,x), d_G(u,y) \leq 2^i \approx \epsilon d_G(u,v)$, and thus $d_G(x,y) \approx (1\pm O(\epsilon))d_G(u,v)$. Querying the estimated distance between $x$ and $y$ gives us an estimate of $d_G(u.v)$ within a factor of $(1\pm O(\epsilon))$. There are two  problems that we need to address: (a) finding the level $i$ and (b) finding the ancestor of a leaf at a given level in the net tree $T$ in $O(1)$ time.

To resolve problem (a), we use the constant stretch oracle to estimate the distance $d_G(u,v)$ within a constant factor, denoted by $\tilde{d}_G(u,v)$. Then the ancestor $\tilde{x}$ ($\tilde{y}$) of $u$ ($v$) at level  $\tilde{i} = \log(\epsilon \tilde{d}_G(u,v))$ is within $O(1)$ hops of $x$ ($y$) in the net-tree $T$. 

Problem (b) is a special case of the  level ancestor problem on $T$: given any (leaf or non-leaf) node $u \in T$, find the ancestor of $u$ at level $i$. Many different data structures have been invented to solve this problem with $O(|T|)$ space and $O(1)$ query time~\cite{HT84,BV94,Dietz91,AH00,BF04}. Unfortunately, the size of $T$ could be as big as $\Omega(n\log \Delta)$. We get around this by working with a compressed version $T_{\im}$ of $T$ that only has $O(n)$ vertices: $T_{\im}$ is obtained from $T$ by contracting \emph{a subset of degree-2} vertices to their neighbors. We cannot contract all degree-2 nodes because some of them will participate in the construction of the restricted oracle at some level of the net tree as discussed in \Cref{subsec:technique}. The level ancestor problem in $T$ is now reduced to the \emph{weighted level ancestor} (WLA) problem in $T_{\cpr}$: given a node $x\in T_{\cpr}$ that has integer weights on the edges and  an integer $\ell$, find the lowest ancestor of $x$ whose depth in $T_{\cpr}$ is at most $\ell$.  Here the depth of a vertex is the distance from that vertex to the root of $T_{\cpr}$. The WLA problem is a generalization of the predecessor search problem  where various  \emph{superconstant} lower bounds for query time have been established~\cite{Ajtai88,Miltersen94,MNSW98,BF02,SV08,PT06} even when the space is superlinear; see~\cite{PT06} for a thorough discussion on the lower bounds.  Here we need a data structure with linear space and constant query time. Our key insight is that  when $\log(\Delta) = \polylog(n)$, we can design such a data structure. Thus,  we obtain a $(1+\epsilon)$-approximate distance oracle with $O(n)$ space and $O(1)$ query time for planar graphs with \emph{quasi-polynomial spread}.

By adapting the contraction trick of Kawarabayashi, Sommer, and Thorup~\cite{KST13} and devising a bit-packing technique, we can reduce the general problem to the case when $\log(\Delta) = O(\polylog(n))$; this completes our data structure.

\paragraph{Fast Preprocessing Time} To obtain a fast construction time, the major obstacle is to construct a net in nearly linear time. Unfortunately, such an algorithm is not known. Here we show how to construct a weaker version of nets. 

Given an edge-weighted graph $G(V,E)$ and a set of terminals $K\subseteq V$. A \emph{weak  $(r,\gamma)$-net} of $K$, for $\gamma \geq 1$, is a subset of vertices $N\subseteq K$ such that: (a) $d_G(p,q)\geq r$ for every $p\not= q \in N$ and (b) for every $x\in K$, there exists a $p \in N$ such that $d_G(p,x)\leq \gamma r$.  An assignment $\mathcal{A}$ associated with a weak $(r,\gamma)$-net $N$ is a family of subsets of $K$ such that for each $x \in N$, there exists a set, denoted by $\mathcal{A}[x] \subseteq K$, in $\mathcal{A}$ that contains $x$ and satisfies $d_G(x,y)\leq \gamma r$ for every $y \in \mathcal{A}[x]$. We say that an assignment $\mathcal{A}$ \emph{covers} $K$ if $\cup_{A\in \mathcal{A}} = K$.

A $(\beta,s,\Delta)$-sparse cover of a graph $G$ is a decomposition of $G$ into subgraphs $\mathcal{C} = \{C_1,\ldots, C_k\}$, called clusters, such that: (a) each subgraph $C\in \mathcal{C}$ has diameter at most $\Delta$, (b) for every $v\in G$, $B_G(v,\Delta/\beta)\subseteq C$ for some $C\in \mathcal{C}$ and (c) each vertex $v\in V(G)$ belongs to at most $s$ clusters.

\begin{theorem}\label{thm:weak-net-const}
	Given an edge-weighted $n$-vertex planar graph $G(V,E)$, a subset of terminals $K\subseteq V$, and a distance parameter $r$, there is an algorithm that can construct in $O(n)$ time an $(r,O(1))$-net $N$ and an associated assignment $\mathcal{A}$ of $N$ that covers $K$. 
\end{theorem}

A corollary of our construction of weak nets Section~\ref{section:weak-net} is a construction of an  $(O(1),O(1),\Delta)$-sparse cover for planar graphs in \emph{linear time}.

\begin{lemma}\label{lm:sparsecover-time} Given a parameter $\Delta  > 0$ and an edge-weighted $n$-vertex planar graph $G(V,E)$, there is an algorithm that can construct in $O(n)$ time an $(O(1),O(1),\Delta)$-sparse cover  for $G(V,E)$.
\end{lemma}

\section{Preliminaries}\label{sec:prelim}

Let $G$ be a graph. We denote by $V(G)$ the vertex set of $G$ and by $E(G)$ the edge set of $G$. We sometimes write $G(V,E)$ to explicitly indicate that $V(G) = V$ and $E(G) = E$, and write $G(V,E,w)$ to indicate that $w$ is the weight function on the edges of $G$. We denote by $\SP_G(x,y)$ a  shortest path between two vertices $x,y\in V$.

In this paper, we sometimes view $G$ as a metric $(V,d_G)$ with the shortest path distance. The \emph{spread}, denoted by $\Delta$, of $G$ is defined to be the spread of $(V,d_G)$, which is:
\begin{equation}
	\Delta = \frac{\max_{u,v}d_G(u,v)}{\min_{u\not=v} d_G(u,v)}
\end{equation}

Let $C$ be a simple closed curve on the plane $\mathbb{R}^2$. Removing $C$ from $\mathbb{R}^2$ divides the plane into two parts, called the \emph{interior} and \emph{exterior} of $C$, denoted by $\int(C)$ and $\ext(C)$, respectively. 

Let $T$ be a tree and $x,y$ be two vertices of $T$. We  denote by $T[x,y]$ the (unique) path between $x$ and $y$ in $T$.

\paragraph{Shortest path separators of planar graphs.~} Let $G$ be given as a planar embedded graph and $G_{\Delta}$ be a triangulation of $G$. We call edges in $E(G_{\Delta})\setminus E(G)$ \emph{pseudo-edges}. Let $T$ be a shortest path tree of $G$; $T$ is also a spanning tree of $G_{\Delta}$.  A path $P$ of $T$ is \emph{monotone} if one endpoint of $P$ is  an ancestor of all vertices in $P$; this endpoint is called the  \emph{root} of $P$. 

A \emph{shortest path separator} $C$ of $G$ is a fundamental cycle of $G_{\Delta}$ w.r.t $T$. Since edges of $G_{\Delta}$ may not be in $G$, $C$ consists of two monotone paths $P_1,P_2$ of $G$ rooted at the same endpoint and a  (possibly imaginary) edge $(u,v)$ between two other endpoints of $P_1$ and $P_2$.  Thorup~\cite{Thorup04} and Klein~\cite{Klein02} observed that the following lemma is implicit the proof of the planar separator theorem by Lipton and Tarjan~\cite{LT79}.

\begin{lemma}[Lipton and Tarjan~\cite{LT79}]\label{thm:shortest-path-sep} Let $T$ be a shortest path tree rooted at a vertex $r$ of an edge-weighted planar graph $G$. Let $\omega: V\rightarrow \mathbb{R}^+$ be a weight function on vertices of $G$, and $W = \sum_{u\in V(G)}\omega(u)$. There is a shortest path separator $C$ of $G$ such that $\max(\omega(V(G)\cap \int(C)),  \omega(V(G)\cap \ext(C))) \leq \frac{2W}{3}$. Furthermore, $C$ can be found in $O(n)$ time.
\end{lemma}

\paragraph{$r$-division} Given an integer $r \geq 1$ and a planar graph $G$, an \emph{$r$-division} of $G$ with $n$ vertices is a partition of the edge set of $G$ into subsets inducing subgraphs $\{R_1,R_2,\ldots,R_k\}$ of $G$, called \emph{pieces}, such that:
\begin{enumerate}
	\item $k = O(\frac{n}{r})$ and $|V(R_i)| \leq r$ for all $i\in [1,k]$,
	\item $|\partial R_i| = O(\sqrt{r})$ where $\partial R_i$ is the set of vertices in $R_i$, called \emph{boundary vertices}, such that each has at least one neighbor outside $R_i$. 
\end{enumerate}

Frederickson~\cite{Frederickson87} introduced the notion of \emph{$r$-division} and devised an algorithm to compute an $r$-division for any given $r$ in time $O(n \log r + \frac{n}{\sqrt{r}}\log n)$. Klein, Mozes, and Sommer~\cite{KMS13} recently improved the running time of finding an $r$-division to linear.

\paragraph{Approximate labeling schemes.~} Thorup~\cite{Thorup04} and Klein~\cite{Klein02} independently came up with similar constructions of a $(1+\eps)$-approximate distance oracle for an edge-weighted, undirected planar graph $G(V,E)$ of $n$ vertices with $O(n\log n\epsilon^{-1})$ space and $O(\epsilon^{-1})$ query time. Thorup~\cite{Thorup04} showed that the oracle can be constructed in nearly linear time. Furthermore, Thorup~\cite{Thorup04} observed that the distance oracle could be distributed as a \emph{labeling scheme}: each vertex $u$ is assigned a label $\ell(u)$ of $O(\epsilon^{-1}\log n)$ words and there is a decoding function $\mathcal{D}$ that, given two vertices $u$ and $v$, returns a $(1+\epsilon)$-approximate distance between $u$ and $v$ in $O(\epsilon^{-1})$ time by looking at their labels only.

\begin{theorem}[Theorem 3.19~\cite{Thorup04}, Lemma 4.1~\cite{Klein02}]\label{thm:Thorup-labeling}
	Given an edge-weighted undirected planar graph $G$, we can construct in $O(n\epsilon^{-2}\log^3n )$ time a labeling scheme for $(1+\epsilon)$-approximate distances with maximum label size $O(\log n \epsilon^{-1})$ and decoding time $O(\epsilon^{-1})$.
\end{theorem}

\paragraph{Sparse cover.~}  A $(\beta,s,\Delta)$-sparse cover of a graph $G(V,E)$, denoted by $\mathcal{C} = \{C_1,\ldots, C_k\}$, is a collection of subgraphs, called \emph{clusters}, of $G$ such that:
\begin{itemize}[noitemsep]
	\item[(1)] The diameter of $C_i$ is at most $\Delta$ for any $i \in [1,k]$.
	\item[(2)] For every $v\in G$, $B_G(v,\Delta/\beta)\subseteq C_i$ for some $i\in [1,k]$.
	\item [(3)] Every vertex $v\in V$ is contained in at most $s$ clusters.
\end{itemize}
We say that $G$ admits a \emph{$(\beta,s)$-sparse covering scheme} if there exists a $(\beta, s,\Delta)$-sparse cover of $G$ for any given $\Delta \in \mathbb{R}^+$. $\Delta$ is called the diameter of the sparse cover. 

Busch, LaFortune, and Tirthapura~\cite{BLT07} constructed an $(O(1),O(1))$-sparse covering scheme for planar graphs; they slightly improved the constants in the journal version~\cite{BLT13}.  Abraham, Gavoille, Malkhi, and Wieder~\cite{AGMW10} constructed an  $(O(r^2),2^r(r+1)!)$-sparse covering scheme for $K_r$-minor-free graphs.

\begin{theorem}[Theorem 5.2~\cite{BLT07}]\label{thm:sparse-cover}Edge-weighted planar graphs admits an $(O(1), O(1))$-sparse covering scheme.
\end{theorem}

\noindent Our~\Cref{lm:sparsecover-time}  shows how to construct a sparse cover in linear time given a diameter parameter $\Delta$.

\paragraph{Weak nets.~} Given an edge-weighted graph $G(V,E)$ and a set of terminals $K\subseteq V$. A \emph{weak  $(r,\gamma)$-net} of $K$, for $\gamma \geq 1$, is a subset of vertices $N\subseteq K$ such that: (a) $d_G(p,q)\geq r$ for every $p\not= q \in N$ and (b) for every $x\in K$, there exists a $p \in N$ such that $d_G(p,x)\leq \gamma r$.   An assignment $\mathcal{A}$ associated with a weak $(r,\gamma)$-net $N$ is a family of subsets of $K$ such that for each $x \in N$, there exists a set, denoted by $\mathcal{A}[x] \subseteq K$, in $\mathcal{A}$ that contains $x$ and satisfies $d_G(x,y)\leq \gamma r$ for every $y \in \mathcal{A}[x]$. We say that an assignment $\mathcal{A}$ \emph{covers} $K$ if $\cup_{A\in \mathcal{A}} = K$.

\paragraph{Lowest common ancestor (LCA).~} Given a tree $T$, an LCA data structure for $T$, denoted by $\lcads_T$, answers the following type of queries: given two vertices $u$ and $v$ in $T$, return the LCA of $u$ and $v$ in the tree $T$, denoted by $\lcads_T(u,v)$. It is well known that we can construct in linear time a data structure $\lcads_T$ with $O(|V(T)|)$ space and $O(1)$ query time~\cite{HT84}; see also \cite{BF00} for a simpler construction.

\section{An Approximate Distance Oracle with Constant Stretch}\label{sec:const-stretch}

In this section, we prove~\Cref{thm:constant-stretch} that we restate below.

\ConstantStretch*


Thorup~\cite{Thorup04} and Klein~\cite{Klein02} independently gave a $(1+\epsilon)$-approximate distance oracle for undirected planar graphs with $O(\frac 1 \epsilon n\log n)$ space and $O(\frac 1 \epsilon)$ query time, and Thorup showed how to obtain worst-case construction time $O(n\log^3n/\epsilon^2)$. We will make use of the special case where $\epsilon = 1$ and we exploit that the distance oracle of Thorup is a labeling scheme, meaning that a query for vertex pair $(u,v)$ can be answered using only the $O(\log n)$ words associated with $u$ and $v$, respectively (\Cref{thm:Thorup-labeling}):
\begin{Lem}[Thourup~\cite{Thorup04}]\label{Lem:ThorupKlein}
Given an undirected planar graph $G$ with $n$ vertices and given a subset $S$ of the vertices of $G$, there is a $2$-approximate distance oracle for $G$ with $O(|S|\log n)$ space and $O(1)$ query time which can answer queries for any vertex pair in $S\times S$. The oracle can be constructed in $O(n\log^3 n)$ worst-case time.
\end{Lem}
In the following, we shall refer to the data structure of~\Cref{Lem:ThorupKlein} as an \emph{$S$-restricted oracle} of $G$.
\begin{wrapfigure}{r}{0.5\textwidth}
	\vspace{-15pt}
	\begin{center}
		\includegraphics[width=0.4\textwidth]{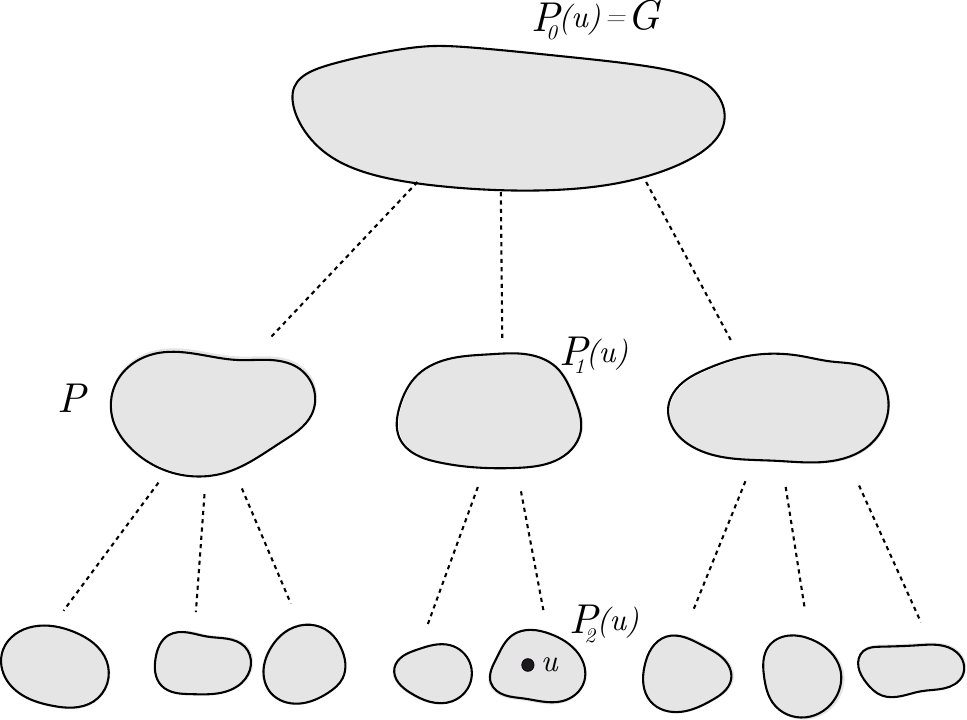}
	\end{center}
	\caption{\footnotesize{A 3-level recursive decomposition $\mathcal{T}$ of $G$.}}
	\vspace{-5pt}
	\label{fig:div-tree-const}
\end{wrapfigure}
We now present our approximate distance oracle for planar graph $G$, ignoring its efficient construction as well as some space-saving tricks for later.

The oracle keeps a $3$-level recursive decomposition of $G$ (see Figure~\ref{fig:div-tree-const}). This decomposition has an associated tree $\mathcal T$ where at level $0$, the root is $G$ having $r_0 = n$ vertices. Letting $r_1 = (\log r_0)^2 = (\log n)^2$, the children of $G$ in $\mathcal T$ are the pieces of an $r_1$-division of $G$ and these are the level $1$-nodes of $\mathcal T$. Finally, each piece $P$ at level $1$ of $\mathcal T$ has as children the pieces of an $r_2$-division of $P$ where $r_2 = (\log r_1)^2 = O((\log\log n)^2)$.


Each vertex $u\in V$ is associated with a leaf piece $P_2(u)$ containing $u$ as well as the two ancestor pieces $P_0(u)$ and $P_1(u)$ of $P_2(u)$ at levels $0$ and $1$, respectively.

For each non-leaf piece $P$ of $\mathcal T$, let $i\in\{0,1\}$ be its level. Associated with $P$ is a $B_P$-restricted oracle of $P$ where $B_P$ is the set of boundary vertices in the $r_{i+1}$-division of $P$.

For $i = 1,2$, each vertex $u$ is associated with a nearest boundary vertex $b_i(u)$ of $P_i(u)$. We also associate the distance $d_{P_i(u)}(u,b_i(u))$ with $u$.


For each leaf piece $P$ of $\mathcal T$, we essentially store a lookup table containing $2$-approximations of distances $d_P(u,v)$ for each pair of vertices $u$ and $v$ in $P$. However, we need some space-saving tricks to ensure that these tables require only linear space in total; for now, we delay the details (see \Cref{subsec:compact-table}) on how to do this and just assume black box lookup tables with constant query time.

\subsection{Answering a query}
A query for a vertex pair $(u,v)$ is answered as follows. For $i = 0,1$, the query algorithm first computes
\[
  d_i = \left\{\begin{array}{ll}d_{P_{i+1}(u)}(u,b_{i+1}(u)) + \tilde d_{P_i(u)}(b_{i+1}(u),b_{i+1}(v)) + d_{P_{i+1}(v)}(v,b_{i+1}(v)) & \mbox{if } P_i(u) = P_i(v)\\
              \infty & \mbox{otherwise,}
\end{array}\right.
\]
where $\tilde d_{P_i(u)}(b_{i+1}(u),b_{i+1}(v))$ is the output of the oracle for $P_i(u) = P_i(v)$ when queried with the pair $(b_{i+1}(u),b_{i+1}(v))$ of vertices from $B_{P_i(u)} = B_{P_i(v)}$ (see Figure~\ref{fig:const-stretch-path}).

If $P_2(u) = P_2(v)$, let $d_2$ be the $2$-approximation of $d_{P_2(u)}(u,v)$ that is output using the lookup table associated with $P_2(u) = P_2(v)$. Otherwise, let $d_2 = \infty$. The query algorithm computes $d_2$ and then outputs $\min\{d_0,d_1,d_2\}$.


\subsection{Bounding space, stretch, and query time}
We now analyze our oracle. We start by showing that stretch is $8$. Note that the approximate distance output cannot be an underestimate of $d_G(u,v)$. Let $Q$ be a shortest path from $u$ to $v$ in $G$.  If $Q$ is fully contained in $P_2(u)$ then $P_2(u) = P_2(v)$ and  $d_2$ is a $2$-approximation of $d_G(u, v)$ (see~\Cref{lm:2approx-lookup}) so the output will be a $2$-approximation as well.

Now, assume that $Q$ is not fully contained in $P_2(u)$. Pick $i\in\{0,1\}$ so that $Q$ is fully contained in $P_i(u)$ but not in $P_{i+1}(u)$. Note that $i$ must exist since $P_0(u) = G$. Note also that $d_G(u,v) = d_{P_i(u)}(u,v)$ and that $P_i(u) = P_i(v)$. (See Figure~\ref{fig:const-stretch-path}.)

Since $Q$ is not fully contained in $P_{i+1}(u)$, it follows from the choice of $b_{i+1}(u)$ that $d_{P_{i+1}(u)}(u,b_{i+1}(u))\leq d_G(u,v)$. Similarly, $d_{P_{i+1}(v)}(v,b_{i+1}(v))\leq d_G(u,v)$. By the triangle inequality,
\begin{align*}
  d_i & = d_{P_{i+1}(u)}(u,b_{i+1}(u)) + \tilde d_{P_i(u)}(b_{i+1}(u),b_{i+1}(v)) + d_{P_{i+1}{v}}(v,b_{i+1}(v))\\
      & \leq d_{P_{i+1}(u)}(u,b_{i+1}(u)) + 2d_{P_i(u)}(b_{i+1}(u),b_{i+1}(v)) + d_{P_{i+1}(v)}(v,b_{i+1}(v))\\
      & \leq d_{P_{i+1}(u)}(u,b_{i+1}(u)) + 2(d_{P_{i+1}(u)}(b_{i+1}(u),u) + d_{P_i(u)}(u,v) + d_{P_{i+1}(v)}(v,b_{i+1}(v))) + d_{P_{i+1}(v)}(v,b_{i+1}(v))\\
      & \leq 8d_G(u,v).
\end{align*}
Hence, the output will be an $8$-approximation of $d_G(u,v)$.

The $O(1)$ bound on query time follows immediately from the stored information.

Finally, we show that space is $O(n)$. Consider a level-$i$ piece $P$ of $\mathcal T$ with $i\in\{0,1\}$. By \Cref{Lem:ThorupKlein}, the space required for the oracle associated with $P$ is $O((|P|/\sqrt{r_{i+1}})\log|P|) = O((r_i\log r_i)/\sqrt{r_{i+1}})$. Summing over all $O(n/r_i)$ pieces of level $i$ and over $i\in\{0,1\}$ gives a total space requirement of $O(n\sum_{i = 0}^1(\log r_i)/\sqrt{r_{i+1}}) = O(n\sum_{i = 0}^1(\log r_i)/\sqrt{(\log r_i)^2}) = O(n)$.

Storing $b_i(u)$ and $d_{P_i(u)}(u,b_i(u))$ for each vertex $u$ and for each $i\in\{1,2\}$ clearly only requires $O(n)$ space. It remains to implement the lookup tables associated with the leaves of $\mathcal T$.

\subsection{Compact lookup tables}\label{subsec:compact-table}
Consider a leaf piece $P$ of $\mathcal T$. We now present a data structure with constant query time which can output a $2$-approximation of $d_P(u,v)$ for any vertices $u,v\in V(P)$, and we will show that its space requirement is $O(|P|) = O(r_2)$.

\begin{wrapfigure}{r}{0.4\textwidth}
	\vspace{-20pt}
	\begin{center}
		\includegraphics[width=0.4\textwidth]{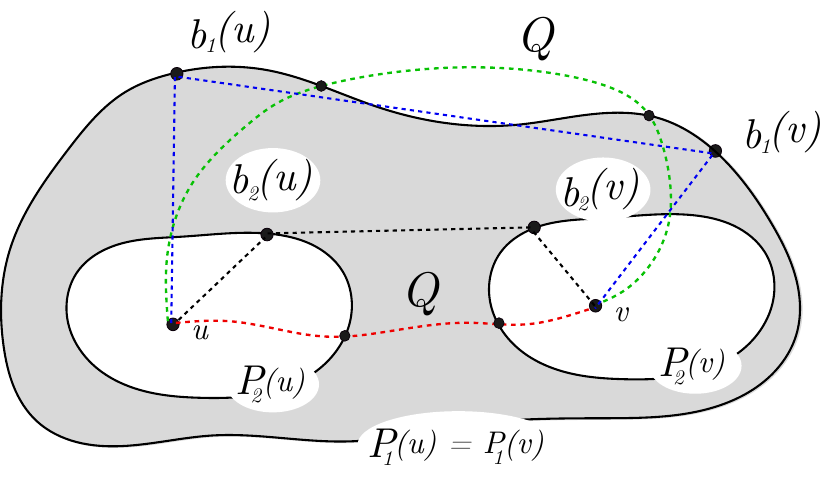}
	\end{center}
	\caption{\footnotesize{When $P_1(u) = P_1(v)$}, the distance oracle returns the minimum length between the lengths of the blue path and the black path. Two scenarios for the shortest path $Q$ between $u$ and $v$: (a) $Q\subseteq P_1(u)$ -- the red path --  or (b) $Q\subseteq P_0(u)$ -- the green path.}
	\vspace{-10pt}
	\label{fig:const-stretch-path}
\end{wrapfigure}

Each edge $e\in E(P)$ is given a unique label $\ell(e)\in\{1,\ldots,|E(P)|\}$ which is represented as a string of $\Theta(\log r_2)$ bits. We also introduce a label $\ell(e_\infty) = 0$ for a dummy edge $e_{\infty}$ of weight $w(e_{\infty}) = \infty$. An array $W_P$ of length $|E(P)|+1$ is associated with $P$ and entry $W_P[\ell(e)]$ contains the weight $w(e)$ of $e$. Each vertex $v\in V(P)$ is given a unique label $\ell(v)\in \{0,1,\ldots,|V(P)|-1\}$ and $\ell(v)$ is represented as a string of $\Theta(\log r_2)$ bits.

For each pair of distinct vertices $u$ and $v$ in $V(P)$, let $e_{uv}$ be an edge of $E(P)\cup\{e_{\infty}\}$ such that $d_P(u,v)\geq w(e_{uv})\geq d_P(u,v)/|E(P)|$; such an edge must exist since either $d_G(u,v) = \infty$ and we can pick $e_{uv} = e_{\infty}$, or there is a shortest path from $u$ to $v$ in $P$ containing at most $|E(P)|$ (in fact $|E(P)| - 1$) edges and not all of these edges can have weight strictly less than $d_P(u,v)/|E(P)|$.

For each pair of distinct vertices $u$ and $v$ in $V(P)$, define the pair $p_{uv} = (\ell(e_{uv}),\rho_{uv})$ where $\rho_{uv} = \lceil d_P(u,v)/w(e_{uv})\rceil$ if $0 < w(e_{uv}) < \infty$, $\rho_{uv} = 0$ if $w(e_{uv}) = 0$, and $\rho_{uv} = 1$ if $w(e_{uv}) = \infty$.

A query for a pair $(u,v)$ is answered by outputting $W_P[\ell(e_{uv})]\cdot\rho_{uv}$. Query time is clearly $O(1)$ if we assume that multiplying an edge weight by an integer written in binary and represented in a single word can be done in $O(1)$ time. We can avoid this assumption by having a look-up table of size $O(r_2)$ whose $i$th entry contains a word with the value $i$ specified using the same number type as that of the edge weights. This way, we only need to assume that numbers of this type can be multiplied together in $O(1)$ time.

The following lemma combined with the fact that $W_P[\ell(e_{uv})] = w(e_{uv})$ shows that the output of a query for $(u,v)$ is a $2$-approximation of $d_P(u,v)$.
\begin{Lem}\label{lm:2approx-lookup}
For each pair of distinct vertices $u$ and $v$ in $V(P)$, $d_P(u,v)\leq w(e_{uv})\rho_{uv}\leq 2d_P(u,v)$.
\end{Lem}
\begin{proof}
The lemma clearly holds if $d_P(u,v) = \infty$ since then $w(e_{uv}) = \infty$ and $\rho_{uv} = 1$ so assume $d_P(u,v) < \infty$. If $d_P(u,v) = 0$ then $w(e_{uv}) = 0$ and the lemma again follows. Finally, assume that $0 < d_P(u,v) < \infty$. Then $0 < w(e_{uv}) < \infty$, $\rho_{uv} = \lceil d_P(u,v)/w(e_{uv})\rceil$, and
\[
  d_P(u,v)\leq w(e_{uv})\lceil d_P(u,v)/w(e_{uv})\rceil < w(e_{uv})\left(\frac{d_P(u,v)}{w(e_{uv})}+1\right) = d_P(u,v) + w(e_{uv}) \leq 2d_P(u,v),
\]
as desired.
\end{proof}

To obtain $O(|P|) = O(r_2)$ space, the only non-trivial part is how to compactly represent all the pairs $p_{uv}$. Recall that $\ell(e_{uv})$ is a bitstring of length $O(\log r_2)$. We claim that the binary representation of integer $\rho_{uv}$ also has length $O(\log r_2)$. This is clear if $w(e_{uv})\in\{0,\infty\}$. When $0 < w(e_{uv}) < \infty$, then by the choice of $e_{uv}$, $\rho_{uv} = \lceil d_P(u,v)/w(e_{uv})\rceil\leq |E(P)| = O(r_2)$, as desired. By padding with $0$s if needed, we can pick an integer $\beta = \Theta(\log r_2)$ such that the binary representations of $\ell(e_{uv})$ and $\rho_{uv}$ both have length exactly $\beta$. Represent the pair $p_{uv}$ with $2\beta$ bits, consisting of the concatenations of the representations of $\ell(e_{uv})$ and $\rho_{uv}$.

For each vertex $u\in V(P)$, we can now store all pairs $p_{uv}$ in a single word $W_u$ using $2\beta|V(P)| = O(r_2\log r_2) = O(\log\log n\log\log\log n)$ bits where the pair $p_{uv}$ is stored with offset $2\beta\ell(v)$ in $W_u$. Both $\ell(e_{uv})$ and $\rho_{uv}$ can be retrieved in constant time by suitable shift and logical-and operations. This gives the total number of words of the lookup table of $O(|P|) = O(r_2)$, as desired. It also follows that query time is constant.

\subsection{Obtaining near-linear construction time}\label{subsec:construction-const-stretch}
It remains to show how to construct our oracle in $O(n\log^3n)$ time.

Finding the recursive decomposition can be done in $O(n)$ time.

By Lemma~\ref{Lem:ThorupKlein}, the total time to construct the oracle for a piece $P$ is $O(|P|\log^3|P|)$. Summing this over all pieces of levels $0$ and $1$ gives a total time bound of $O(n\log^3n)$.

For $i = 1,2$, identifying $b_i(u)$ for each vertex $u$ of a level $i$-piece $P$ can be done in $O(|P|\log |P|)$ time using Dijkstra's algorithm on an augmented graph obtained from $P$ by adding a super source connected to all boundary vertices of $P$ with edges of weight $0$. Over all such pieces $P$, this takes $O(n\log n)$ time.

For the compact lookup tables, the bottleneck in the construction time is to identify the edge $e_{uv}$ for each pair of vertices $u$ and $v$ of a leaf piece $P$ of $\mathcal T$. Since $P$ is small, this can be done sufficiently fast using a simple brute-force approach: first compute all-pairs shortest paths in $P$ in $O(|P|^2\log|P|)$ time with Dijkstra's algorithm from each vertex. Then for each pair of vertices $u$ and $v$ in $P$, traverse the shortest path between them in $O(|P|)$ time to identify $e_{uv}$; over all vertex pairs, this takes $O(|P|^3)$ time. Hence, the total time to construct lookup tables for all leaf pieces of $\mathcal T$ is $O((n/r_2)r_2^3) = O(nr_2^2) = O(n(\log\log n)^4)$.

\paragraph{Avoiding ceiling operations} For our compact lookup tables, we needed the ceiling operation in order to compute each value $\rho_{uv} = \lceil d_P(u,v)/w(e_{uv})\rceil$. We implicitly made the assumption that this operation can be executed in constant time.

To get rid of this assumption, consider the following procedure. Initialize $d \leftarrow 0$. Now, iteratively update $d\leftarrow d + w(e_{uv})$ until $d \geq d_P(u,v)$. Then $\lceil d_P(u,v)/w(e_{uv})\rceil$ is the number of updates to $d$ so a simple counter that is incremented in each step suffices to compute this value. This slows down the preprocessing time for the compact lookup tables by no more than a factor $\rho_{uv} = O(r_2) = O((\log\log n)^2)$ which will not affect the overall $O(n\log^3n)$ preprocessing time bound.


\section{An Approximate Oracle with $(1+\eps)$ Stretch for Undirected Planar Graphs}

In this section, we construct a $(1+\epsilon)$-approximate distance oracle with linear space and constant query time for edge-weighted planar graphs; graphs in this section are \emph{undirected}. We refer the readers to~\Cref{sec:high-level-ideas} for an overview of the proof. The construction  is divided into four major steps:

\begin{enumerate}
	\item In~\Cref{subsec:additive}, we construct a distance oracle restricted to any given subset of vertices $S$ with a small \emph{additive stretch}. The oracle has space $O(|S|)$ and constant query time. 
	\item In~\Cref{subsec:mulitplicative}, we use the additive oracle to construct an oracle restricted to any given subset of vertices $S$ with multiplicative $(1+\epsilon)$ stretch. The oracle has space $O(|S|)$ and constant query time.  The caveat is that the stretch guarantee only applies to pairs of vertices whose distances are in $[d,\alpha d]$ for a given parameter $d$ and a constant $\alpha$. 
	\item In~\Cref{subsec:qusipolynomial-spread}, we show that planar graphs with quasi-polynomial edge length have $(1+\epsilon)$-approximate distance oracle with linear space and constant query time. The construction combines three different tools: a net-tree, a weighted level ancestor data structure, and the restricted oracles in the second step.  
	\item Finally, in~\Cref{subsec:linear-oracle}, we remove the assumption on the edge length of the graph.
\end{enumerate}

We use $\spc(\mathcal{X})$ to denote the total space (in words) of a data structure $\mathcal{X}$. 

\subsection{Additive Restricted Distance Oracles}\label{subsec:additive}

 An \emph{$S$-restricted distance oracle} $\mathcal{D}$ for a planar graph $G$ with \emph{additive stretch $t$} is a data structure that given any two vertices $u,v \in S$,  the estimated distance returned by the oracle, denoted by $d_{\mathcal{D}}(u,v)$, satisfies:
\begin{equation}
d_{G}(u,v) \leq d_{\mathcal{D}}(u,v) \leq d_{G}(u,v) + t
\end{equation}

\noindent This section is devoted to proving the following theorem.

\begin{restatable}{theorem}{AddtiveOracle}
\label{lm:additive-oracle}
Given an edge-weighted $n$-vertex planar graph $G(V,E)$ with diameter $D$, an error parameter $\epsilon < 1$, and a subset of vertices $S$, there is an $S$-restricted distance oracle $\mathcal{D}$ with  $O(|S|\epsilon^{-2})$ space, $O(\epsilon^{-2})$ query time, and additive stretch $\epsilon D$.  Furthermore, $\mathcal{D}$ can be constructed in $O(\epsilon^{-3}n\log^3 n)$ time.
\end{restatable}

We first present our distance oracle ignoring its efficient construction. In~\Cref{subsec:prepro-additive-oracle}, we show how to construct the oracle in nearly linear time.

We base our construction on the idea of Kawarabayashi, Sommer, and Thorup~\cite{KST13}, called \emph{KST construction}, that relies on  recursive decompositions of $G$ using shortest path separators. However, our construction is different from KTS in two respects. First, we restrict  the distance query to be between vertices in a given subset of vertices $S$, and the space bound must be linear in $S$, which could be much smaller than $n$; the space bound in $KTS$ oracle is $\Omega(n)$. As a result, our recursive decompositions must be tailored specifically to $S$, and there are several properties that the decomposition must satisfy altogether. Second, our construction only has three levels, instead of $\log^*(n)$ levels as in the KTS construction.  Specifically, the top level is the vertex set $S$ and a recursive decomposition for $S$; the second level contains subsets of $S$ corresponding to leaves of  the recursive decomposition for $S$; the third level contains subsets of those in the second level. Subsets of $S$ in the third level are small enough that we can afford to have a table lookup that contains \emph{encodings} of approximate distances (instead of these distances themselves); the same idea was used in~\Cref{subsec:compact-table}.  With these ideas, we are able to achieve space and query time bounds independent of $n$, while KTS has a query time bound of $O(\log^* n)$ and a space bound of $\Omega(n)$ when $\epsilon$ is a constant.

\begin{wrapfigure}{r}{0.4\textwidth}
	\vspace{-15pt}
	\begin{center}
		\includegraphics[width=0.4\textwidth]{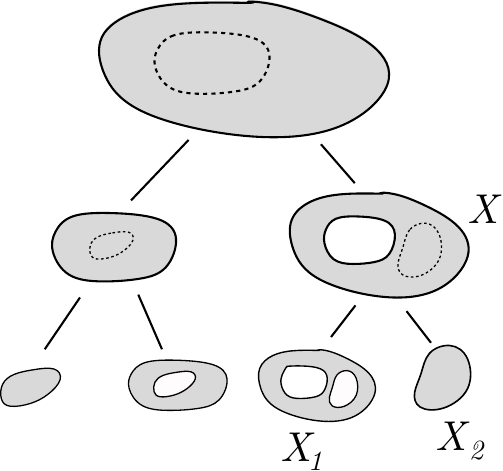}
	\end{center}
	\caption{\footnotesize{A recursive $(S,\tau)$-decomposition $\mathcal{T}$. Child pieces of $X$ is obtained by separating $ X$ a shortest path separator. While regions are \emph{holes} of $X$.  $\Pi(X)$ contains the shortest paths on the boundaries of the holes and paths in the separator that separates $X$.}}
	\vspace{-20pt}
	\label{fig:S-tau-decomp}
\end{wrapfigure}

\paragraph{Recursive $(S,\tau)$-decomposition.~} A recursive $(S,\tau)$-decomposition for a given $S \subseteq V$ and an integer $\tau \geq 1$ is a binary tree $\mathcal{T}$ (see Figure~\ref{fig:S-tau-decomp})  such that:
\begin{enumerate}
	\item[(1)] Each node $X \in \mathcal{T}$ is associated with a subset $\chi(X)$ of vertices of $S$ and a set of $\Pi(X)$ of at most $10$ shortest paths of $G$. (We use nodes to refer to vertices of $\mathcal{T}$.)
	\item[(2)]  If $X$ is an internal node with two children $X_1,X_2$, then (a) $|\chi(X)| > \tau$, (b) $\chi(X_1)\cap \chi(X_2) = \emptyset$,  (c) $\chi(X_1)\cup \chi(X_2) \subseteq \chi(X)$ and (d) $\chi(X)\setminus (\chi(X_1)\cup \chi(X_2)) \subseteq \cup_{P\in \Pi(X)} P$. In addition, if $X$ is a leaf of $\mathcal{T}$, then (e) $\Pi(X) = \emptyset$ and (f) $|\chi(X)|  \leq \tau$.
	\item[(3)] For any node $X \in \mathcal{T}$, any pair $(u,v)\in \chi(X)\times (S\setminus \chi(X))$, any path between $u$ and $v$ must intersect some path $P \in \Pi(\parent(X))$ where $\parent(X)$ is the parent of $X$. 
	\item[(4)] $\depth(\mathcal{T}) = O(\log |S|)$ and $|\leavs(\mathcal{T})| \leq O(\frac{|S|}{\tau})$ where $\leavs(\mathcal{T})$ is the set of leaves of $\mathcal{T}$.
\end{enumerate}

 In constructing a recursive $(S,\tau)$-decomposition $\mathcal{T}$, we will operate on a triangulation $G_{\Delta}$ of $G$. We also associate each node $X$ of $T$ with a subgraph $H_X$ of $G_{\Delta}$; the root is associated with $G_{\Delta}$. If $\chi(X) > \tau$, we use shortest path separator $C_X$ to separate $H_X$ into two subgraphs $H_{X_1},H_{X_2}$ associated with children $X_1,X_2$ of $X$, respectively. One subgraph, say $H_{X_1}$, is induced by vertices and edges on and inside $C_X$ and another subgraph is induced  by vertices and edges on and outside $C_X$. $C_X$ becomes a \emph{hole} of $H_{X_2}$ and an infinite face of $H_{X_1}$; note that $H_{X_1}\cap H_{X_2} = C_X$.

Naturally, $\chi(X_1)$ and $\chi(X_2)$ are the subsets of $\chi(X)$ inside and outside $C_X$.  The set of paths $\Pi(X)$ contains two paths in $C_X$ and the shortest paths in (a constant number of) holes of $H_X$ due to separations in previous steps. It could be that a vertex $s\in \chi(X)$ belongs to $C_X$ and hence, is not included in $\chi(X_1)\cup \chi(X_2)$. We stop the decomposition at node $X$ when $\chi(X) \leq \tau$; this guarantees property (2).

The number of paths $\Pi(X)$ is at most twice the number of holes of $H_X$. We will guarantee that the number of holes of $H_X$ is at most $4$. This guarantees property (1) as $\Pi(X)$ contains, in addition to 8 shortest paths on 4 holes, 2 shortest paths on the shortest path separator that separates $X$ into two children. If $H_X$ has exactly four holes, we design the weight function $\omega_X$ in a way that the child graphs $H_{X_1}$ $H_{X_2}$ have at most $3$ holes each. Otherwise, $\omega_X(v) = 1$ if $v\in \chi(X)$ and $0$ otherwise. This way we can argue that the depth of $\mathcal{T}$ is $O(\log |S|)$ as claimed in property (4). The fact that the number of leaves is at most $O(\frac{|S|}{\tau})$ is closely related to the weight function $\omega_X$:  Except when reducing the number of holes of $H_X$, $\omega_X$ guarantees that $\chi(X)$ is reduced by at least a $\frac{2}{3}$ fraction.  

Property (3) follows directly from the observation that any path from a vertex $u$ inside $H_X$ to a vertex $v$ outside $H_X$ must cross at least one of the holes of $H_X$ and that the boundaries of these holes are shortest paths associated with $X$'s parent. These are the ideas behind the following decomposition lemma, whose formal proof will be deferred to~\Cref{subsec:proof-decomp-lemma}.

\begin{restatable}[Decomposition Lemma]{lemma}{DecompositionLemma}\label{lm:S-division}
Given any parameter $\tau \geq 1$ and a subset of vertices $S \in V(G)$, there exists a recursive $(S,\tau)$-decomposition  $\mathcal{T}$ of $G$ satisfying all properties (1)-(4). 
\end{restatable}

\noindent We observe the following simple properties of $\mathcal{T}$. 

\begin{observation}\label{obs:T-prop}
 $\{\chi(X)\}_{X\in \leavs(\mathcal{T})}$ are disjoint subsets of $S$. 
\end{observation}
\begin{proof}
By induction and the construction, it holds that  if $X$ is not an ancestor of $Y$ and vice versa, then $\chi(X)\cap \chi(Y) = \emptyset$. Thus, the subsets of $S$ associated with leaves of $\mathcal{T}$ are disjoint.  
\end{proof}

For each $v\in S$, we define $\home(v)$ to be a node $X$ such that $v\in \chi(X)$ and either (i) $X$ is a leaf of $\mathcal{T}$, or (ii) $X$ has two children $X_1$ and $X_2$ where $v\in \chi(X)\setminus (\chi(X_1)\cup \chi(X_2))$. Node $X$ in case (ii) exists by property (2d) of $\mathcal{T}$, and in this case, $v\in P$ for some path $P \in \Pi(X)$. Thus, if $u$ and $v$ do not belong to the same leaf of $\mathcal{T}$, then either $\home(u)\not= \home(v)$ or $\home(u) = \home(v) = X$ for some internal node $X \in \mathcal{T}$. 

\begin{wrapfigure}{r}{0.4\textwidth}
	\vspace{-15pt}
	\begin{center}
		\includegraphics[width=0.4\textwidth]{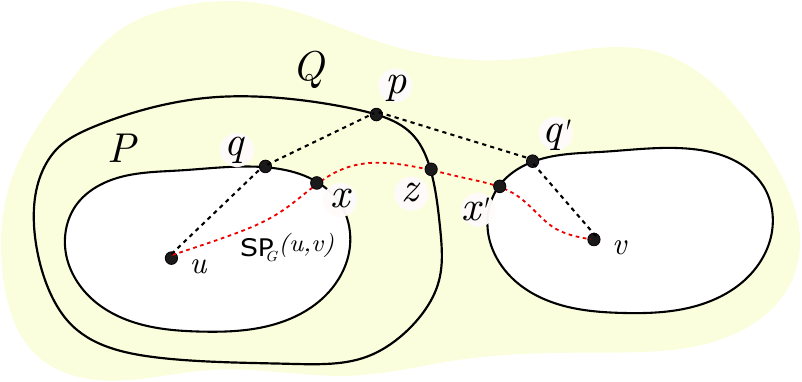}
	\end{center}
	\caption{\footnotesize{Illustration for the proof of~\Cref{lm:D1}}.}
	\vspace{-20pt}
	\label{fig:additive-stretch}
\end{wrapfigure}

\begin{claim}\label{clm:separation-T}
	For any pair of vertices $(u,v)\in S\times S$, either (a) $\home(u)$ and $\home(v)$ are the same leaf of $\mathcal{T}$, or (b) there is a node $Z = \lcads_{\mathcal{T}}(\home(u),\home(v))$ such that any path between $u$ and $v$ must intersect some path $P \in \Pi(Z)$.
\end{claim}
\begin{proof}
	Suppose that (a) does not happen. Then $\chi(Z)$ contains both $u$ and $v$ by property (2) of $\mathcal{T}$, and that either (i) $u$ and $v$ belongs to the subsets in different children of $Z$ or (ii) at least one of $u$ and $v$ belongs to some path $P$ of $\Pi(Z)$. Case (ii) immediately implies (b) of the claim. For case (i), suppose that $u \in \chi(Z_1)$ where $Z_1$ is a child of $Z$. Then applying property (3) of $\mathcal{T}$ to $Z_1$, it holds that any path between $u$ and $v$ must intersect some path $P \in \Pi(Z)$, as desired.
\end{proof}

\begin{lemma}\label{lm:D1}Given a recursive $(S,\log(|S|))$-decomposition
	 $\mathcal{T}$ of a planar graph $G(V,E,w)$ for a subset of vertices $S$, 	there is a data structure $\mathcal{D}$ with $O(|S|\epsilon^{-2})$ space and query time $O(\epsilon^{-2})$ such that for any pair of vertices $(u,v)\in S\times S$, if $u$ and $v$ do not belong to the same leaf of $\mathcal{T}$, $\mathcal{D}$  returns an estimate distance $d_{\mathcal{D}}(u,v)$ such that:
	\begin{equation}
	d_G(u,v) \leq d_{\mathcal{D}}(u,v)\leq d_{G}(u,v) + \epsilon D
	\end{equation}	
	Otherwise, $\mathcal{D}$ returns \textsc{Failed} and the (id of the) leaf $X$ of $\mathcal{T}$ whose associated set $\Pi(X)$ contains both $u$ and $v$. 
\end{lemma}
\begin{proof} Let $X$ be a node in $\mathcal{T}$ and $\parent(X)$ be the parent node of $X$ in $\mathcal{T}$.  For each shortest path $P$ in $\Pi(X)$ , we \emph{portalize} $P$ by a set $\port(P)$ of $O(\frac{1}{\epsilon})$ \emph{portals} such that the distance between any two nearby portals is at most $\epsilon D$; we can take $\port(P)$ to be a set of $(\epsilon D/2)$-net  of $P$.  (Note that $w(P)\leq D$.) We denote by $\port(X) = \cup_{P\in \Pi(X)}\port(P)$. Observe that $|\port(X)| = O(\frac{1}{\epsilon})$ since $|\Pi(X)|\leq 10$.
	
	For each portal $v\in \port(X)$,  we store the distances from $v$ to all vertices in $\cup_{Y \in \ansc(X)}\port(Y)$ where $\ansc(X)$ is the set of ancestors of $X$ in $T$ (including $X$).   
	
	For each vertex $s\in S$, we store (the id of) $\home(s)$  and the distance from $s$ to every vertex in the set of \emph{portals} $p(s)\stackrel{\mbox{\tiny{def.}}}{=} \port(\home(s))\cup \port(\parent(\home(s)))$.
	
	We also store the tree $\mathcal{T}$ (but not the set of vertices and paths associated with nodes of $\mathcal{T}$), the set of portals $\port(X)$ for each internal node $X$ and the subset of vertices in $S$ associated with each leaf node of $\mathcal{T}$. Additionally, we store a data structure $\lcads_{\mathcal{T}}$ to suport LCA queries in $O(1)$ time; $\lcads_{\mathcal{T}}$ occupies $O(|\mathcal{T}|)$ words. This completes the construction of $\mathcal{D}$.
	
	We now bound the space of $\mathcal{D}$.  Since $\mathcal{T}$ is a binary tree, $|V(\mathcal{T})| = O(|\leavs(\mathcal{T})|) = O(\frac{|S|}{\log(|S|)})$. Thus, the total number of portals is $O(\frac{|S|}{(\log |S|) \epsilon})$. For each portal, by property (4), we store distances to $ O(\depth(\mathcal{T}))O(\frac{1}{\epsilon})  = O(\log (|S|) \epsilon^{-1})$ other vertices. Thus, the total space to store portal distances is:
	\begin{equation}\label{eq:portal-distace-space}
	O(\frac{|S|}{(\log |S|) \epsilon}) O((\log |S| )\epsilon^{-1}) = O(|S|\epsilon^{-2})
	\end{equation}
	
	For each vertex $s\in S$, we only store distances to $O(\epsilon^{-1})$  portals, which costs $O(|S|\epsilon^{-1})$ space. Thus, the total space is still bounded by~\Cref{eq:portal-distace-space}.
	
	We now describe distance queries in $\mathcal{D}$. Let $(u,v)\in S\times S$  be a query. If $\home(u) = \home(v) = X$ where $X$ is the leaf of $\mathcal{T}$, $\mathcal{D}$ will return $X$ and output \textsc{Failed}. Otherwise, let $Z = \lcads_{\mathcal{T}}(\home(u),\home(v))$. 	Let $p\in \port(Z)$ be a portal. We compute the estimate distance from $u$ to $p$ as follows. 
	\begin{equation}\label{eq:dist-portals}
	d_{\mathcal{D}}(u,p) = \min_{q \in p(u)} d_G(u,q) + d_G(q,p)
	\end{equation}
	Observe that either $p \in p(u)$, when $Z = \home(u)$, or $p(u)$ contains the portals of two descendant nodes of $Z$, when $Z\not= \home(u)$. In both cases,  $d_G(p,q)$ is stored in our data structure. Since the distance from $u$ to every vertex in $p(u)$ is stored and since $|p(u)| = O(\epsilon^{-1})$, it follows that $d_{\mathcal{D}}(u,p)$ can be computed in $O(\epsilon^{-1})$ time.  The estimate distance $d_{\mathcal{D}}(v,p)$ between $v$ and $p$ can be obtained in $O(\epsilon^{-1})$ time by exactly the same way.
		
	We then approximate the distance between $u$ and $v$ by computing:
	\begin{equation}\label{eq:terminal-dist-D1}
	d_{\mathcal{D}}(u,v) =  \min_{p \in \port(Z)} d_\mathcal{D}(u,p) + d_{\mathcal{D}}(p,v)
	\end{equation}
	The total query time is $O(\epsilon^{-2})$ since $|\port(Z)| = O(\epsilon^{-1})$. 
	
	It remains to bound the additive stretch of $\mathcal{D}$. By~\Cref{clm:separation-T}, there is a path $Q \in \Pi(Z)$ such that there exists $z \in \SP_G(u,v)\cap Q$. Let $p$ be the portal in $\port(Z)$ closest to $z$ (see~\Cref{fig:additive-stretch}). Then $d_G(p,z)\leq \epsilon D$. By the triangle inequality, it holds that:
	\begin{equation}\label{eq:through-port-dist}
	d_G(u,p) + d_G(v,p) \leq d_G(u,v) + 2\epsilon D
	\end{equation}
	If we can show that:
	\begin{equation}\label{eq:dist-ancestor}
	d_{\mathcal{D}}(u,p)\leq d_G(u,p) + 4\epsilon D  \quad \& \quad d_{\mathcal{D}}(v,p)\leq d_G(v,p) + 4\epsilon D.
	\end{equation}
	then, by~\Cref{eq:through-port-dist}, we have $d_{\mathcal{D}}(u,v)~\leq~ d_G(u,p) +  10\epsilon D$. By setting $\epsilon' = \epsilon/10$, we obtain the desired additive stretch $\epsilon' D$, and we are done.
	
	Now we focus on showing~\Cref{eq:dist-ancestor} for $u$; the argument for $v$ is symmetric. If $Z = \home(u)$,  then $p \in p(u)$ and hence, $d_{\mathcal{D}}(u,p) = d_G(u,p) $. Otherwise, $\SP_G(u,v)$ intersects some vertex $x$ on a path $P \in \Pi(\parent(\home(u)))$ by property (3) of $\mathcal{T}$. Let $q$ be the closest portal to $x$ on $P$. By the triangle inequality, we have:
	\begin{equation*}
	\begin{split}
	d_{G}(u,q) + d_{G}(p,q) &\leq\left(  d_G(u,x) + d_G(x,q) \right) +  (d_G(q,x) + d_G(x,z) + d_G(z,p))\\
	&\leq d_G(u,x)  + d_G(x,z) + 3\epsilon D = d_G(u,z)  + 3\epsilon D\\
	&\leq  d_G(u,p) + 4\epsilon D.
	\end{split}
	\end{equation*}
	Since $q\in p(u)$, by~\Cref{eq:terminal-dist-D1}, $d_{\mathcal{D}}(u,p) ~\leq d_{G}(u,q) + d_{G}(p,q) ~\leq d_G(u,p) + 4\epsilon D$ as desired.
\end{proof}

In the following lemma, we devise a simple construction of a data structure with  linear space and constant query time when $S$ has a sufficiently small size, 

\begin{lemma}\label{lm:small-S}
	Given an $n$-vertex planar graph $G$ with diameter $D$, an error parameter $\epsilon < 1$, and a subset $S\subseteq V$ such that $|S| \leq \log \log n$, there is an $S$-restricted distance oracle of  additive stretch $\epsilon D$ with  $O(|S|\epsilon^{-1})$ space  and $O(1)$ query time .
\end{lemma}
\begin{proof}
	If $\frac{1}{\epsilon} > \log\log n $, then by storing pairwise distances between vertices in $S$, the total amount of space is $O(|S|^2) = O(|S|\epsilon^{-1})$. The query time is $O(1)$ in this case. Thus, we can assume that $\epsilon^{-1}\leq \log\log n$. Index vertices of $S$ as $\{1,2,\ldots ,|S|\}$. For each vertex, we can afford to store its index in $S$ in one word. Let $v_i\in S$ be a vertex corresponding to index $i$ . For each vertex $v\in S$, we divide a word (of at least $\mword = \Omega(\log(n))$ bits) into $\Omega(\frac{\log n}{\log(2\log \log n)})$ blocks of size $\lceil \log(2\log \log n)\rceil$ each. For the $i-$th block, we store $\rho(v,v_i) = \lceil \frac{d_G(v,v_i)}{\epsilon D} \rceil$. Observe that $\rho(v,v_i) \leq \frac{1}{\epsilon} + 1 \leq 2\log\log n$. Thus, we can afford to store all $\rho(v,v_i) $, $1\leq i\leq |S|$ in one word, called the \emph{distance word}, since $|S| \ll \frac{\log n}{\log\log n}$ when $n$ is sufficiently big. The total space is then $O(|S|)$ words. To query the distance from $v$ to $v_i$ in $S$, we first look up the index $i$ in $S$, extract the $i$-th block of the distance word of $v$ in $O(1)$ time, and then return $\rho(v,v_i)\epsilon D$ as the estimate distance between $v$ and $v_i$. The additive stretch is at most $\epsilon D$ since $\rho(v,v_i)\epsilon D = \lceil \frac{d_G(v,v_i)}{\epsilon D} \rceil \epsilon D \leq d_G(v,v_i) + \epsilon D$.
\end{proof}

\noindent We now have all necessary tools to prove~\Cref{lm:additive-oracle}.

\begin{proof}[Proof of~\Cref{lm:additive-oracle}]
	We construct a 3-level data structure represented by a tree $\mathbb{T}$. The root of $\mathbb{T}$ at level $0$ corresponds to $S$. We construct a recursive $(S,\log(|S|))$-decomposition $\mathcal{T}_0$ and a data structure $\mathcal{D}_0$ corresponding to $\mathcal{T}_0$ using~\Cref{lm:D1}. 
	
	Children of the root, or level-$1$ nodes, are the subsets of $S$ associated with leaves of $\mathcal{T}_0$.  Observe that subsets at level-$1$ have size $O(\log |S|)$.  
	
	Let $\mathcal{X}$ be a level-1 node, and $S_{\mathcal{X}}$ be the subset of $S$ associated with $\mathcal{X}$. We again construct a  recursive $(S_\mathcal{X},\log(|S_\mathcal{X}|))$-decomposition $\mathcal{T}_\mathcal{X}$ and a data structure $\mathcal{D}_\mathcal{X}$ corresponding to $\mathcal{T}_\mathcal{X}$ using~\Cref{lm:D1}.  Children of $\mathcal{X}$ are associated with the subsets of $S_\mathcal{X}$ corresponding to leaves of $\mathcal{T}_\mathcal{X}$.

	Observe that for each subset $S_\mathcal{Y}\subseteq S$ associated with a node $\mathcal{Y}$ at level-2 of $\mathbb{T}$, $|S_\mathcal{Y}| = O(\log \log |S|) = O(\log \log n)$. We construct a data structure with $O(|S_\mathcal{Y}|\epsilon^{-1})$ space and $O(1)$ query time to query distances between pairs $(u,v) \in S_\mathcal{Y}\times S_\mathcal{Y}$ using \Cref{lm:small-S}.

	Our distance oracle $\mathcal{D}$ will consist of all the data structure $\mathcal{D}_\mathcal{X}$ associated with each node $\mathcal{X}$ in $\mathbb{T}$. Since subsets of $S$ at each level of $\mathbb{T}$ are disjoint, the total space of $\mathcal{D}$ is $O(|S|\epsilon^{-2})$ by~\Cref{lm:D1} and~\Cref{lm:small-S}. 

	We now show how to answer a query $(u,v) \in S\times S$. The data structure  $\mathcal{D}_0$ at the root of $\mathbb{T}$, in $O(\epsilon^{-2})$ time, either returns a $(1+\epsilon)$-approximate distance between $u$ and $v$ -- in this case we are done -- or the id of the child, say $\mathcal{X}$, of the root whose associated subset $S_\mathcal{X}$ contains both $u$ and $v$. We then query $\mathcal{D}_\mathcal{X}$. In $O(\epsilon^{-2})$ time,  $\mathcal{D}_\mathcal{X}$ either returns a $(1+\epsilon)$-approximate distance between $u$ and $v$ -- again in this case we are done -- or the id of the child, say $\mathcal{Y}$, whose associated subset $S_\mathcal{Y}$ contains both $u$ and $v$. In the latter case, by querying $\mathcal{D}_\mathcal{Y}$ in $O(1)$ time, we get a $(1+\epsilon)$-approximate distance between $u$ and $v$. The total query time is hence $O(\epsilon^{-2})$. 
	
	The stretch guarantee of $\mathcal{D}$ follows directly from~\Cref{lm:D1} and~\Cref{lm:small-S}. 
\end{proof}

\subsubsection{Obtaining near-linear construction time}\label{subsec:prepro-additive-oracle}

A recursive $(S,\tau)$-decomposition $\mathcal{T}$ along with portals on each shortest path associated with nodes of $\mathcal{T}$ can be computed in $O(n\log^3 n \epsilon^{-2})$ time following the same approach by Kawarabayashi, Sommer and Thorup~\cite{KST13} (see the proof of Theorem 1.3 in~\cite{KST13}). Given $\mathcal{T}$, we can compute $\home(u)$ for each $u\in V(G)$ in $O(n)$ time.

To compute distances between vertices quickly, we use  the distance labeling scheme of Thorup in~\Cref{thm:Thorup-labeling}. Although the distances obtained through the labeling scheme are $(1+\epsilon)$-approximate, they suffice for our purpose, as we will point out in each step of the construction below.  

\paragraph{Step 1: constructing the oracle in ~\Cref{lm:D1}.~}  We assume that the recursive $(S,\tau)$-decomposition $\mt$ and the set of portals $\Pi(X)$ associated with each node $X$ in $\mt$ are given as an input to the construction in Lemma~\ref{lm:D1}.

 Recall that $\mt$ has $O(\frac{|S|}{\log(|S|)}) = O(|S|)$ nodes and depth $O(\log |S|)$.  For each node $X\in \mt$ and every node $v\in \port(X)$, we need to compute the distance from $v$ to every vertex in the portal sets of all ancestors of $X$; there are $O(\epsilon^{-1} \depth(\mt) = O(\epsilon^{-1}\log |S|)$ such vertices. Thus the total number of \emph{distance computations} is $O(|V(\mt)|\epsilon^{-2}\log |S|) = O(\epsilon^{-2} |S|\log |S|)$.

Next, for each vertex $s\in S$, we need to compute the distances from $s$ to all vertices in $p(s)$ -- the set of portals of $s$. Since $|p(s)| = O(\epsilon^{-1})$, the total number of distances computed in this step is $O(|S|\epsilon^{-1})$.

In summary, the total number of distance computations remains to be $O(\epsilon^{-2} |S|\log |S|)$. We use Thorup's distance labeling scheme  to query each $(1+\epsilon)$-approximate distance in time $O(\epsilon^{-1})$; this incurs $O(|S|\epsilon^{-3} \log |S|)$ total  running time.

Since $(1+\epsilon)$-approximate distances are stored instead of exact distances, for each query $(u,v)$, $\md$ returns:
\begin{equation*}
\begin{split}
d_{\md}(u,v)\leq (1+\epsilon) d_G(u,v) + \epsilon D  \leq d_G(u,v) + 2\epsilon D
\end{split}
\end{equation*}
 since $d_G(u,v)\leq D$. Thus, by rescaling $\epsilon \leftarrow \epsilon/2$, we recover additive stretch $\epsilon D$ while increase the running time and space by only a constant factor.

\paragraph{Step 2: constructing the oracle in~\Cref{lm:small-S}.~} Since $|S| = O(\log \log n)$, we can use Thorup's labeling scheme to query approximate distances between all pairs of vertices in $S$; the total running time to query all the distances is $O(\epsilon^{-1}(\log \log n)^2)$.

Computing $\rho(v,v_i)$ involves the ceiling function. Note that $x = 2\lceil \log\lceil\log n\rceil\rceil$ is an integer upper bound on $\rho(v,v_i)$. If we are given $x$, the ceiling involved in $\rho(v,v_i)$ can then be computed in $O(\log\log\log n)$ time with binary search in the range of integers from $1$ to $x$. Computing $x$ needs only be done once during the entire construction of the oracle and it can be done in $O(\log n)$ time as follows: first compute $\lceil\log n\rceil$ in $O(\log n)$ time by repeated doubling; using the same approach, compute $x$ in additional $O(\log\log n)$ time.

For each node $v\in S$, packing distances of other nodes in one word take $O(\log n)$ time. Thus, the total running time is $O(\epsilon^{-1} |S| \log n)$. 

The returned distance between $(u,v)$ by the oracle will have an additive stretch $2\epsilon D$ instead of $\epsilon D$; we can recover stretch $\epsilon D$ by scaling.

\begin{wrapfigure}{r}{0.45\textwidth}
	\vspace{-25pt}
	\begin{center}
		\includegraphics[width=0.45\textwidth]{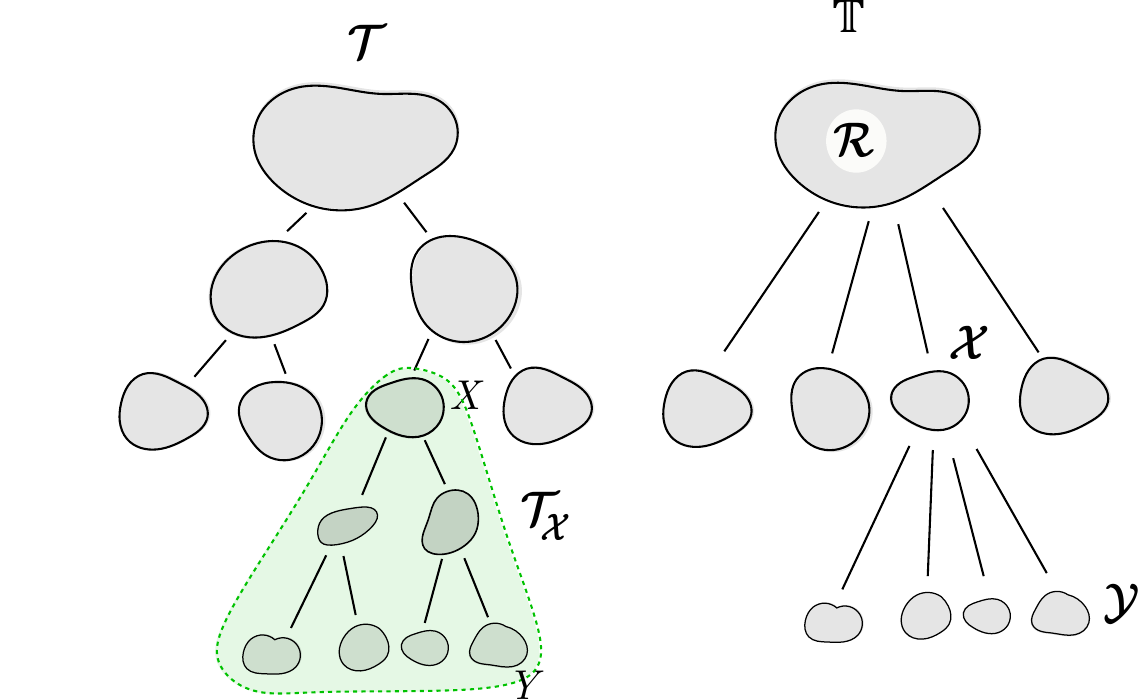}
	\end{center}
	\caption{\footnotesize{The 3-level data structure $\mathbb{T}$ and the corresponding recursive decomposition tree $\mt$}.}
	\vspace{-10pt}
	\label{fig:tree-preproc}
\end{wrapfigure}

\paragraph{Step 3: constructing the oracle in~\Cref{lm:additive-oracle}.} In the proof of~\Cref{lm:additive-oracle}, we construct a 3-level data structure $\mathbb{T}$, where each node $\mx$ of $\mathbb{T}$ is associated with a subset $S_\mx$ of $S$. (Subsets associated with nodes in the same level are pairwise disjoint.) For each node $\mx$, we compute a recursive $(S_\mx,\log(|S_\mx|))$-decomposition $\mt_\mx$ and a data structure $\md_\mx$ by Step 1. However, if we compute  $\mt_\mx$ from the original graph $G$, the running time for each node of $\mathbb{T}$ is $O(n\epsilon^{-2}\log^3)$ for each node. This could results in $\Omega(n^2)$ running time as $\mathbb{T}$ can have up to $\Omega(n)$ nodes.

Our idea to obtain an efficient construction is to construct a \emph{single} recursive decomposition tree $\mt$, and for every node $\mx$ of $\mathbb{T}$, we can extract the decomposition $\mt_{\mx}$ from $\mt$ (see~\Cref{fig:tree-preproc}). 

Specifically, starting from the root node, say $\mr$, of $\mathbb{T}$,  we compute an $(S,\log|S|)-$decomposition $\mt$ of $\mr$. (Note that $S_\mr = S$.)   Each child $\mx$ of $\mr$ in $\mathbb{T}$ has a corresponding leaf node $X$ in $\mt$ such that $S_{\mx} = \chi(X)$. We then continue the recursive decomposition following  the algorithm of  Kawarabayashi, Sommer, and Thorup~\cite{KST13}  to  decompose leaf node $X$ of $\mt$ further. Each leaf $Y$ in the subtree rooted at $X$ has $\chi(Y) = \log(|\chi(X)|) = \log(|S_{\mx}|)$, and corresponds to a leaf of $\mathbb{T}$.  The tree $\mt_{\mx}$ of $\mx$ is now the subtree rooted at node $X$ of $\mt$ (see~\Cref{fig:tree-preproc}).  Given $\mt_{\mx}$ and the portals associated with each node,  the running time to construct $\md_{\mx}$ is now $O(\epsilon^{-3} |S_{\mx}|\log(|S_{\mx}|))$. Thus, the total running time for each level $i$ of $\mathbb{T}$ is:
\begin{equation}\label{eq:time-per-level}
O(\epsilon^{-3}) \sum_{\mx \mbox{ is at level }i} |S_{\mx}|\log(|S_{\mx}|)) = O(\epsilon^{-3} n\log n)
\end{equation}

Observe that each leaf $\my$ of $\mt$ corresponds to a leaf of $\mathbb{T}$. We apply the construction in Step 2 with running time $O(\epsilon^{-2}|S_{\my}| \log n)$ and hence, the total running time to construct the oracles of leaves of $\mathbb{T}$ is $O(\epsilon^{-2} n\log n)$.

The time to construct $\mt$ and  all portals associated with nodes of $\mt$  is  $O(\epsilon^{-2} n \log^3 n)$ by the algorithm of  Kawarabayashi, Sommer, and Thorup~\cite{KST13}. Since $\mathbb{T}$ only has 3 levels, by~\Cref{eq:time-per-level},  the total running time to construct $\mathcal{D}$ is:
 \begin{equation*}
 O(\epsilon^{-3} n \log n) + O(\epsilon^{-2} n \log^3 n) ~=~ O(\epsilon^{-3}n\log^3 n)
 \end{equation*}

\subsection{Multiplicative Restricted Distance Oracles}\label{subsec:mulitplicative}

In this subsection, we prove~\Cref{thm:restricted-oracle} that we restate below. The main tool we use in this section is sparse covers.

 \RestrctedOracle*
 
 \begin{proof} Let $\mathcal{C} = \{C_1,\ldots, C_k\}$ be a $(\beta, s, \beta \alpha d)$-sparse cover of $G$ with $\beta, s = O(1)$ as guaranteed by~\Cref{lm:sparsecover-time}.  We remove from $\mathcal{C}$ every cluster $C_i$ such that $C_i \cap S = \emptyset$. Let $\mathcal{C}_S$ be the resulting set of clusters. 
 	
 	Let $D = \beta  \alpha d = O(\alpha d)$.  For each cluster $C \in \mathcal{C}_S$, let $S_C = S\cap C$. We apply~\Cref{lm:additive-oracle} to construct a distance oracle $\mathcal{D}_C$ for $S_C$ in (planar) graph $C$ with additive stretch $\epsilon_0D$ with $\epsilon_0 = \frac{\epsilon}{\beta \alpha} = O(\frac{\epsilon}{\alpha})$. Our data structure $\mathcal{D}$ consists of all oracles  $\{\mathcal{D}_C\}_{C\in \mathcal{C}_S}$, and additionally, for each vertex $v \in S$, we will store (the id of) $\mathcal{D}_C$ such that $C$ contains $B_G(v,\alpha d)$; $C$ exists by property (2) of sparse covers.
 	
 	We first bound the space of $\mathcal{D}$. Observe that $\sum_{C \in \mathcal{C}_S}|S_C| \leq s|S|$ since every vertex of $S$ belongs to at most $s$ clusters in $\mathcal{C}_S$. By~\Cref{lm:additive-oracle}, $\spc(\mathcal{D}_C) = O(\epsilon_0^{-2} |S_C|) = O(\epsilon^{-2} \alpha^2|S_C|)$. Thus, $\sum_{C\in \mathcal{C}_S}\spc(\mathcal{D}_C) = O(\epsilon^{-2}\alpha^2 s) |S| = O(\epsilon^{-2} \alpha^2 |S|)$ as desired.
 	
 	Given a query pair $(u,v)\in S\times S$, we first identify in $O(1)$ time the oracle $\mathcal{D}_C$ such that $C$ contains $B_G(v,\alpha d)$. If $d_G(u,v) < \alpha d$,  $u$ may not belong to $C$, and if this is the case, the oracle returns $+\infty$. Otherwise, $u \in B_G(v,\alpha d)$ and hence, $u\in C$. We then query $\mathcal{D}_C$ in $O(\epsilon_0^{-2}) = O(\alpha^2\epsilon^{-2})$ time to get the approximate distance $d_{\mathcal{D}_C}(u,v)$. We  return this distance as an approximate distance by $\mathcal{D}$. 
 	
 	Observe that the total query time is $O(\alpha^2\epsilon^{-2})$. To bound the stretch, observe that if $d_G(u,v)\geq d$, then:
 	\begin{equation*}
 	d_{\mathcal{D}}(u,v) = d_{\mathcal{D}_C}(u,v) \leq d_G(u,v) + \epsilon_0 D = d_G(u,v)  + \epsilon d \leq (1+\epsilon)d_G(u,v)
 	\end{equation*}
  Otherwise, $d_{\mathcal{D}}(u,v) \geq d_G(u,v)$, but there may not be any useful upper bound on $d_{\mathcal{D}}(u,v)$.

  It remains to show that the construction time of $\md$ is $O(\epsilon^{-3}n\log^3 n)$ time. By~\Cref{lm:sparsecover-time}, $\mathcal{C}$ can be constructed in $O(n)$ time. By~\Cref{lm:additive-oracle}, each oracle $\md_C$ for each cluster $C\in \mathcal{C}_{S}$ can be constructed in time $O(\epsilon_0^{-3} |V(C)|\log^3(|V(C)|)) = O(\epsilon^{-3}\alpha^{3} |V(C)|\log^3(|V(C)|))$. Thus, the total construction time of $\md$ is:
  \begin{equation*}
  \sum_{C\in \mathcal{C}_S}O(\epsilon^{-3}\alpha^{3} |V(C)|\log^3(|V(C)|)) = O(\epsilon^{-3}\alpha^3n\log n),
  \end{equation*}
  as desired.
 \end{proof}

\subsection{Oracles for Planar Graphs with Quasi-polynomial Spread}\label{subsec:qusipolynomial-spread}

In this section, we construct a $(1+\epsilon)$-approximate distance oracle for planar graphs with quasi-polynomial edge weights. The oracle has linear space and constant query time.

\begin{theorem}\label{thm:quasi-poly-spread} Given an $n$-vertex planar graph $G(V,E,w)$  with spread $\Delta = 2^{O(\log^c n)}$ for some constant $c\geq 1$, there is a $(1+\epsilon)$-approximate distance oracle $\mathcal{D}$ with $O(n\epsilon^{-2}\log\frac{1}{\epsilon})$ space and query time $O(\epsilon^{-2})$. Furthermore, $\md$ can be constructed in time $O(\epsilon^{-3}n\log^{c+3}n \log \frac{1}{\epsilon})$.
\end{theorem}

\noindent Before proving~\Cref{thm:quasi-poly-spread}, we introduce the toolbox used in this section.

\paragraph{Weak net tree.~} Let $\eta\geq 1$ be a constant.  We view $G(V,E,w)$ as a metric space $(V,d_G)$ with  shortest path distances. (We use points and vertices interchangeably.) The spread of $(V,d_G)$ is $\Delta$. Let $\mathcal{H}$ be a hierarchy of nets $V = N_0 \supseteq N_1 \supseteq \ldots \supseteq N_{\lceil \log \Delta \rceil} $ where $N_{i+1}$ is a \emph{weak $(2^{i+1},\eta)$-net} of $N_i$ for each $i \in [1, \lceil \log \Delta \rceil]$.  By~\Cref{thm:weak-net-const}, we can construct $N_{i+1}$ along with an assignment $\mathcal{A}_{i+1}$ that covers $N_i$ in $O(n)$ time.

The hierarchy of nets naturally induces a \emph{$\eta$-weak net tree} $T$ where $i$-th level of $T$ is $N_i$, and the children of each point $p\in N_{i+1}$ are points in $\mathcal{A}_{i+1}(p)\subseteq N_{i}$. 

\begin{lemma}\label{lm:weak-net-const} We can construct an $O(1)$-weak net tree $T$ of $(V,d_G)$ in $O(n \log \Delta)$ time.
\end{lemma}
\begin{proof}
Since $T$ has $O(\log \Delta)$ levels, and the construction of each level can be done in $O(n)$ by~\Cref{thm:weak-net-const},	 the total construction time is $O(n\log \Delta)$.
\end{proof}

\begin{wrapfigure}{r}{0.45\textwidth}
	\vspace{-30pt}
	\begin{center}
		\includegraphics[width=0.45\textwidth]{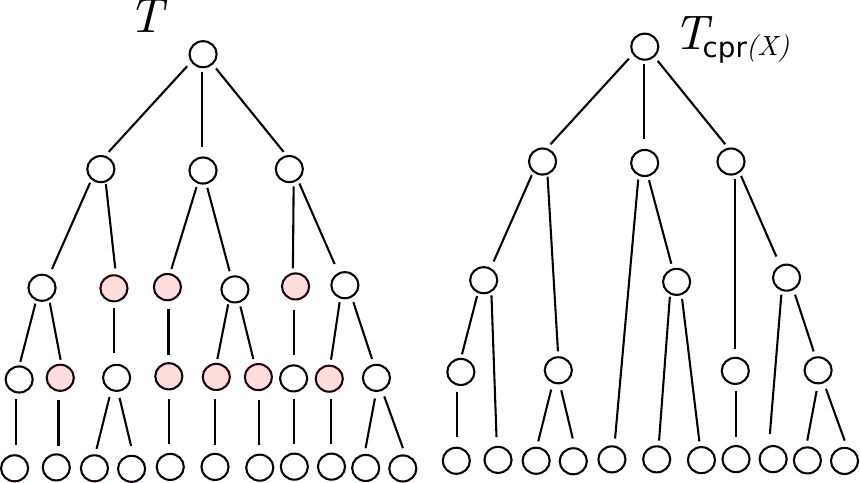}
	\end{center}
	\caption{\footnotesize{$X$ is the set of red points.}}
	\vspace{-5pt}
	\label{fig:compresed-net}
\end{wrapfigure}

Since a point $p$ can appear in many levels of a net tree $T$, to avoid confusion, we sometimes use $(p,i)$ to refer to the copy of $p$ at level $i$.

One operation that will be very useful in our construction is querying an ancestor of a given leaf at  a given level. Such queries can be done in $O(1)$ time using a level ancestor data structure with space $O(|V(T)|)$. However, it could be that the  number of nodes in $T$ is superlinear in $n$; note that $T$ only has $n$ leaves. Thus, it is more space-efficient to work with a \emph{compressed} version of $T$ that compresses nodes of degree $2$ in $T$. For a technical reason that will  be apparent later, we will not compress all degree-$2$ nodes but a subset of them.

\paragraph{$X$-compressed net tree.~} Given a  weak net tree $T$ and a subset of degree $2$ nodes $X$ in $T$, an \emph{ $X$-compressed  net tree}, denoted by $T_{\cpr(X)}$, is an edge-weighted tree obtained by sequentially contracting, in an arbitrary order, each vertex of $X$ to one of its neighbors (see~\Cref{fig:compresed-net}). That is, we replace any monotone maximal path of $T$, whose internal vertices are in $X$ only, with an edge between its endpoints. The weight of each new edge $(x,y)$ of $T_{\cpr(X)}$ is the distance between its endpoints in $T$.  Observe that the weight of every edge in $T_{\cpr(X)}$ is at most $\lceil \log \Delta \rceil$.  Note that we still label each vertex $p\in T_{\cpr(X)}$ with its level in $T$.

Since the compressed net tree has weights on its edges, we need a data structure to query the weighted level ancestor (WLA). In general, querying level ancestors in a weighted tree is a generalization of the predecessor search problem~\cite{KL07} that we cannot hope to have a data structure with linear space and constant query time. On the other hand, we show that it is possible to construct such a WLA data structure if the \emph{hop depth} (defined below) of the tree is polylogarithmic.

\paragraph{Weighted level ancestor data structures.~}  Let $T$ be an \emph{edge-weighted} tree with $n$ vertices rooted at $r$ where every edge $e\in T$ is assigned an integral weight $\omega(e)$.  The \emph{depth} of a node $u \in T$ is the distance $d_T(u,r)$. The depth of $T$ is the maximum depth over all vertices of $T$. The \emph{hop depth} of $u$ is the number of edges on the path from $u$ to $r$, and the hop depth of $T$ is the maximum hop depth over all vertices in $T$. A WLA data structure is a data structure that, given a query of the form $(u,d)$ where  $u \in V(T)$ and  $d \in \mathbb{Z}^+$, returns the \emph{lowest ancestor} of $u$ at depth at most $d$. (It could be possible that there is no ancestor of $u$ at depth exactly $d$.) 
In \Cref{subsec:WLA}, we construct a WLA data structure as stated in \Cref{lm:LA-weighted} below.

\begin{restatable}{lemma}{WLA}
	\label{lm:LA-weighted} Given a rooted, edge-weighted tree $T$ with $n$ vertices and hop depth $\polylog(n)$, there is an algorithm that runs in $O(n)$ time and constructs a level ancestor data structure with $O(n)$ space and $O(1)$ query time. 
\end{restatable}

We now have all necessary tools to prove~\Cref{thm:quasi-poly-spread}.  Let $(V,d_G)$ be the shortest path metric of the input planar graph $G$. Let $T$ be a $\eta$-weak net tree of $(V,d_G)$ with $\eta = O(1)$. We define a parameter $\tau$ as follows:
\begin{equation}\label{eq:def-tau}
\tau ~=~  (\frac{8}{\epsilon} + 12)\eta
\end{equation}

 For technical convenience that will be elaborated later, we extend the net tree $T$ to include \emph{negative levels}:
\begin{equation}
 N_{-\lfloor \log(\tau) \rfloor -1} = N_{-\lfloor \log(\tau) \rfloor} = \ldots = N_{-1} = N_0 = V 
\end{equation}
where we can still interpret each $N_{i}$ as a $(2^{i},\eta)$-net of $N_{i-1}$ when $-\lfloor \log(\tau) \rfloor \leq i \leq 0$. Note that the minimum pairwise distance is $1$. 

 Let $N_i$ be the weak $(2^i,\eta)$-net associated with  $i$-the level of $T$ for some $i \in [-\lfloor \log(\tau) \rfloor -1 , \lceil \log \Delta \rceil]$. Let:
\begin{equation}\label{eq:net-oracle}
N^{\tau}_i = \{v| v \in N_i \wedge (\exists u\not= v \in N_i, d_G(u,v) \leq \tau 2^{i})\}
\end{equation}  
That is, $N^{\tau}_i$ is the set of net points in $N_i$ that have at least one other net point within distance $\tau 2^{i}$.  While the set $N^{\tau}_i$ has several interesting properties that can be exploited to construct our distance oracle, it is unclear how to compute $N^{\tau}_i$ efficiently without considering distances between all pairs of points in $N_i$, which could costs $\Omega(n^2)$ time. We instead consider a bigger set $N_i^{\tau,+}$ defined below that can be computed in $O(n)$ time. (Using $N_i^{\tau,+}$ instead of $N_i^{\tau}$ makes the argument for space bound somewhat more complicated, but the bound we get remains the same.) When $i = -\lfloor \log(\tau) \rfloor -1$,  $N^{\tau}_i = \emptyset$. 

\begin{quote}
	\textbf{Construct $N_i^{\tau,+}$:~}
	Let $\mathcal{C}_i$ be a  $(\beta, s, \beta\tau 2^i)$-sparse cover of $G$ with $\beta = s  = O(1)$. For each set $C\in \mathcal{C}_i$, if $|C\cap N_i|\leq 1$, we remove $C$ from $\mathcal{C}_i$. Let $\mathcal{C}_i^{-}$ be the resulting cover. 	 We then define $N_i^{\tau,+}$ to be the set of all points $v \in N_i$ such that there exists $C\in \mathcal{C}_i^{-}$ containing $v$.
\end{quote}

\begin{lemma}\label{lm:NtauPlus}$N_i^{\tau,+}$ satisfies three properties: (a) 	$N_i^{\tau,+}$ is computable in $O(n)$ time, (b) $N_i^{\tau}\subseteq N_i^{\tau,+}$ and (c) for each $v\in N_i^{\tau,+}$, there exists $u\not= v\in N_i^{\tau,+}$ such that $d_G(u,v)\leq \beta \tau 2^{i}$.
\end{lemma}
\begin{proof}
	(a) By~\Cref{lm:sparsecover-time}, $\mathcal{C}_i$ can be computed in $O(n)$ time. Thus, the total time to construct $\mathcal{C}^{-}_i$ is $\sum_{C\in \mathcal{C}_i}O(|V(C)|)=  O(s\cdot n) = O(n)$. We then compute $N_i^{\tau,+}$ by considering each cluster $C \in  \mathcal{C}_i^{-}$ and adding to $N_i^{\tau,+}$ any point $v \in N_i\cap C$; this can be done in $\sum_{C\in \mathcal{C}_i^{-}}O(|V(C)|) = O(n)$ time.
	
	(b) Let $v$ be a vertex in $N^{\tau}_i$; we want to show that $v\in N_i^{\tau,+}$. By definition, there exists $u\in N_i$ such that $u \not = v$ and $d_G(u,v)\leq \tau 2^i$. That is, $u\in B_G(v,\tau 2^{i})$. By property (2) of sparse covers, $B_G(v, \tau 2^i) \subseteq C$ for some cluster $C\in \mathcal{C}_i$. That is, $\{u,v\}\subseteq V(C)$ and hence $C$ is in $\mathcal{C}_i^{-}$. By definition,  $v$ is in $N^{\tau, +}_i$ as desired.
	
	(c) By definition, if $v\in N_i^{\tau,+}$, there is a cluster $C\in \mathcal{C}_i^{-}$ containing $v$. Since $|C\cap N_i|\geq 2$ by construction, there exists $u\in (C\cap N_i)\setminus \{v\}$. As the diameter of $C$ is at most $ \beta\tau 2^i$, $d_G(u,v)\leq \beta\tau 2^i$, as desired.
\end{proof}

\paragraph{Distance oracle construction.~}  Let $T_2$ be the set of all degree-2 vertices in $T$. The oracle $\mathcal{D}$ consists of:

 \begin{enumerate}[noitemsep]
	\item Oracles $\{\mathcal{D}_i\}_{i = -\lfloor \log(\tau) \rfloor -1}^{\lceil \log \Delta\rceil}$ where $\mathcal{D}_i$ is the distance oracle for $N_i^{\tau,+}$ constructed by applying~\Cref{thm:restricted-oracle} with $d = (\tau/2-2\eta)2^{i}$ and $\alpha = \frac{2\tau}{\tau-4\eta}$. 

	\item A constant stretch oracle $\mathcal{D}_c$ for $G(V,E,w)$ constructed by applying~\Cref{thm:constant-stretch}; $\spc(\md_c) = O(n)$.
	\item An $X$-compressed net tree $T_{\cpr(X)}$ with 
	\begin{equation}\label{eq:X-compress}
	X = T_2 \bigcap (V(T) \setminus (\cup_{i=-\lfloor \log(\tau) \rfloor -1}^{\lceil \log \Delta\rceil }N_i^{\tau,+})
	\end{equation}
	We store at each node of $T_{\cpr(X)}$ its depth in the tree.
	\item A weighted level ancestor data structure $\mathcal{W}$ for $T_{\cpr(X)}$ by~\Cref{lm:LA-weighted}. 
\end{enumerate}

\paragraph{Oracle query.~}  Given a query pair $(u,v)$, $\mathcal{D}_c$ returns $d_{\mathcal{D}_c}(u,v)$ such that $d_G(u,v)\leq d_{\mathcal{D}_c}(u,v) \leq 5\cdot d_G(u,v)$ by~\Cref{thm:constant-stretch}. We define:
\begin{equation}\label{eq:def-bar-i}
\bar{i} = \lceil \log_2 \frac{2(1+\epsilon)d_{\mathcal{D}_c}(u,v)}{\tau-4\eta}\rceil
\end{equation}

Since $\log(\Delta) = \polylog(n)$, the hop depth of $T_{\cpr(X)}$ is $\polylog(n)$. Thus, each level ancestor query in $\mathcal{W}$ can be answered in $O(1)$ time by~\Cref{lm:LA-weighted} if we assume that $\bar i$ can be obtained from $d_{\mathcal{D}_c}(u,v)$ in $O(1)$ time. This assumption requires justification since we are not assuming that the logarithm nor the ceiling function can be computed in $O(1)$ time. However, we omit the details here since we will focus on essentially the same problem in Section~\ref{subsec:linear-oracle}.

For each $j \in [\bar{i}-5,\bar{i}]$, in $O(1)$ time, we query the ancestors $p_j(u)$ and $p_j(v)$ in $T_{\cpr(X)}$ at level $j$, or equivalently, at depth $\lceil \log \Delta \rceil + \lfloor \log(\tau) \rfloor + 1 - j$, of $u$ and $v$, respectively, using $\mathcal{W}$.  We then query the distance between $p_j(u)$ and $p_j(v)$ using oracle  $\mathcal{D}_j$ in $O(\alpha^2\epsilon^{-2}) = O(\epsilon^{-2})$ time. Here we use the fact that:

\begin{equation}\label{eq:alpha-up}
\alpha ~=~ \frac{2\tau}{\tau-4\eta} ~ \stackrel{\mbox{\tiny{Eq.~\ref{eq:def-tau}}}}{=} ~ \frac{2(8/\epsilon + 12)}{8/\epsilon + 8} ~=~ \frac{3\epsilon + 2}{\epsilon + 1} ~\leq~ 3.
\end{equation}
Finally, we return: 
\begin{equation}\label{eq:uv-approximate}
d_{\mathcal{D}}(u,v) \stackrel{\mathrm{def}}{=} \min_{j\in [\bar{i}-5,\bar{i}]}(d_{\mathcal{D}_j}(p_j(u), p_j(v)) + \eta 2^{j+2})
\end{equation}

Using a lookup table to precompute powers of $2$, we then get a total query time of $O(\epsilon^{-2})$. We note that there could be possible that the ancestors returned by $\mathcal{W}$ are not $p_j(u)$ and/or $p_j(v)$ because they may be compressed in  $(T_{\cpr(X)}, \varphi_{\cpr(X)})$; we can easily check if this is the case by comparing the depth of the returned ancestors and the desired depth  in $T$, which is $\lceil \log \Delta \rceil + \lfloor \log(\tau) \rfloor + 1 - j$. In this case, we will exclude $j$ in computing the approximate distance in~\Cref{eq:uv-approximate}.

To complete the proof of~\Cref{thm:quasi-poly-spread}, it remains to analyze the space used by the oracle (\Cref{subsec:Qasi-space}), stretch (\Cref{subsec:Qasi-time}), and preprocessing time (\Cref{subsec:Qasi-prepro}).

\subsubsection{Space}\label{subsec:Qasi-space}

First, we bound the total space of  $\{\mathcal{D}_i\}_{i = -\lfloor \log(\tau) \rfloor -1}^{\lceil \log \Delta\rceil}$. Our proof is inspired by the amortized analysis by Chan et al.~\cite{CGMZ16} to bound the number of edges of $(1+\epsilon)$-spanners in doubling metrics (see Lemma 5.10 in~\cite{CGMZ16}). In their proof, they heavily used the packing property of doubling metrics -- which is not available in our setting -- and edge orientation to bound \emph{in-degree} of each net point. We instead carefully construct a forest of net points and base the charging argument on the structure of the forest.

\begin{wrapfigure}{r}{0.5\textwidth}
	\vspace{-20pt}
	\begin{center}
		\includegraphics[width=0.5\textwidth]{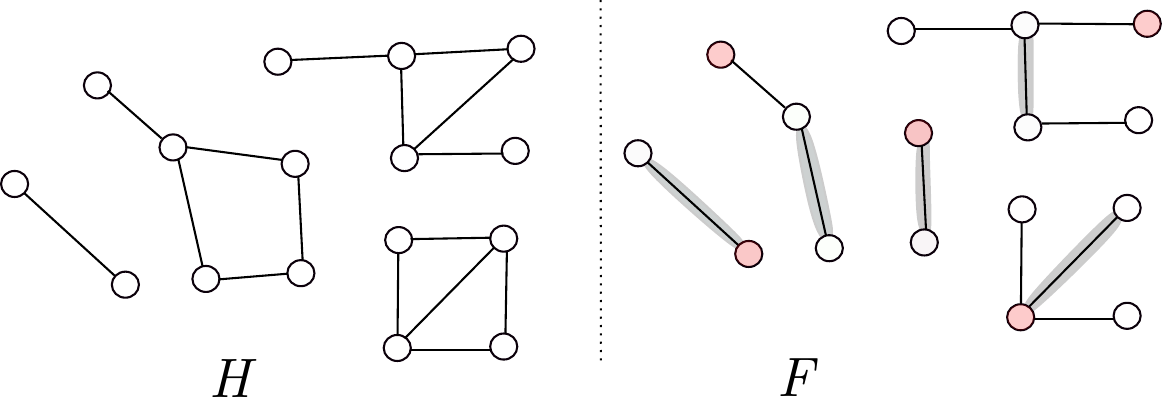}
	\end{center}
	\caption{\footnotesize{Shaded edges in $F$ are edges in a maximal matching}. Leader vertices are highlighted in red.}
	\vspace{-25pt}
	\label{fig:charging}
\end{wrapfigure}

\begin{lemma}\label{lm:total-space-netoracle}
	$ \sum_{i=-\lfloor \log(\tau) \rfloor -1 }^{\lceil \log \Delta \rceil} \spc(\mathcal{D}_i) = O(n\epsilon^{-2}\log\frac{1}{\epsilon})$.
\end{lemma}
\begin{proof}
	Recall that $\alpha \leq  3$ by~\Cref{eq:alpha-up}.  Observe by~\Cref{thm:restricted-oracle} that  each $\mathcal{D}_i$ has space $O(|N_i^{\tau,+}|\epsilon^{-2}\alpha^2) = O(|N_{i}^{\tau,+}|\epsilon^{-2})$. However, in the worst case, one net point can appear in $\Omega(\lceil \log \Delta \rceil)$ levels of $T$. To show a linear bound, we use an amortized argument.  For each point $v \in V$, let $i^{*}(v)$ be the highest level of $v$ in the net tree $T$.

	We construct an \emph{unweighted} graph $H_i$ in which $V(H_i) = N_i^{\tau,+}$ and there is an edge between two vertices $u,v$ of $H_i$ if  and only if $d_G(u,v)\leq \beta\tau 2^{i}$. By Item (c) in~\Cref{lm:NtauPlus}, $H_i$ has no isolated vertex. Let $F$ be a spanning forest of $H_i$ such that every tree in $F$ has at least two vertices and diameter at most $3$ (see~\Cref{fig:charging}). $F$ can be constructed by  choosing a maximal matching of $H_i$ and then assigning each unmatched vertex to an (arbitrary) matched neighbor.  For each tree $X$ in $F$, we designate a vertex $x  = \argmax_{v\in X} i^*(v)$ as the leader; ties are broken arbitrarily and consistently. Let $Y$ be the set of leaders in $F$. Observe that $|Y| \leq |N_i^{\tau,+}|/2$ since $H_i$ has no isolated vertices. We charge the space cost of $\mathcal{D}_i$ equally to all non-leader vertices; these are in $N_i^{\tau,+}\setminus Y$. The number of words that each vertex is charged to is:
	\begin{equation}\label{eq:charged-cost}
	\frac{O(N_{i}^{\tau,+} \epsilon^{-2})}{|N_i^{\tau,+}\setminus Y|} \leq  \frac{O(N_{i}^{\tau,+} \epsilon^{-2})}{|N_i^{\tau,+}|/2} = O(\epsilon^{-2}) 
	\end{equation}
	
	To complete the proof, it remains to show that each vertex is charged at most $O(\log \frac{1}{\epsilon})$ times. Suppose that $i$ is the first level that a vertex $v$ is charged to. Then $v$ belongs to some tree $X$ of $F$ whose leader $x$ has $i^{*}(x) \geq i^{*}(v)$. Since $X$ has diameter at most $3$, $d_G(x,v)\leq 3\beta\tau 2^{i}$. Thus, at level $i' = i + \lceil \log(3\beta \tau) \rceil+1$, $d_G(x,v) <  2^{i'}$ and hence, at most one of them ``survives" at level $i'$. That is, $|N_{i'}\cap \{x,v\}| \leq 1$. Thus, $i^*(v) \leq i + \lceil \log(3\beta\tau) \rceil+1$. This implies that the number of levels that  $v$ will be charged to is at most $ \lceil \log(3\beta\tau) \rceil+2 = O(\log \frac{1}{\epsilon} + \log(\eta) + \log(\beta)) = O(\log \frac{1}{\epsilon})$ since $\eta = O(1)$ and $\beta = O(1)$.
\end{proof}

\noindent Using the same amortized argument, we obtain the following  corollary of~\Cref{lm:total-space-netoracle}.

\begin{corollary}\label{cor:net-size-all} $\sum_{i=-\lfloor \log(\tau) \rfloor -1 }^{\lceil \log \Delta \rceil} |N_{i}^{\tau,+}| = O(n \log\frac{1}{\epsilon})$.
\end{corollary}

Next, we bound the size of $T_{\cpr(X)}$ and $\mathcal{W}$. 

\begin{lemma}\label{lm:size-compression} $\spc(T_{\cpr(X)})  + \spc(\mathcal{W})  =  O(n\log\frac{1}{\epsilon})$.
\end{lemma}
\begin{proof}
	Observe that $T$ has $n$ leaves. Thus, the number of vertices of degree at least $3$ in $T$ is bounded by $n$. By~\Cref{eq:X-compress} and~\Cref{cor:net-size-all}, the number of degree-2 nodes in $T_{\cpr(X)}$ is at most $|\cup_{i=-\lfloor \log(\tau) \rfloor -1}^{\lceil \log \Delta\rceil }N_i^{\tau,+}| = O(n\log \frac{1}{\epsilon})$; thus $\spc(T_{\cpr(X)})  = O(n \log\frac{1}{\epsilon})$. By ~\Cref{lm:LA-weighted}, $\spc(\mathcal{W}) = O(n \log\frac{1}{\epsilon})$; this implies the lemma.
\end{proof}

By~\Cref{lm:total-space-netoracle} and~\Cref{lm:size-compression},  $\spc(\md) = O(n\eps^{-2}\log \frac{1}{\epsilon})$ as claimed in~\Cref{thm:quasi-poly-spread}. 

\subsubsection{Stretch}\label{subsec:Qasi-time}

Before we analyze the stretch of the distance returned by $\mathcal{D}$, we present several important properties of the weak net tree $T$.

\begin{wrapfigure}{r}{0.45\textwidth}
	\vspace{-20pt}
	\begin{center}
		\includegraphics[width=0.45\textwidth]{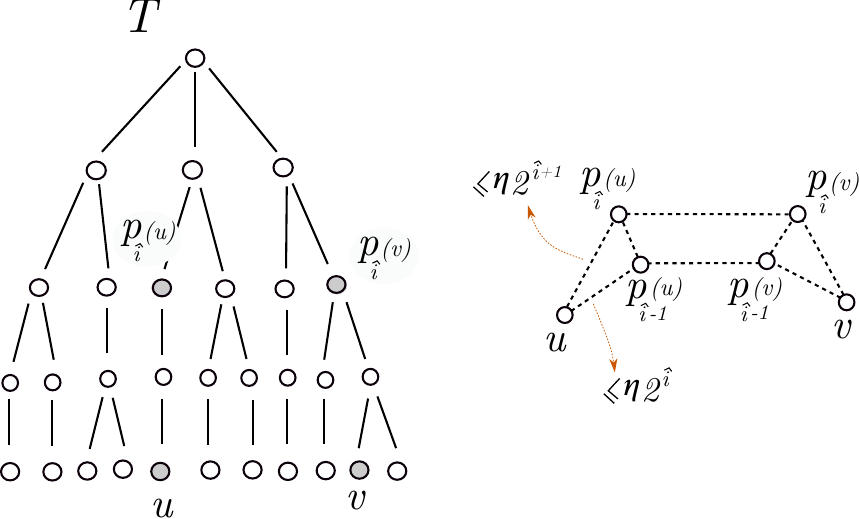}
	\end{center}
	\caption{\footnotesize{Illustration for the proof of ~\Cref{lm:parent-query}.}}
	\vspace{-40pt}
	\label{fig:parent-dist}
\end{wrapfigure}

For each vertex $v\in V$, let $p_i(v)$ be the ancestor of $v$ at level $i$ of the net tree $T$. We have:

\begin{claim}\label{clm:dist-v-parent}$d_G(v,p_i(v)) \leq \eta 2^{i+1}$.
\end{claim}
\begin{proof}
	Let $v = v_{-\lfloor \log(\tau) \rfloor -1},v_{-\lfloor \log(\tau) \rfloor},\ldots, v_{i}  = p_{i}(v)$ be a path from $v$ to $p_{i}(v)$ in $T$. Observe that $d(v_j,v_{j-1}) \leq \eta 2^j$ since $N_j$ is a weak  $(2^j,\eta)$-net of $N_{j-1}$ for every $j\in [-\lfloor \log(\tau) \rfloor,i]$. Thus, $$d_G(v,p_{i}(v)) ~\leq~ \sum_{j=-\lfloor \log(\tau) \rfloor}^{i} \eta 2^j ~\leq~ \eta 2^{i + 1},$$ 
	as desired.
\end{proof}

Next we show that the distance between the ancestors of $u$ and $v$ at some level $\hat{i}$ is close to $d_G(u,v)$.

\begin{lemma}\label{lm:parent-query}
	Given any two vertices $u\not= v \in V$, denote by $\hat{i}$ the lowest level where $d_G(p_{\hat{i}}(u),p_{\hat{i}}(v))\leq 2^{\hat{i}}\tau$.	Then
	\begin{equation}
	(1-\epsilon)d_G(u,v)\leq d_G(p_{\hat{i}}(u),p_{\hat{i}}(v)) \leq (1+\epsilon)d_G(u,v)
	\end{equation}
	Furthermore, $d_G(p_{\hat{i}}(u),p_{\hat{i}}(v))\geq (\tau/2-2\eta)2^{\hat{i}}$.
\end{lemma}
\begin{proof}
	(See~\Cref{fig:parent-dist}.) Recall that $\tau = (\frac{8}{\epsilon}+12)\eta$ by~\Cref{eq:def-tau}. By~\Cref{clm:dist-v-parent} and the triangle inequality, we have:
	\begin{equation}\label{eq:uv-parent-triangle}
	d_G(u,v) - 4\eta\cdot 2^{\hat{i}} \leq d_G(p_{\hat{i}}(u),p_{\hat{i}}(v))\leq d_G(u,v) + 4\eta\cdot 2^{\hat{i}} 
	\end{equation}
	By the minimality of $\hat{i}$, $d_G(p_{\hat{i}-1}(u), p_{\hat{i}-1}(v))> \tau 2^{\hat{i}-1}$. By the triangle inquality,  
	\begin{equation*}
	d_G(p_{\hat{i}}(u),p_{\hat{i}}(v)) ~\geq~ d_G(p_{\hat{i}-1}(u),p_{\hat{i}-1}(v)) - 2\eta 2^{\hat{i}} ~>~\tau 2^{\hat{i}-1}  - 2\eta\cdot 2^{\hat{i}} ~=~ (\tau/2-2\eta)2^{\hat{i}}.
	\end{equation*}
	Next, we show that $4\eta 2^{\hat{i}}\leq \epsilon d_G(u,v)$; this will imply the lemma.   By~\Cref{eq:uv-parent-triangle},
	\begin{equation*}
	d_G(u,v) \geq d_G(p_{\hat{i}}(u),p_{\hat{i}}(v)) - 4\eta \cdot 2^{\hat{i}} \geq (\tau/2 - 6\eta)2^{\hat{i}}  = \frac{4\eta\cdot 2^{\hat{i}}}{\epsilon}
	\end{equation*}
	as desired. 
\end{proof}

We are now ready to show that the stretch of $\mathcal{D}$ is $(1 + 4\epsilon)$; we can recover stretch $(1+\epsilon)$ by scaling $\epsilon \leftarrow \epsilon/4$ at the cost of a constant factor in the query time and the space bound.

\begin{lemma}\label{lm:qasi-stretch} The distance $d_{\mathcal{D}}(u,v)$ returned by oracle $\mathcal{D}$ for a query pair $(u,v)$ satisfies:
	\begin{equation*}
	d_G(u,v)\leq d_{\mathcal{D}}(u,v)\leq (1+4\epsilon) d_G(u,v)
	\end{equation*}
\end{lemma}
\begin{proof}
Let $\hat{i}$ be the lowest level  of the net tree $T$ where $d_G(p_{\hat{i}}(u),p_{\hat{i}}(v))\leq 2^{\hat{i}}\tau$.  Recall that $\bar{i}$ is a parameter defined in~\Cref{eq:def-bar-i}. In the following claim, we show that $\bar{i}$ is close to $\hat{i}$.

\begin{restatable}{claim}{IBarVSIHat}\label{clm:ibar-vs-ihat}
	$\hat{i} \leq  \bar{i} \leq \hat{i}+ 5$ when $\eps\leq \frac{1}{3}$.
\end{restatable}
\begin{proof}
	By~\Cref{lm:parent-query}, $d_G(p_{\hat{i}}(u),p_{\hat{i}}(v))\geq (\tau/2 -2\eta) 2^{\hat{i}}$. Observe that:
	\begin{equation}\label{eq:ihat-vs-ibar-lb}
	d_{\mathcal{D}_c}(u,v)~\geq~ d_{G}(u,v)~\overset{\mbox{\footnotesize{\Cref{lm:parent-query}}}}{\geq}~ \frac{d_G(p_{\hat{i}}(u),p_{\hat{i}}(v))}{1+\epsilon} ~\geq~ \frac{(\tau/2 -2\eta) 2^{\hat{i}}}{1+\eps}
	\end{equation} 
	which implies $\bar{i}\geq \hat{i}$ by the definition in~\Cref{eq:def-bar-i}. Since $d_G(p_{\hat{i}}(u),p_{\hat{i}}(v))\leq \tau 2^{\hat{i}}$ by definition, it holds that:
	\begin{equation*}
	d_{\mathcal{D}_c}(u,v)~\overset{\mbox{\footnotesize{\Cref{thm:constant-stretch}}}}{\leq} ~ 5d_{G}(u,v)~\overset{\mbox{\footnotesize{\Cref{lm:parent-query}}}}{\leq}~ \frac{5d_G(p_{\hat{i}}(u),p_{\hat{i}}(v))}{1-\epsilon} ~\leq~\frac{5\tau 2^{\hat{i}}}{1-\epsilon}
	\end{equation*} 
	This implies:
	\begin{equation*}
	\begin{split}
	\bar{i} &\stackrel{\mbox{\tiny{Eq.~\ref{eq:def-bar-i}}}}{=} \lceil \log_2 \frac{2(1+\epsilon)d_{\mathcal{D}_c}(u,v)}{\tau-4\eta}\rceil  ~\leq~  \lceil \log_2 \frac{10(1+\epsilon)\tau 2^{\hat{i}}}{(1-\epsilon)(\tau-4\eta)}\rceil\\
	&= \hat{i} + \lceil\log_2  \frac{10(1+\eps)(8/\epsilon+12)}{(1-\epsilon)(8/\epsilon + 8)}\rceil ~\leq~ \hat{i}+5
	\end{split}
	\end{equation*}
	when $\epsilon \leq \frac{1}{3}$.
\end{proof}

By~\Cref{clm:ibar-vs-ihat}, if we consider every $j\in [\bar{i} - 5,\bar{i}]$, there exists a value of $j$ such that $j = \hat{i}$. Thus, to show that the stretch is $1+4\epsilon$, it suffices to show the following:
\begin{itemize}[noitemsep]
	\item (a) for \emph{every} $j \in [\bar{i} - 5,\bar{i}], d_{\mathcal{D}_j}(p_j(u), p_j(v)) + \eta 2^{j+2} \geq d_G(u,v)$. This implies that the distance returned by $\md$ is always at least $d_G(u,v)$.
	\item (b) $d_{\mathcal{D}_{\hat{i}}}(p_{\hat{i}}(u), p_{\hat{i}}(v)) + \eta 2^{\hat{i}+2} \leq (1+4\epsilon)d_G(u,v)$. This implies that the minimum in~\Cref{eq:uv-approximate} is at most $(1+4\epsilon)d_G(u,v)$.
\end{itemize}

\noindent First, we show (a). Consider any $j \in [\bar{i}-5,\bar{i}]$. By the triangle inequality, we have that:
\begin{equation*}
d_{\mathcal{D}_j}(p_j(u), p_j(v)) + \eta 2^{j+2} \stackrel{\mbox{\footnotesize{\Cref{clm:dist-v-parent}}}}{\geq} d_G(p_j(u), p_j(v)) + d_G(u,p_j(u)) +  d_G(v,p_j(v))  \geq d_G(u,v)
\end{equation*}
Now we show (b). By~\Cref{lm:parent-query}, 
\begin{equation}\label{eq:uv-vs-2i}
\begin{split}
d_G(u,v)\geq \frac{d_G(p_{\hat{i}}(u), p_{\hat{i}}(v))}{1+\epsilon}  \geq \frac{(\tau/2-2\eta)}{1+\epsilon} 2^{\hat{i}} = \frac{\eta 2^{\hat{i}+2}}{\epsilon}
\end{split}
\end{equation}
Next, we claim that:
\begin{claim}\label{clm:oracle-return}
	The distance  $d_{\mathcal{D}_{\hat{i}}}(p_{\hat{i}}(u), p_{\hat{i}}(v))$ returned by oracle $\mathcal{D}_{\hat{i}}$ satisfies:
	\begin{equation*}
   d_{\mathcal{D}_{\hat{i}}}(p_{\hat{i}}(u), p_{\hat{i}}(v)) \leq (1+\eps)d_G(p_{\hat{i}}(u), p_{\hat{i}}(v)) 
	\end{equation*}
\end{claim}
\begin{proof}
	First, we observe by definition of $\hat{i}$ that $d_G(p_{\hat{i}}(u), p_{\hat{i}}(v)) \leq \tau 2^{\hat{i}}$. Thus, $\{p_{\hat{i}}(u), p_{\hat{i}}(v)\}\subseteq N^{\tau}_{\hat{i}}$. By~\Cref{lm:NtauPlus}, $\{p_{\hat{i}}(u), p_{\hat{i}}(v)\} \subseteq N^{\tau,+}_{\hat{i}}$. That is, querying the distance  in $\mathcal{D}_{\hat{i}}$ between  $p_{\hat{i}}(u)$ and $p_{\hat{i}}(v)$   is a valid query; however, the distance returned by  $\mathcal{D}_{\hat{i}}$ may have a large multiplicative stretch.
	 
	To show the (multiplicative) stretch of $\mathcal{D}_{\hat{i}}$ is $(1+\epsilon)$,   by~\Cref{thm:restricted-oracle}, it suffices to show that $d_G(p_{\hat{i}}(u), p_{\hat{i}}(v)) \in [d,\alpha d]$ where $d = (\tau/2-2\eta)2^{\hat{i}}$ and $\alpha = \frac{2\tau}{\tau - 4\eta}$.	
	By~\Cref{lm:parent-query}, $d_G(p_{\hat{i}}(u), p_{\hat{i}}(v)) \geq d$. By definition of $\hat{i}$, $d_G(p_{\hat{i}}(u), p_{\hat{i}}(v)) \leq \tau 2^{\hat{i}} = \alpha d$ as desired.
\end{proof}

\noindent We observe that:
\begin{equation*}
\begin{split}
d_{\mathcal{D}_{\hat{i}}}(p_{\hat{i}}(u), p_{\hat{i}}(v)) + \eta 2^{\hat{i}+2} & \leq  (1+\epsilon)d_G(p_{\hat{i}}(u), p_{\hat{i}}(v)) + \eta 2^{\hat{i}+2} \qquad \mbox{(by~\Cref{clm:oracle-return})}  \\
&\leq(1+\epsilon)^2 d_G(u,v) + \eta 2^{\hat{i}+2}  \qquad \mbox{(by~\Cref{lm:parent-query})}\\
&\stackrel{\mbox{\tiny{Eq.~\ref{eq:uv-vs-2i}}}}{\leq} (1+\epsilon)^2 d_G(u,v) + \epsilon d_G(u,v) \leq (1+4\epsilon)d_G(u,v),
\end{split}
\end{equation*}
as desired. 
\end{proof}

\subsubsection{Preprocessing Time}\label{subsec:Qasi-prepro}

Let $L \stackrel{\mbox{\tiny{def.}}}{=} (\log \Delta + \log(\tau)) ~=~ O(\log^{c}(n) + \log(\frac{1}{\epsilon}))$ be the number of levels of the weak net tree $T$.  By~\Cref{lm:weak-net-const}, $T$ can be computed in $O(nL)$ time. By~\Cref{lm:NtauPlus}, computing $\mathcal{N} = \{N_i^{\tau,+}\}^{\lceil \log \Delta\rceil }_{i= - \lfloor\tau \rfloor-1}$ also takes $O(nL)$ time. 
 
 Given $\mathcal{N}$, constructing  $\{\md_i\}^{\lceil \log \Delta\rceil }_{i= - \lfloor\tau \rfloor-1}$ takes $O(\epsilon^{-3}\alpha^3n\log^3 n L) = O(\epsilon^{-3}n\log^3 n L)$ time by~\Cref{thm:restricted-oracle}; note that $\alpha \leq 3$. By~\Cref{thm:constant-stretch}, constructing $\md_c$ takes $O(n)$ time.

Since $|V(T)| = O(n L)$, given $\mathcal{N}$, we can construct $X$-compressed net tree $T_{\cpr(X)}$ in time $O(nL)$. The weighted level ancestor data structure $\mathcal{W}$ can be constructed in time $O(|V(T_{\cpr(X)})|) = O(nL)$ by~\Cref{lm:LA-weighted}.  Thus, the total time to construct $\md$ is:

\begin{equation*}
O(\epsilon^{-3}n\log^3 n L) = O(\epsilon^{-3}n\log^{c+3}n \log \frac{1}{\epsilon}),
\end{equation*}
as claimed in~\Cref{thm:quasi-poly-spread}.

\subsubsection{A Weighted Level Ancestor Data Structure}\label{subsec:WLA}

In this section, we construct a weighted level ancestor (WLA) data structure as claimed in \Cref{lm:LA-weighted}; we restate the lemma below.

\WLA*

Our construction combines several ideas that were developed for the WLA problem and the \emph{predecessor search} problem. Specifically, we will use a data structure of P\u{a}tra\c{s}cu and Thorup~\cite{PT06} for predecessor search in sets of polylogarithmic size, a data structure developed (implicitly) in the work of Gawrychowski, Lewenstein, and Nicholson~\cite{GLN14} for WLA queries in trees of logarithmic size, and a decomposition technique of Kopelowitz and Lewenstein~\cite{KL07} to reduce the WLA problem to the predecessor search problem.

Before getting into details, we review the predecessor search problem. In this problem, we are given a set $S$ containing elements from the set $[n]$. We are tasked to construct a data structure that answers the following \emph{predecessor query}\footnote{By sorting $S$ in an increasing order of the elements, each predecessor query can be answered in $O(\log n)$ time using binary search. The goal is to construct a data structure with $O(n)$ space and $o(\log n)$, or ideally $O(1)$, query time.} efficiently: given a number $x \in [n]$, find the largest integer $p_x \in S$, called the predecessor of $x$, such that $p_x \leq x$. Observe that the predecessor search problem is equivalent to the WLA problem on a path graph.  Our data structure uses the following data structure of P\u{a}tra\c{s}cu and Thorup~\cite{PT06} for predecessor search:

\begin{lemma}[P\u{a}tra\c{s}cu-Thorup~\cite{PT06}]\label{lm:PT-Predecessor-Small} Let $S$ be a set of $s$ elements from the set $[n]$ where $s = \polylog(n)$. Suppose that the machine word's size is $\mword = \Omega(\log(n))$. We can construct in $O(s\log(n))$ time a predecessor search data structure $\mathcal{P}_S$ for $S$ with size $O(s)$ that can answer each  predecessor query in $O(1)$ time.
\end{lemma}

We remark that while the construction time was not explicitly mentioned in the work of P\u{a}tra\c{s}cu and Thorup~\cite{PT06}, it can be seen from the construction (Section 5.2 in the full version of~\cite{PT06}) that the running time is $O(s\log n)$, since they only used two data structures that can be constructed in $O(s\log n)$ time: fusion trees~\cite{FW93}  and $B$-trees.

Another data structure that we use is a WLA data structure for trees of polylogarithmic sizes that is implicit in the work of Gawrychowski, Lewenstein, and Nicholson~\cite{GLN14}. For completeness, we include their proof below.

\begin{lemma}[Gawrychowski-Lewenstein-Nicholson~\cite{PT06}]\label{lm:WLA-GLN} Let $T$ be an edge-weighted rooted tree  of size $k = O(\log(n))$. Suppose that the machine word's size is $\mword \geq k$. We can construct in $O(k\log(n))$ time a WLA data structure $\mathcal{W}_T$ for $T$ with size $O(k)$ that can answer each WLA query in $O(1)$ time.
\end{lemma}
\begin{proof} Let $L$ be a list of vertices in $T$ that is sorted in non-decreasing order of the depth. For each vertex $v \in T$, we create a word $B_v$. We mark $i$-th bit of $B_v$ with value $1$ if the $i$-th vertex in $L$ is an ancestor of $v$ and with value $0$ otherwise. We then create a predecessor data structure $\mathcal{P}_{T}$ to perform predecessor search for the depth of vertices in $T$; duplicate depths are removed in  $\mathcal{P}_{T}$. By \Cref{lm:PT-Predecessor-Small}, $\mathcal{P}_T$ has $O(k)$ space, $O(1)$ query time, and $O(k \log n)$ construction time.

Given a WLA query $(v,d)$, we first perform a search in $\mathcal{P}_{T}$ to find a predecessor of $d$ of value $d_i$ and a $i$-th vertex $u_i$ in $L$ that has depth $d_i$. Note that there could be multiple vertices that have the same depth $d_i$, and it suffices for the data structure to return any of these vertices.  Then we look for the largest index $j \leq i$ in $B_v$ that has the $j$-th bit set to $1$. This operation can be done in $O(1)$ time using $O(1)$ bitwise operations\footnote{$x\leftarrow x\&((-1) \ll i)$. \newline\hspace*{1.7em} $x\leftarrow x\&(-x)$.}. Thus, the total query time is $O(1)$, as claimed.
\end{proof}

The last tool we need for our data structure is the \emph{centroid decomposition}.  Given a tree $T$, we denote by $s_T(v)$ the size (number of vertices) of the subtree rooted at $v$. A centroid path of rank $i$ for some integer $i$ is a maximal path containing every vertex $v$ such that $2^i \leq s_T(v) < 2^{i+1}$. For a centroid path $\pi$, we denote by $\rank(\pi)$ the rank of $\pi$ and $\head(\pi)$ the vertex $x\in \pi$ of minimum depth.  Observe by the definition that a centroid path is a \emph{monotone path}: for any two vertices $x,y\in \pi$, either $x$ is an ancestor of $y$ or $y$ is an ancestor of $x$. For each vertex $u \in T$, we say that $\pi$ is an \emph{ancestral centroid path} of $u$ if $\pi$ contains an ancestor of $u$. By definition, the centroid path containing $u$ is also an ancestral centroid path of $u$. The following lemma is folklore.

\begin{lemma}\label{lm:centrod-decomp} Given a tree $T$ with $n$ vertices, we can find in $O(n)$ time a decomposition of $T$ into centroid paths such that every vertex  $u\in T$ has $O(\log n)$ ancestral centroid paths.   
\end{lemma}

\begin{proof}[Proof of \Cref{lm:LA-weighted}] Let $w$ be the machine word size; $w = \Omega(\log n)$. Let $r$ be the root of $T$. For each vertex $v \in T$, let $T_v$ be the  subtree of $T$ rooted at $v$. Let $U$ be the set of vertices $u \in V(T)$ such that $s(u) \geq w+1$ and for every child $v$ of $u$, $s_T(v) \leq w$. Observe by the definition that $|U| \leq \frac{n}{w} = O(\frac{n}{\log n})$. Let $T^+_U$ be the subtree of $T$ induced by $U$ and all ancestors of vertices in $U$; $T^+_U$  has $U$ as the set of leaves. We construct the data structure in four steps:
	
\begin{itemize}
	\item \textbf{Step 1.~} For each vertex $v\in T^+_U$, we store a pointer to an (arbitrary) leaf $u\in U$ in the subtree rooted at $v$ of $T^+_U$. Next,	we construct a centroid decomposition $\mathcal{C}$ for $T^+_U$ as in \Cref{lm:centrod-decomp}. For each (leaf) vertex $u \in U$, we define $H_u = \{\head(\pi): \pi \mbox{ is an ancestral centroid path of }u\}$. We construct a predecessor search data structure $\mathcal{P}_{H_u}$ for $H_u$ with keys being the depths in $T^+_U$ of the vertices. By \Cref{lm:PT-Predecessor-Small} and \Cref{lm:centrod-decomp}, $\mathcal{P}_{H_u}$ has space $O(\log n)$ and support $O(1)$ query time for the predecessor search. Thus, the total space of the predecessor search data structures for all vertices in $U$ is $O(n)$.
	\item \textbf{Step 2.~}  For each centroid path $\pi$ in the centroid decomposition of $T^+_U$, we observe that $\pi$ has $\polylog(n)$ number of vertices since the depth of $T$ is $\polylog(n)$.  We construct a predecessor data structure $\mathcal{P}_{\pi}$ for vertices in $\pi$ with keys being the depths of the vertices. By \Cref{lm:PT-Predecessor-Small}, $\mathcal{P}_{\pi}$ has size $O(|V(\pi)|)$ and supports $O(1)$ query time.  Thus, the total space of all $\{\mathcal{P}_{\pi}\}_{\pi \in \mathcal{C}}$ is $O(n)$.
	
	\item \textbf{Step 3.~} Let $F$ be the forest induced by vertices in $V(T)\setminus V(T^+_U)$. By the definition of $U$, every tree in $F$ has at most $w$ vertices. For each vertex $x\in F$, we denote by $r_x$ the root of the subtree $\bar{T}$ in $F$ containing $x$. We store at $x$ a pointer to $r_x$. For each tree $\bar{T} \in F$, we construct a data structure $\mathcal{W}_{\bar{T}}$ by \Cref{lm:WLA-GLN}; $\mathcal{W}_{\bar{T}}$ has space $O(|V(\bar{T})|)$ and answers each WLA query in $O(1)$ time. The total space of all $\{\mathcal{W}_{\bar{T}}\}_{\bar{T}\in F}$  is $O(n)$.
	\item \textbf{Step 4.~}  We construct an LCA data structure  $\mathcal{LCA}_T$ for $T$ with $O(n)$ space, $O(1)$ query time, and $O(n)$ construction time.
\end{itemize}	

Our final data structure is the union of all data structures constructed in four steps above; the space bound $O(n)$ follows directly from the construction. 

We now focus on answering a query. Let $(v,d)$ be a WLA query. We assume that $d \leq d_T(r,v)$; otherwise, we return $v$. If $v\in F$, we search for a WLA of $v$ in the tree $\bar{T}$ that contains $v$ by giving a query $(v,d-d_T(r,r_v))$ to the data structure  $\mathcal{W}_{\bar{T}}$; the returned ancestor is the lowest ancestor of $v$ with depth at most $d$ in $T$.  If $v\not\in F$, then $v\in T^+_U$. Let $u$ be a leaf vertex in $T^+_U$ that $v$ has a pointer to. Observe that the lowest ancestor of depth at most $d$ of $v$ is also the lowest ancestor of depth at most $d$ of $u$. Thus, we only need to focus on querying the lowest ancestor of depth at most $d$ of $u$. 

We first find the lowest ancestor in $H_u$ of $u$ with depth at most $d$  by querying the data structure $\mathcal{P}_{H_u}$ constructed in Step 1. Let $x$ be the returned vertex; $x = \head(\pi)$ for some centroid path $\pi$. We then find the vertex $y$ whose depth is a predecessor of $d$  by querying the data structure  $\mathcal{P}_{\pi}$ constructed in Step 2. If $y$ is an ancestor of $u$, which can be checked in $O(1)$ time using $\mathcal{LCA}_T$, we return $y$. Otherwise, let $\pi'$ be the ancestral centroid path of $u$ such that $\head(\pi')$ is closest to $\head(\pi)$. Observe that the parent of $\head(\pi')$, denoted by $z$, is a vertex on $\pi$ and an ancestor of $y$. Furthermore, the lowest ancestor of depth at most $d$ of $u$ is $z$. Thus, we return $z$ in this case. In all cases, a WLA query can be answered in $O(1)$ time; this completes the proof of \Cref{lm:LA-weighted}.
\end{proof}

\subsection{Removing the Spread Assumption}\label{subsec:linear-oracle}

In this subsection, we remove the assumption on the spread using the contraction technique of Kawarabayashi, Sommer, and Thorup~\cite{KST13}. The same technique was used in previous results~\cite{GX19,CS19}. The idea is to have for each scale $r \in \{2^0, 2^{1}, \ldots, 2^{\lceil  \log \Delta \rceil}\}$, a graph $G_r$ obtained from $G$ by removing every edge of weight more than $r$ and contracting every edge of weight less than $\frac{r}{n^2}$. Then, a distance oracle is constructed for each $G_r$ and the total space bound (typically of $\Omega(n\log n)$) follows from the observation that each edge $e\in G$ belongs to at most $O(\log n)$ different graphs $G_r$ (for different values of $r$). 

To show a linear space bound, we need the scale $r$ to be bigger, so that each edge $e \in G$ belongs to at most $O(1)$ graphs $G_r$; we naturally choose the scale to be $\{n^{0}, n^{4}, \ldots, n^{4i}, \ldots, n^{\lceil \log_{n^4}\Delta \rceil}\}$. To construct each $G_r$, we apply the same idea: delete every edge of weight more than $n^{4}r$ and contract every edge of weight at most $\frac{r}{n^2}$. It follows directly from the construction that each edge $e$ belongs to at most 2 graphs $G_r$. The issue now is that, while the spread of $G_r$ is polynomial in $n$, it could be exponential in the \emph{number of vertices of $G_r$}. In this case, we use the bit-packing technique that we formalize in the following lemma.

\begin{lemma}\label{lm:small-graphs} Let $G(V,E,w)$ be an undirected and edge-weighted planar graph with $n$ vertices. If the machine word size is $\omega = \Omega(\log n^{3})$, then in $O(\epsilon^{-2}n\log^3 n)$ time, we can construct a $(1+\epsilon)$-approximate distance oracle for $G(V,E,w)$ with $O(n\epsilon^{-1})$ space and $O(\epsilon^{-1})$ query time.
\end{lemma}

\begin{proof}[Proof of~\Cref{lm:small-graphs}]
	We assume that vertices of $G$ are indexed from $\{1,2,\ldots, n\}$ . We also assume that $\epsilon \geq \frac{1}{n}$; otherwise, we can just store all the pairwise distances to achieve the claimed bounds. In the labeling scheme of Thorup~\cite{Thorup04} and Klein~\cite{Klein02} in~\Cref{thm:Thorup-labeling},  each vertex $u$ must maintain distances to a set $S_u \subseteq V(G)$ of $O((\log n)\epsilon^{-1})$ vertices of $G$. Since each distance accounts for one word, the total number of words is $O(n\epsilon^{-1}\log n)$. 
	
	Instead of maintaining exact distances, we store the \emph{encoding} of the approximate distance from $u$ to each vertex $v\in S_u$ using $O(\log n)$ bits as follows. Let $e_{uv}$ be an edge such that:
	\begin{equation}\label{eq:def-euv}
	\frac{d_G(u,v)}{n} \leq w(e_{uv}) \leq 2d_G(u,v)
	\end{equation}
	Edge $e_{uv}$ exists since the heaviest edge on $\SP_G(u,v)$ satisfies~\Cref{eq:def-euv}. (The problem with the heaviest edge is that it could take up to $\Omega(n)$ time to find.) 
	
	Let $\rho_{u,v} \stackrel{\mbox{\tiny{def.}}}{=} \lceil \frac{d'(u,v)}{ \epsilon w(e_{uv})} \rceil$	where $d'(u,v)$ is the distance computed from the label of $u$ and $v$; $d'(u,v)\leq (1+\epsilon)d_G(u,v)\leq 2d_G(u,v)$ when $\epsilon \leq 1$. 	Observe from~\Cref{eq:def-euv} that $\rho_{u,v} \leq \lceil \frac{d'(u,v)n}{d_G(u,v)\epsilon} \rceil ~\leq~ \lceil \frac{2n}{\epsilon} \rceil$ and thus, it requires $O(\log n)$ bits to store. For each vertex $v\in S_u$, we maintain a  triple $t_{u,v} = (v,\rho_{u,v}, p_{e_{uv}})$ where $p_{e_{uv}}$ is the pointer (of size $O(\log n)$ bits) to the word holding the weight of the edge $e_{uv}$. 
	
	Given $t_{u,v}$, we can compute $\rho_{u,v}(\epsilon/2) w(e_{uv})$ in O(1) time as an approximate distance between $u$ and $v$.  Observe that $\rho_{u,v}(\epsilon/2) w(e_{uv})\geq d_G(u,v)$ and that 
	
	\begin{equation*}
	\rho_{u,v}(\epsilon/2) w(e) \leq (\frac{2d_G(u,v)}{ \epsilon w(e)} +1 )(\epsilon /2)w(e) \leq (1+\epsilon)d_G(u,v). 
	\end{equation*}
	Let $\beta$ be a constant such that the encoding of $t_{u,v}$ is exactly $\beta \log n$ bits, padding 0s if necessary. The total number of bits to represent all triples $\{t_{u,v}\}_{v\in S_u}$ is $O(\epsilon^{-1}\log^2 n) = O(\epsilon^{-1} \omega)$ bits. Thus, one can pack all $\{\rho_{u,v}\}_{v\in S_u}$ into $O(\epsilon^{-1})$ words. We must pack these bits in blocks of length $\beta \log n$ each so that whenever a decoding function $\mathcal{D}$ of the labeling scheme must access the distance $d(u,v)$ in the label of $u$, we can retrieve the block corresponding to $\rho_{u,v}$ in $O(1)$ time using suitable shift and logical-and operations. This can be done by storing $\rho_{u,v}$ in the same order that $v$ appears in the labeling of $u$.   
	
	We now show how to construct our oracle efficiently. By~\Cref{thm:Thorup-labeling}, computing the label set $S_{u}$ for every $u\in V(G)$ takes $O(\epsilon^{-2} n\log^3n)$ time.  In the next claim, we find $e_{uv}$ for every $u\in V(G)$ and every $v\in S_{u}$.
	
	\begin{claim}\label{clm:Findall-euv}  We can find an edge $e_{uv}$ as in~\Cref{eq:def-euv} for every $u\in V(G)$ and every $v\in S_{u}$ in $O(\epsilon^{-1} n\log n)$ time. 
	\end{claim}
	\begin{proof}
		First, we sort all edges of $G$ in increasing order of weight in $O(n\log n)$ time. Then for each pair $(u,v)$ where $v\in S_{u}$ -- there are only $O((\log n)\epsilon^{-1})$ such pairs -- we (a) compute the distance $d'(u,v) \in [d_G(u,v), (1+\epsilon)d_G(u,v)]$ from the labels of $u$ and $v$ in $O(\epsilon^{-1})$ time,  and (b) looking for an edge $e_{uv}$ of maximum weight with $w(e_{uv})\leq d'(u,v)$ in $O(\log n)$ time using binary search. 
		
		Clearly $e_{uv}\leq d'(u,v)\leq 2d_G(u,v)$ when $\epsilon \leq 1$. Let $e_{\max}$ be the heaviest weight edge on $\SP_G(u,v)$; $w(e_{\max})\geq \frac{d_G(u,v)}{n}$. By the maximiality of $e_{uv}$, $w(e_{uv})\geq w(e_{\max})\geq \frac{d_G(u,v)}{n}$ as desired.
	\end{proof}

	By~\Cref{clm:Findall-euv}, computing $e_{u,v}$ for every $u\in V(G)$ and every $v\in S_{u}$ can be done in $O(\epsilon^{-1} n\log n)$ time. To efficiently compute a triple $t_{u,v}$ when given $e_{u,v}$, the non-trivial part is to obtain $\rho_{u,v}$ which involves the ceiling function. However, since $\rho_{u,v}\leq\lceil\frac{2n}\epsilon\rceil$, we can use our earlier approach of using binary search on integers between $1$ and $\lceil 2n/\epsilon\rceil$ to compute $\rho_{u,v}$ in $O(\log (n/\epsilon))$ time which for each pair $(u,v)$ with $v\in S_u$ is $O(n(\log (n/\epsilon))/\epsilon)$.

        Bitpacking takes $O(\epsilon^{-1}\log n)$ time per vertex. This implies the total construction time of our oracle is dominated by the running time to construct the labeling scheme, which is $O(\epsilon^{-2} n\log^3n)$ by~\Cref{thm:Thorup-labeling}.
\end{proof}

\noindent We are now ready to prove~\Cref{thm:main} that we restate below.

\Main*
\begin{proof}
	Assume that the edges of $G$ have weight in $[1,\Delta]$. We also assume that $\epsilon > \frac{1}{n}$; otherwise, the theorem trivially holds. We partition the interval $[1,\Delta]$ into $ k =\lceil \log_{n^4}\Delta \rceil+1$ subintervals $I_0,I_1,\ldots, I_k$ where $I_j  = [n^{4j},n^{(j+1)4})$ for $j\in [0,k]$. Let $L_j = n^{4j}$ and $U_j = n^{4(j+1)}$.  For each interval $I_j$, let $G_j$ be the graph obtained from $G$ by first removing every edge of weight more than $U_j$ from $G$  and contracting all the edges of weight at most $\frac{L_j}{n^2}$.  The spread of $G_j$, denoted by $\Delta_j$,  is $\Delta_j \leq \frac{U_j}{L_j/n^2} = n^{6}$.  
	
	Let $n_j$ be the number of vertices of $G_j$. We consider two cases:
	
	\begin{itemize}[noitemsep]
		\item If $n_j \geq 2^{(\log n)^{\delta}}$ for $\delta = \frac{1}{3}$, then $\log(n) = (\log n_j)^{\frac{1}{\delta}}$ and hence $\Delta_j = 2^{O(\log n_j)^{\frac{1}{\delta}}}$. By~\Cref{thm:quasi-poly-spread}, there is a $(1+\epsilon)$-approximate distance oracle $\mathcal{D}_j$ for $G_j$  with $O(n_j \epsilon^{-2}\log \frac{1}{\epsilon})$ space, $O(\epsilon^{-2})$ query time, and  $O(\epsilon^{-3}n_j\log^{6}n_j \log\frac{1}{\epsilon})$ preprocessing time. 
		\item  If $n_j <  2^{(\log n)^{\delta}}$, then $\log n \geq (\log n_j)^{\frac{1}{\delta}}  = (\log n_j)^{3}$. Thus, we  construct an oracle $\mathcal{D}_j$ with $O(n_j \epsilon^{-1})$ space,  $O(\epsilon^{-1})$ query time, and  $O(\epsilon^{-2}n_j \log^3 n_j)$ preprocessing time by~\Cref{lm:small-graphs}. 
	\end{itemize}

	Let $\mathcal{W}$ be the constant stretch oracle of $G(V,E,w)$ by~\Cref{thm:constant-stretch} with $O(n\log^3 n)$ preprocessing time.	Our oracle $\mathcal{D}$ consists of $\mathcal{W}$ and all oracles $\mathcal{D}_j, j \in [1,k]$. Since each edge of $G$ appears in at most two graphs $G_j$, $\sum_{j=1}^{k} n_j  = O(n)$. Thus, the total space of all oracles is:
	\begin{equation*}
	\spc(\mathcal{D}) =  \spc(\mathcal{W}) +  \sum_{j=1}^k \spc(\mathcal{D}_j) = O(n \epsilon^{-2})
	\end{equation*}
	 
	 \noindent For $j\in [1,k]$ where $n_j = 0$, we can skip past the construction for such $j$ by traversing the list of non-decreasing edge weights of $G$. Thus,	the total preprocessing time is:
	 \begin{equation*}
	 \sum_{j=1}^k\left( O(\epsilon^{-3}n_j\log^{6}n_j \log\frac{1}{\epsilon}) + O(\epsilon^{-2}n_j \log^3 n_j)\right) + O(n\log^3 n) = O(\epsilon^{-3}n\log^6 n \log\frac{1}{\epsilon})
	 \end{equation*}
	 
	To query the distance between any two given vertices $u$ and $v$, we first query $\mathcal{W}$ and compute $\bar{j} = \lceil \log_{n^{4}} d_{\mathcal{W}}(u,v) \rceil$. Let $j_{uv} =  \lceil \log_{n^{4}} d_{G}(u,v) \rceil$. Since $ d_{G}(u,v)\leq  d_{\mathcal{W}}(u,v)\leq 5  d_{G}(u,v)$ by~\Cref{thm:constant-stretch}, $j_{uv} \leq \bar{j}\leq j_{uv}+1$. Thus, we can query the distance between $u$ and $v$ in $G_{\bar{j}}$ and $G_{\bar{j}-1}$ and return:
	\begin{equation}\label{eq:final-query}
		d_{\mathcal{D}}(u,v) = \min_{j \in \{\bar{j}-1,\bar{j}\}} d_{\mathcal{D}_j}(u,v) + \frac{L_j}{n}.
	\end{equation} 
(If $u$ and $v$ are contracted to a single vertex in $G_{j}$, then the oracle simply returns $\infty$; note that they will not be contracted to a single vertex in $G_{j_{uv}}$.) Assuming the index $\bar{j}$ can be computed in $O(1)$ time, the query time is dominated by the query time of $\mathcal{D}_j$, which by ~\Cref{thm:quasi-poly-spread} is bounded by $O(\epsilon^{-2}) = O(\epsilon^{-2})$.

To complete the query time analysis, we need to make some modifications to ensure that $\bar{j} = \lceil \log_{n^{4}} d_{\mathcal{W}}(u,v) \rceil$ can be computed in constant time. We may assume that when ordering the edge weights by non-decreasing value, any two consecutive edge weights $w_i\leq w_{i+1}$ differ by no more than a factor $n^{12}$; this follows since otherwise one of the graphs $G_j$ will contain no edges so we can scale down all edge weights of value at least $w_{i+1}$ by the same value to ensure that $w_{i+1}/w_i\leq n^{12}$ (when answering a query, the distance estimate obtained from a subgraph $G_{j'}$ then needs to be scaled up again by the product of the values that divided its edge weights). Since $G$ is planar, it has no more than $6n$ edges, and we now have $\Delta\in [1,(n^{12})^{6n}] = [1,n^{72n}]$.

For each edge $e$, we precompute the value $x_e\in\mathbb N_0$ as well as $n^{4x_e}$ such that for the actual edge weight $w(e)$, $n^{4x_e}\leq w(e) < n^{4(x_e+1)}$. This can be done in $O(n\log n)$ time as follows. First compute the $\Theta(n)$-length sequence $1,n^4,n^{4\cdot 2},n^{4\cdot 3},\ldots,n^{4\cdot\delta}$ where $\delta = \lceil\log_{n^4}\Delta\rceil$ is the first integer encountered such that $n^{4\cdot\delta}\geq\Delta$. Then apply binary search to this sequence for each edge weight.

We extend the oracle $\mathcal W$ so that on query $(u,v)$, it also returns an edge $e_{uv}$ such that $w(e_{uv})\leq d_{\mathcal W}(u,v)\leq 5w(e_{uv})n$; we show below how to do that. Since $\lceil\log_{n^4}(w(e_{uv}))\rceil\leq\bar j\leq \lceil\log_{n^4}(5w(e_{uv})n)\rceil$, we have $\bar j\in [x_{e_{uv}},x_{e_{uv}}+1]$ so $\bar j$ can be obtained in $O(1)$ time by comparing $d_{\mathcal W}(u,v)$ to the precomputed value $n^{4x_{e_{uv}}}$ and to $n^{4(x_{e_{uv}}+1)} = n^4\cdot n^{4x_{e_{uv}}}$.

We now extend $\mathcal W$ to also report the edge $e_{uv}$. Using the notation in Section~\ref{sec:const-stretch}, recall that $\mathcal W$ outputs $\min\{d_0,d_1,d_2\}$ where for $i = 0,1$,
\[
  d_i = \left\{\begin{array}{ll}d_{P_{i+1}(u)}(u,b_{i+1}(u)) + \tilde d_{P_i(u)}(b_{i+1}(u),b_{i+1}(v)) + d_{P_{i+1}(v)}(v,b_{i+1}(v)) & \mbox{if } P_i(u) = P_i(v)\\
              \infty & \mbox{otherwise,}
\end{array}\right.
\]
and $d_2$ is obtained from the lookup table associated with $P_2(u)$. Since $|P_2(u)| = O(\log\log n)$, a simple brute-force extension of the preprocessing algorithm for the lookup table suffices which precomputes and stores the heaviest edge on a shortest path between each vertex pair of $P_2(u)$; this will not increase the overall preprocessing time. Thus, $e_{uv}$ is reported if $d_2$ is the output of the query to $\mathcal W$.

Now, assume that $d_i$ is output with $i\in\{0,1\}$. We may assume that $P_i(u) = P_i(v)$. It is straightforward to precompute the heaviest edge on the two subpaths of weight $d_{P_{i+1}(u)}(u,b_{i+1}(u))$ resp.~$d_{P_{i+1}(v)}(v,b_{i+1}(v))$ without an increase in preprocessing time since these paths are obtained with Dijkstra's algorithm.

The path corresponding to the term $\tilde d_{P_i(u)}(b_{i+1}(u),b_{i+1}(v))$ above is implicitly represented in the oracle of Thorup~\cite{Thorup04}. It consists of three subpaths, a subpath $Q_1$ from $b_{i+1}(u)$ to a portal $p_1$, a subpath $Q_2$ from $p_1$ to another portal $p_2$, and a subpath $Q_3$ from $p_2$ to $b_{i+1}(v)$. The subpaths $Q_1$ and $Q_3$ are precomputed by Thorup's oracle and since this oracle explicitly stores $d_G(b_{i+1}(u),p_1)$ and $d_G(b_{i+1}(v),p_2)$, it can also precompute and store the heaviest edge weight on each of these paths without an increase in precrocessing time and space. Finally, the subpath $Q_2$ is contained in a root-to-leaf path in a fixed single-source shortest path tree of $G$. Hence, finding the heaviest edge on $Q_2$ corresponds to solving the online path-maxima problem on a static tree which can be done with $O(1)$ query time and $O(n)$ space with $O(n\log n)$ preprocessing time~\cite{productqueries}. Thus we have shown how to obtain in $O(1)$ time the heaviest edge on each of the five subpaths of the approximate shortest $u$-to-$v$ path of weight $d_i = d_{\mathcal W}(u,v)$. $\mathcal W$ now outputs the heaviest of these five edges as $e_{uv}$. Since the approximate shortest $u$-to-$v$ path consists of five simple paths, it contains no more than $5n$ edges and so $w(e_{uv})\leq d_{\mathcal W}(u,v)\leq 5w(e_{uv})n$, as desired.

To bound the stretch, we observe that the additive factor $\frac{L_j}{n}$ in~\Cref{eq:final-query} is to account for the fact that (at most $n-1$) edges of weight at most $L_j/n^2$ are contracted in the construction of $G_j$. That is, $d_{\mathcal{D}}(u,v)\geq d_G(u,v)$. Since $d_G(u,v)\geq L_{j_{uv}}, \frac{L_{j_{uv}}}{n}\leq \epsilon d_G(u,v)$ as $\epsilon \geq \frac{1}{n}$. Thus, $d_{\mathcal{D}}(u,v)\geq d_G(u,v) \leq (1+\epsilon)d_G(u,v) + \epsilon d_G(u,v) = (1+2\epsilon) d_G(u,v)$. One can recover stretch $1+\epsilon'$ by setting $\epsilon' = \epsilon/2$. 
\end{proof}

\section{Approximate Distance Oracle for Digraphs}\label{sec:digraphs}

In this section, we present our distance oracle for planar digraphs. Let us restate \Cref{thm:Digraphs}.

\Digraphs*

It will be useful in the following to assume that edge weights are strictly positive and that the smallest positive edge weight is $1$. The latter can be ensured by simply dividing all edge weights by the minimum positive edge weight. Let $w$ be the edge weight function obtained. To ensure that all edge weights are strictly positive, we introduce a modified edge weight function $w'$. For every edge $e$ of weight $w(e) = 0$, let the modified edge weight be $w'(e) = \frac 1 n$; for every other edge $e$, let $w'(e) = w(e)$. Since a shortest path can always be chosen such that it has less than $n$ edges, a $(1+\epsilon)$-approximate distance oracle for the modified edge weight function $w'$ will give $(1+\epsilon)$-approximate distances w.r.t.~the original edge weight function $w$ by simply rounding down to $0$ if the estimate found is less than $1$. Furthermore, since the ratio between the largest and smallest modified edge weight is $N/(1/n) = Nn$ and since $\lg(Nn) = \Theta(\lg(Nn^2))$, our time and space bounds will not be affected. Hence, assume from now on that edge weights are between $1$ and $N$.

As we will show, the theorem is fairly easily obtained from the lemma below. First we need some definitions. A \emph{$(t,\alpha)$-layered spanning tree} is a disoriented rooted spanning tree of an edge-weighted digraph where every path from the root of the tree is the concatenation of at most $t$ oriented paths each of weight at most $\alpha$. We refer to each of these oriented paths as a \emph{dipath} of $T$. A digraph is said to be \emph{$(t,\alpha)$-layered} if it contains a $(t,\alpha)$-layered spanning tree.

Thorup~\cite{Thorup04} gives a construction of a collection of $(3,\alpha)$-layered digraphs from $G$ and parameter $\alpha$. First $V$ is partitioned into layers $L_0^\alpha,L_1^\alpha,\ldots,L_p^\alpha$ as follows. Pick an arbitrary root vertex $r_0^\alpha$ in $G$ and let $L_0^\alpha$ be the set of vertices reachable from $v_0^\alpha$ within distance $\alpha$ in $G$. For $i = 1,2,\ldots,p$, let $S_i^\alpha = (V_i^\alpha,E_i^\alpha)$ be $G$ with $L_0^\alpha\cup L_1^\alpha\cup\ldots\cup L_{i-1}^\alpha$ contracted to a single vertex $v_i^\alpha$. If $i$ is even then $L_i^\alpha$ consists of the set of vertices of $V_i^\alpha - \{v_i^\alpha\}\subset V$ \emph{reachable from} $v_i^\alpha$ within distance $\alpha$ in $S_i^\alpha$. If $i$ is odd then $L_i^\alpha$ consists of the set of vertices of $S_i^\alpha$ that can \emph{reach} $v_i^\alpha$ within distance $\alpha$ in $S_i^\alpha$. We assume that an index $p$ exists such that $\cup_{i = 0}^p L_i^\alpha = V$; if this is not the case then that must mean that $G$ with all edges of weight greater than $\alpha$ removed consists of more than one connected component when ignoring edge orientations; the construction above is then applied to each subgraph of $G$ corresponding to the connected components.

For $i = 0,1,\ldots,p-2$, let $G_i^\alpha(r_i)$ be the $(3,\alpha)$-layered graph obtained from $G[\cup_{j = 0}^{i+2}L_i^\alpha]$ by contracting $L_0^\alpha\cup L_1^\alpha\cup\ldots\cup L_{i-1}^\alpha$ to a single vertex $r_i^\alpha$. Let $G_i^\alpha = G_i^\alpha(r_i^\alpha) - \{r_i^\alpha\} = G[\cup_{j = i}^{i+2}L_j]$. Since the layers $L_0^\alpha,L_1^\alpha,\ldots,L_p^\alpha$ form a partition of $V$, each vertex of $V$ belongs to at most three graphs $G_i^\alpha$ and hence the total size of all these graphs is $O(n)$; we let $\mathcal G^\alpha$ denote the collection of these graphs. It follows from Thorup that for any two vertices $u$ and $v$ in $G$ with a shortest path $Q_{uv}$ of weight at most $\alpha$ in $G$, there is an index $i$ such that $Q_{uv}$ is contained in $G_i^\alpha$.

Given $\epsilon > 0$ and given one of the subgraphs $H = G_i^{\alpha}\in\mathcal G^\alpha$ of $G$, a \emph{scale-$(\alpha,\epsilon)$ distance oracle} for $H$ is a data structure which for any pair of query vertices $u$ and $v$ in $H$ outputs an approximate distance $\tilde d(u,v)$ such that $d_G(u,v)\leq \tilde d(u,v)$ and such that if $d_H(u,v)\leq\alpha$ then $\tilde d(u,v)\leq d_H(u,v) + \epsilon\alpha$. The lower bound on $\tilde d(u,v)$ is $d_G(u,v)$ and not $d_H(u,v)$ using Thorup's definition. This weaker lower bound turns out to be simpler for us to ensure and it will not break correctness since for a query $(u,v)$ in $G$, we merely need to ensure that the output is at least $d_G(u,v)$.

We are now ready to state our lemma which gives fast and compact scale-$(\alpha,\epsilon)$ distance oracles, allowing us to obtain our approximate distance oracle for $G$.
\begin{lemma}\label{lemma:DigraphsOneScale}
Given $\epsilon > 0$ and parameters $r_1,r_2,\ldots,r_k\in [2,n]$ with $k = O(1)$ and with $r_{i+1}\leq \frac 1 2 r_i$ for $i = 1,2,\ldots,k-1$. Letting $r_0 = n$, there is a scale-$(\alpha,\epsilon)$ distance oracle for each of the $(3,\alpha)$-layered graphs in $\mathcal G^\alpha$. The total space of these oracles is $O(n\lg(r_k)/\lg n + (n/\epsilon)\sum_{i = 0}^{k-1}\lg(r_i)/\sqrt{r_{i+1}})$ plus $O(n)$ space independent of $\alpha$. Each oracle has $O(r_k/\epsilon^{2k-1})$ query time.
\end{lemma}
Before proving this lemma, let us show that it implies Theorem~\ref{thm:Digraphs} (and Corollary~\ref{cor:Digraphs}). We ignore $\epsilon$ for now and focus on getting multiplicative $O(1)$-approximation.

We have $\lceil\lg(Nn)\rceil$ distance scales $\alpha = 2^i$ for $i = 1,2,\ldots,\lceil\lg(Nn)\rceil$. For each such $\alpha$, we have for each $(3,\alpha)$-layered subgraph of $G$ in $\mathcal G^\alpha$ a scale-$(\alpha,\frac 1 2)$ distance oracle. As shown by Thorup, a query for an $O(1)$-approximation of a distance $d_G(u,v)$ can be answered by applying binary search on the distance scales; for the current distance scale, the oracles for the at most three graphs in $\mathcal G^\alpha$ containing both $u$ and $v$ are queried with $(u,v)$ to decide which subsequence of distance scales to recurse on. Thus, only $O(\lg\lg(Nn))$ queries are made to the oracles over all scales. The binary search identifies the smallest $\alpha = 2^i$ such that the estimate $\tilde d^{\alpha,\frac 1 2}(u,v)$ produced by the scale-$(\alpha,\frac 1 2)$ distance oracle satisfies $\tilde d^{\alpha,\frac 1 2}(u,v)\in [\alpha/4,\alpha]$. The following lemma follows from Thorup's analysis:
\begin{lemma}
$\tilde d^{\alpha,\frac 1 2}(u,v)$ exists and is a $4$-approximation of $d_G(u,v)$. Furthermore, $d_G(u,v)\in [\alpha/4,\alpha]$.
\end{lemma}
\begin{proof}
The existence of $\alpha$ follows since for some $j\in\mathbb N$, $d_G(u,v)\in [2^{j-2}, 2^{j-1}]$ and hence $\tilde d^{2^j,\frac 1 2}(u,v)\in [2^{j-2}, 2^{j-1}+2^{j-1}] = [2^{j-2},2^{j}]$.

To show that $\tilde d^{\alpha,\frac 1 2}(u,v)$ is a $4$-approximation of $d_G(u,v)$, it suffices to show that $d_G(u,v)\geq\alpha/4$; this will also imply the last part of the lemma. Assume for contradiction that $\alpha/2^{j+1}\leq d_G(u,v) < \alpha/2^{j}$ for some $j\geq 2$. Then the scale $(\alpha/2^{j-1},\frac 1 2)$ distance oracle would produce an estimate in $[\alpha/2^{j+1},\alpha/2^{j}+\alpha/2^{j}] = [\alpha/2^{j+1},\alpha/2^{j-1}]$, contradicting the minimality of $\alpha$.
\end{proof}

We now describe our choices of oracles at each distance scale.

For $i = 1,2,\ldots,\lceil\lg(Nn)/\lg\lg\lg(Nn)\rceil$ and $\alpha = (\lceil\lg\lg(Nn)\rceil)^i$, we store scale-$(\alpha,\frac 1 2)$ distance oracles of Lemma~\ref{lemma:DigraphsOneScale} for each $(3,\alpha)$-layered subgraph in $\mathcal G^\alpha$; the parameters are $r_i = (\lg^{(i)}n)^2$ for $i = 1,2,\ldots,k-1$ and $r_k = 2$; here, $k = O(1)$ is the parameter of Theorem~\ref{thm:Digraphs}~\footnote{We may assume that $n$ is bigger than some constant to ensure the requirement that $r_{i+1}\leq\frac 1 2 r_i$ for $i = 1,2,\ldots,k-1$ in Lemma~\ref{lemma:DigraphsOneScale}.}. Since the total size of the graphs in $\mathcal G^\alpha$ is $O(n)$ and since the number of scales considered is $O(\lg(Nn)/\lg\lg\lg(Nn))$, the total space for the oracles over all these scales is
\[
O\left(\frac{\lg(Nn)}{\lg\lg\lg(Nn)}\cdot n\cdot\left(\sum_{i = 0}^{k-1}\frac{\lg r_i}{\sqrt{r_{i+1}}}\right)\right) = O\left(\frac{n\lg(Nn)\lg^{(k)} n}{\lg\lg\lg (Nn)}\right).
\]

In a query, applying binary search to these oracles takes a total time of 
\[
O(\lg\lg(Nn)r_k) = O(\lg\lg(Nn)).
\]
This brings us down to a range of only $O(\lg\lg(Nn))$ distance scales. To do the remaining $O(\lg\lg\lg(Nn))$ steps of the binary search on such a range, we introduce additional oracles. For $i = 1,2,\ldots,\lceil\lg(Nn)\rceil$ and $\alpha = 2^i$, we store scale-$(\alpha,\frac 1 2)$ distance oracles of Lemma~\ref{lemma:DigraphsOneScale} with parameters $r_i = (\lg^{(i)}(Nn))^4$ for $i = 1,2,3$. The total space for these oracles over all $\lceil\lg(Nn)\rceil$ scales is
\[
O\left(\lg(Nn)\cdot n\cdot\left(\sum_{i = 0}^{2}\frac{\lg r_i}{\sqrt{r_{i+1}}}\right)\right) = O\left(\frac{n\lg(Nn)}{\lg\lg\lg(Nn)}\right).
\]
The total time to query these oracles in the $O(\lg\lg\lg(Nn))$ binary search steps is
\[
O(\lg\lg\lg(Nn)r_3) = O((\lg\lg\lg(Nn))^5) = O(\lg\lg(Nn)).
\]
This shows Theorem~\ref{thm:Digraphs} for $O(1)$ approximation.

To get a $(1+\epsilon)$-approximate distance oracle, we store for $i = 1,2,\ldots,\lceil\lg(Nn)\rceil$ and $\alpha = 2^i$ scale-$(\alpha,\epsilon/4)$ distance oracles of Lemma~\ref{lemma:DigraphsOneScale} with parameters $r_1 = \lg^4(Nn)$, $r_2 = (\lg\lg(Nn) + \lg(1/\epsilon))^4$, and $r_3 = (\lg\lg\lg(Nn) + \lg(1/\epsilon))^4$. The total space for these oracles over all scales is
\begin{align*}
 & O\left(\lg(Nn)\cdot\frac{n}{\epsilon}\cdot\left(\frac{\lg n}{\sqrt{\lg^4(Nn)}} + \frac{\lg((\lg^4(Nn))/\epsilon)}{\sqrt{(\lg\lg(Nn) + \lg(1/\epsilon))^4}} + \frac{\lg((\lg\lg(Nn) + \lg(1/\epsilon))^4/\epsilon)}{\sqrt{(\lg\lg\lg(Nn) + \lg(1/\epsilon))^4}} \right)\right)\\
 & = O\left(\frac{n\lg(Nn)}{\epsilon\lg\lg\lg(Nn)}\right)
\end{align*}
and querying one of them takes time $O(r_3/\epsilon^5) = O((\lg\lg\lg(Nn))^4/\epsilon^5 + (\lg(1/\epsilon))^4/\epsilon^5)$.

To obtain the desired $(1+\epsilon)$-approximation of $d_G(u,v)$, our distance oracle first obtains a $4$-approximation as described above. This gives an estimate $\tilde d^{\alpha,\frac 1 2}(u,v) [\alpha/4,\alpha]$ and $d_G(u,v)$ is in this range as well. It then queries the scale-$(\alpha,\epsilon/4)$-oracles for the at most three graphs in $\mathcal G^\alpha$ containing $u$ and $v$ to get an estimate $\tilde d(u,v)$ with $d_G(u,v)\leq \tilde d(u,v)\leq d_G(u,v) + \alpha\epsilon/4\leq d_G(u,v) + (4d_G(u,v))\epsilon/4 = (1+\epsilon)d_G(u,v)$, as desired.

Our final $(1+\epsilon)$-approximate distance oracle thus has space $O(n\lg(Nn)/(\lg\lg\lg(Nn)\epsilon))$ and query time $O(\lg\lg(Nn) + (\lg\lg\lg(Nn))^4/\epsilon^5 + (\lg(1/\epsilon))^4/\epsilon^5)$.

Theorem~\ref{thm:Digraphs} follows if we can show that $(\lg\lg\lg(Nn))^4/\epsilon^5 = O(\lg\lg(Nn) + 1/\epsilon^{5.01})$. We may assume that $(\lg\lg\lg(Nn))^4/\epsilon^5\geq \lg\lg(Nn)$ since otherwise, $(\lg\lg\lg(Nn))^4/\epsilon^5 = O(\lg\lg(Nn))$. Then $(1/\epsilon)^5 \geq (\lg\lg(Nn))/(\lg\lg\lg(Nn))^4 = \omega((\lg\lg\lg(Nn))^{2000})$ so $(\lg\lg\lg(Nn))^4 = o((1/\epsilon)^{0.01})$.

It follows that $O((\lg\lg\lg(Nn))^4/\epsilon^5) = O(\lg\lg(Nn) + (1/\epsilon)^{5 + 0.01})$; the $0.01$ constant can of course be made arbitrarily small.

\subsection{The oracles for a single distance scale}\label{subsec:OracleSingleDistScale}
To show Lemma~\ref{lemma:DigraphsOneScale}, we focus on a single $(3,\alpha)$-layered subgraph $H\in\mathcal G^\alpha$ of $G$. We let $n_H$ denote the number of vertices in $H$ and let $T$ be a $(3,\alpha)$-layered spanning tree of $H$ rooted at a vertex $r$. We are given $\epsilon > 0$ and $r_1,r_2,\ldots,r_k\in [2,n]$ with $k = O(1)$ such that $r_{i+1}\leq \frac 1 2 r_i$ for $i = 1,2,\ldots,k-1$. In the following, we present the oracle of Lemma~\ref{lemma:DigraphsOneScale} for $H$; we will refer to this oracle as $\mathcal O$. To simplify our calculations, we will present a scale-$(\alpha,O(\epsilon))$ distance oracle; adjusting $\epsilon$ suitably then gives the desired scale-$(\alpha,\epsilon)$ distance oracle.

Some implementation details for $\mathcal O$ are left out in the description below and delayed until the space and time analysis.

\paragraph{$k$-level recursive decomposition of $G$ with $r$-divisions:}
The oracle $\mathcal O$ keeps a $k$-level recursive decomposition of $G$ (not $H$) into subgraphs which we call pieces. The root at level $0$ is the entire graph $G$. For $i = 0,1,\ldots,k-1$, each piece $P$ at level $i$ has as children the pieces of an $r_{i+1}$-division of $P$. For $i = 0,1,\ldots,k$ and for each vertex $u$ of $G$, we let $P_i(u)$ denote a piece at level $i$ containing $u$; the pieces are chosen such that $P_{i+1}(u)\subseteq P_i(u)$ for $i = 0,1,\ldots,k-1$. Note that $P_0(u) = G$ and $\partial P_0(u) = \emptyset$.

\paragraph{Vertex levels:}
The \emph{level} $\ell(u)\geq 0$ of a vertex $u$ of $G$ is the largest level $\ell$ such that $u$ is a boundary vertex of $P_{\ell+1}(u)$; if no such level exists, $\ell(u) = k$. We use the shorthand $P(u)$ for $P_{\ell(u)}(u)$; hence, $P(u)$ is a piece having $u$ in its interior such that none of the children of $P(u)$ have this property.

\paragraph{$O(\lg n)$-level recursive decomposition of $H$ with shortest path separators:}
Recall that $H$ is a subgraph $G_j^\alpha$ of $G$. We denote by $H^+$ the $(3,\alpha)$-layered graph $G_j^{\alpha}(r)$ where $r = r_j^{\alpha}$; recall that $H$ is obtained from $H^+$ by removing the super source $r$. Let $G_{\Delta}$ be a triangulation of $G$ with edge weights ignored.

Similar to Thorup's oracle, $\mathcal O$ keeps a recursive decomposition of $H^+$ using path separators of the $(3,\alpha)$-layered spanning tree $T$ of $H^+$. Let $\mathcal T$ denote the associated decomposition tree; its height is $O(\lg n)$. The nodes of $\mathcal T$ correspond to certain subgraphs of $H^+$. The root of $\mathcal T$ is the entire graph $H^+$ and the children of each non-leaf subgraph-node $H'$ are obtained by partitioning $H'$ in two with a balanced path separator of $T$; the $H'$-node of $\mathcal T$ is associated with this separator. Each path separator $S$ is a fundamental cycle, i.e., it consists of two (disoriented) paths in $T$ from the root $r$ and the two endpoints of the paths are connected by an edge of $G_{\Delta}$ not in $T$.

The separators are balanced such that on any root-to-leaf path in $\mathcal T$, the triangles of $G_{\Delta}$ contained in the subgraphs of $H^+$ along this path go down in size geometrically until reaching a leaf containing only a single triangle of $G_{\Delta}$.

Associate each vertex $u$ of $G$ with a triangle $\Delta(u)$ of $G_{\Delta}$ containing $u$. If $u$ belongs to $H$ then we associate $\Delta(u)$ with the leaf subgraph $L_u$ of $\mathcal T$ containing $\Delta(u)$. Let $S_i(u)\subseteq T$ denote the separator at level $i$ of $\mathcal T$ on the path from the root $H$ to the leaf $L_u$.

\paragraph{Local portals:}
Thorup's oracle stores for each $u$ of $H$ and each separator $S_i(u)$ a portal set $\port_i(u)$ which is a vertex set contained in $S_i(u)$ of size $O(1/\epsilon)$ with the following property: for each vertex $v\in S_i(u)$, there exists a portal $p\in\port_i(u)$ such that if $d_H(u,v)\leq\alpha$ then $d_H(u,p) + d_T(p,v)\leq d_H(u,v) + \epsilon\alpha$.

In order to save space, $\mathcal O$ stores a \emph{local portal set} $\lport_i(u)$ for $u$ instead of $\port_i(u)$ and only when $\ell(u) < k$. For each $p\in \port_i(u)$, $\lport_i(u)$ contains the first boundary vertex of $P(u)$ (if any) on a shortest $u$-to-$p$ path in $H$. Furthermore, $\mathcal O$ stores the additional portal set $\lport_i'(u) = \port_i(u)\cap P_{\ell(u)}(u)$; these portals are ordered in the same way as in Thorup's oracle, i.e., they are ordered linearly along the dipaths in $T$. Note that $\lport_i(u)\cup\lport_i'(u)\subseteq V(P(u))$.

For each $p\in\lport_i(u)\cup\lport_i'(u)$, $\mathcal O$ stores an estimate $\hat d_{P(u)}(u,p)$ such that $d_{P(u)}(u,p)\leq\hat d_{P(u)}(u,p)\leq d_{P(u)}(u,p) + \epsilon\alpha$.

All of the above information is also stored for the graph $\overline{P(u)}$ obtained from $P(u)$ by reversing the orientation of each edge. For the corresponding portal sets, we use the notation $\portrev_i(u)$, $\lportrev_i(u)$, and $\lportrev_i'(u)$, respectively.

\paragraph{Tree data structures:}
Denote by $\partial_i G$ the set of boundary vertices of pieces of $r_i$-divisions in the recursive decomposition of $G$ above. Let $\partial_i H$ be the subset of those boundary vertices belonging to $H$. $\mathcal O$ keeps a data structure of the following lemma for identifying the depth $\lcadepth_{\mathcal T}(L_u,L_v)$ in $\mathcal T$ of the $\lcads$ of two query leaves $L_u$ and $L_v$ (see the proof in Section~\ref{subsec:ProofSuccinctLca}).
\begin{lemma}\label{lemma:SuccinctLca}
There is a data structure $\mathcal O_{\lcads}$ for $\mathcal T$ which, given the labels in $G$ of two vertices $u$ and $v$ in $H$, can report $\lcadepth_{\mathcal T}(L_u,L_v)$ in $O(1)$ time. Its space bound is within that of Lemma~\ref{lemma:DigraphsOneScale}.
\end{lemma}
In addition, $\mathcal O$ keeps a data structure $\mathcal O_T$ described in the following lemma (see the proof in Section~\ref{subsec:ProofApproxTreeDist}).
\begin{lemma}\label{lemma:ApproxTreeDist}
There is a data structure $\mathcal O_T$ for the $(3,\alpha)$-layered spanning tree $T$ in $H$ which for any query vertices $u$ and $v$, where the $u$-to-$v$ path in $T$ is contained in a dipath, can report a value $\hat d_T(u,v)$ such that $d_T(u,v)\leq\hat d_T(u,v)\leq d_T(u,v) + \epsilon\alpha$. For all other vertex pairs $(u,v)$, $\mathcal O_T$ returns $\hat d_T(u,v) = \infty$. $\mathcal O_T$ has $O(1)$ query time and its space bound is within that of Lemma~\ref{lemma:DigraphsOneScale}.
\end{lemma}
This completes the description of the data stored by $\mathcal O$.

\subsubsection{Answering a query}
Algorithm~\ref{alg:QueryDigraph} gives the pseudo-code describing how $\mathcal O$ answers a query for a vertex pair $(u,v)$ using the procedure \textproc{Dist}.

\begin{algorithm}[!h]
\begin{algorithmic}[1]
  \Procedure{Dist}{$u,v$}
  \If{$\ell(u) = k$}
    \State Compute $d_{P(u)}(u,b)$ for each $b\in\partial P(u)$\label{Dist_Dijkstra_u1}
    \State Compute $d_{P(u)}(u,v)$ \Comment{$d_{P(u)}(u,v) = \infty$ if $v\notin P(u)$.\label{Dist_Dijkstra_u2}}
    \State \Return $\min\{d_{P(u)}(u,v),\min\{d_{P(u)}(u,b) + \Call{Dist}{b,v}\vert b\in\partial P(u)\}\}$\label{Dist_Recu}
  \EndIf
  \If{$\ell(v) = k$}
    \State Compute $d_{P(v)}(b,v)$ for each $b\in\partial P(v)$\label{Dist_Dijkstra_v1}
    \State Compute $d_{P(v)}(u,v)$ \Comment{$d_{P(v)}(u,v) = \infty$ if $u\notin P(v)$.\label{Dist_Dijkstra_v2}}
    \State \Return $\min\{d_{P(v)}(u,v),\min\{\Call{Dist}{u,b} + d_{P(v)}(b,v)\vert b\in\partial P(v)\}\}$\label{Dist_Recv}
  \EndIf
  \State Let $i = \lcadepth_{\mathcal T}(L_u,L_v)$\Comment{Using oracle $\mathcal O_{\lcads}$.\label{Dist_lca}}
  \State $\tilde d_u \gets\min\{\hat d_{P(u)}(u,p) + \Call{Dist}{p,v}\vert p\in\lport_i(u)\}$ \Comment{$\tilde d_u = \infty$ if $\lport_i(u) = \emptyset$.\label{Dist_du}}
  \State $\tilde d_v \gets\min\{\hat d_{\overline{P(u)}}(v,q) + \Call{Dist}{u,q}\vert q\in\lportrev_i(v)\}$ \Comment{$\tilde d_v = \infty$ if $\lportrev_i(v) = \emptyset$.\label{Dist_dv}}
  \State $\tilde d\gets\min\{\hat d_{P(u)}(u,p) + \hat d_T(p,q) + \hat d_{\overline{P(v)}}(v,q)\vert p\in\lport_i'(u)\land q\in\lportrev_i'(v)\}$ \Comment{Similarly for $\tilde d$\label{Dist_d}}
  \State \Return $\min\{\tilde d_u,\tilde d_v,\tilde d\}$
  \EndProcedure
\caption{The query procedure \textproc{Dist} for $\mathcal O$ for a vertex pair $(u,v)$.}
\label{alg:QueryDigraph}
\end{algorithmic}
\end{algorithm}

\subsubsection{Correctness}
To show correctness, we first observe that \textproc{Dist} can never underestimate distances in $G$; this follows since $\hat d_H$, $\hat d_{\overline H}$, and $\hat d_T$ have this property and since all shortest path computations are in subgraphs of $G$.

For the upper bound on the approximation, we show by induction on $\ell\geq 0$ that for all vertices $u$ and $v$ with $\ell(u) + \ell(v) = \ell$ and $d_H(u,v)\leq\alpha$, the output $\tilde d_H(u,v)$ of \textproc{Dist} on input $(u,v)$ satisfies $\tilde d_H(u,v)\leq d_H(u,v) + (4+2\ell)\epsilon\alpha$. This suffices since $(4+2\ell)\epsilon\alpha \leq (4+2k)\epsilon\alpha = O(\epsilon\alpha)$.

Let $i = \lcadepth_{\mathcal T}(L_u,L_v)$. For the base case, we have $\ell(u) = \ell(v) = 0$. Since $P_0(u) = P_0(v) = G$, $\lport_i(u) = \lportrev_i(u) = \emptyset$ so no recursive calls occur in lines~\ref{Dist_du} and~\ref{Dist_dv}. Furthermore, $\lport_i'(u) = \port_i(u)\cap P_i(u) = \port_i(u)$ and $\lportrev_i'(v) = \portrev_i(v)\cap \overline P_0(v) = \portrev_i(v)$ so by Thorup's analysis and by the property of $\hat d_H$, $\hat d_{\overline H}$, and $\hat d_T$, there are portals $p\in\port_i(u)$ and $q\in\portrev_i(v)$ such that
\[
  \tilde d_H(u,v)\leq \tilde d\leq d_H(u,p) + d_T(p,q) + d_H(q,v) + 3\epsilon\alpha\leq d_H(u,v) + 4\epsilon\alpha = d_H(u,v) + (4+2\ell)\epsilon\alpha,
\]
as desired.

Now assume that $\ell > 0$ and that the claim holds for smaller values. Let $u$ and $v$ be vertices with $\ell(u) + \ell(v) = \ell$ and $d_H(u,v)\leq\alpha$.

Assume first that $\ell(u) = k$ and let $S_{uv}$ be a shortest $u$-to-$v$ path in $H$. If $S_{uv}\subseteq P(u)$ then $\tilde d_H(u,v) = d_{P(u)}(u,v) \leq d_H(u,v)$ is returned. Otherwise, let $b$ be the first boundary vertex of $\partial P(u)$ from $u$ along $S_{uv}$. Since $\ell(b) + \ell(v) \leq \ell-1$, the induction hypothesis implies that $\tilde d_H(u,v)\leq d_{P(u)}(u,b) + \textproc{Dist}(b,v)\leq d_H(u,v) + (4+2(\ell-1))\epsilon\alpha$, as desired. The same bound holds if $\ell(v) = k$.

It remains to consider the case where both $\ell(u) < k$ and $\ell(v) < k$. It follows from Thorup's analysis that there are portals $p\in\port_i(u)$ and $q\in\port_i(v)$ such that $d_H(u,p) + d_T(p,q) + d_{\overline H}(v,q)\leq d_H(u,v) + \epsilon\alpha$ and the corresponding approximate $u$-to-$v$ path $Q_{uv}$ is contained in $H$. Denote its subpaths by $Q_{up} = u\leadsto p$, $Q_{pq} = p\leadsto q$, and $Q_{qv} = q\leadsto v$.

If there is a local portal $p_u\in\lport_i(u)\cap Q_{up}$ then let $p_u$ be the first such portal so that the subpath of $Q_{up}$ from $u$ to $p_u$ is contained in $P(u)$. Since $\ell(p_u) + \ell(v) \leq \ell-1$, the induction hypothesis applied to $p_u$ and $v$ gives
\begin{align*}
  \tilde d_H(u,v) & \leq \tilde d_u\leq\hat d_{P(u)}(u,p_u) + d_H(p_u,v) + (4+2(\ell-1))\epsilon\alpha\\
                  & \leq d_{P(u)}(u,p_u) + \epsilon\alpha + d_H(p_u,v) + (4+2(\ell-1))\epsilon\alpha\\
                  & \leq d_H(u,p_u) + d_H(p_u,v) + (4+2\ell-1)\epsilon\alpha\\
                  & \leq d_H(u,p) + d_T(p,q) + d_{\overline H}(v,q) + (4+2\ell-1)\epsilon\alpha\\
  & \leq d_H(u,v) + (4+2\ell)\epsilon\alpha,
\end{align*}
as desired. By symmetry, we get the same bound if there is a local portal in $\lportrev_i(v)\cap Q_{qv}$.

Now assume that $\lport_i(u)\cap Q_{up} = \lportrev_i(v)\cap Q_{qv} = \emptyset$. Then $p\in\lport_i'(u)$ and $q\in\lportrev_i'(v)$ so we get the same bound as for the case $\ell = 0$ above.

\subsubsection{Query time analysis}
Next, we show that \textproc{Dist}$(u,v)$ runs in time $O(r_k/\epsilon^{2k-1})$. Since every recursive call decreases $\ell(u) + \ell(v)$ by at least $1$, lines~\ref{Dist_Dijkstra_u1},~\ref{Dist_Dijkstra_u2},~\ref{Dist_Dijkstra_v1}, and~\ref{Dist_Dijkstra_v2} are each executed at most once during the recursion. To execute each of these lines, we require that $G$ is represented explicitly. For each vertex of a piece, we use a single bit to mark if it is a boundary vertex. Lines~\ref{Dist_Dijkstra_u1},~\ref{Dist_Dijkstra_u2},~\ref{Dist_Dijkstra_v1}, and~\ref{Dist_Dijkstra_v2} can now be executed by first running, say, a DFS from $u$ in $G$ that backtracks when reaching a marked vertex; this DFS identifies $P(u)$. Then a single-source shortest path algorithm and a single-target shortest path algorithm from $u$ in $P(u)$ are executed. Since this can be done in linear time in planar graphs, the total time spent on these lines is $O(r_k)$.

By Lemma~\ref{lemma:SuccinctLca}, $\mathcal O_{\lcads}$, line~\ref{Dist_lca} can be executed in $O(1)$ time. Since local portals in $\lport_i'(u)$ and in $\lportrev_i'(v)$ are ordered along dipaths of $T$ and since these sets are contained in $\port_i(u)$ and $\portrev_i(v)$, respectively, it follows from Lemma~\ref{lemma:ApproxTreeDist} and from Thorup's analysis that line~\ref{Dist_d} can be implemented to run in $O(1/\epsilon)$ time.

We will show that the total number of recursive calls is $O(r_k/\epsilon^{2k-2})$. The above will then imply that query time is $O(r_k/\epsilon^{2k-1})$.

If for some node in the query recursion tree, $\ell(u) = k$, the number of children of that node is $O(|\partial P_u|) = O(\sqrt r_k)$. The same bound holds if $\ell(v) = k$.

In lines~\ref{Dist_du} and~\ref{Dist_dv}, a total of $|\lport_i(u)| + |\lportrev_i(v)|\leq |\port_i(u)| + |\port_i(v)| = O(1/\epsilon)$ recursive calls are executed.

As argued above, each recursive call decreases $\ell(u) + \ell(v)$ by at least $1$. Since $\ell(p) + \ell(v)\leq \ell(u)-1 + \ell(v)\leq 2k-3$ in line~\ref{Dist_du} and $\ell(u) + \ell(q)\leq \ell(u) + \ell(v) - 1\leq 2k-3$ in line~\ref{Dist_dv}, the number of nodes in the recursion tree is bounded by $O(\sqrt r_k\cdot\sqrt r_k\cdot(1/\epsilon)^{2k-2}) = O(r_k/\epsilon^{2k-2})$. Hence, query time is $O(r_k/\epsilon^{2k-1})$, as desired.

\subsubsection{Space analysis}
We now show the space bound of Lemma~\ref{lemma:DigraphsOneScale}.

For $\ell\in\{0,1,\ldots,k-1\}$ and for each vertex $u$ for which $\ell(u) = \ell$, $\mathcal O$ stores four local portal sets. By construction, each of these sets has size at most that of Thorup's portal set for $u$ which is $O(\lg(n)/\epsilon)$. By definition of $P(u)$, $u$ is a boundary vertex of $P_{\ell(u)+1}(u)$ if $\ell(u) < k$. The total size (in terms of the number of portals) of the local portal sets over all $u$ of $H$ with $\ell(u) = \ell$ is therefore $O(|\partial_{\ell+1}H|\lg(n)/\epsilon)$.

To get the desired space bound, the idea is to compactly represent each local portal $p$ of a vertex $u$ as well as approximate distances $\hat d_{P(u)}(u,p)$ and $\hat d_{\overline{P(u)}}(u,p)$ using a small number of bits when $\ell = \ell(u) > 0$ by exploiting that both $u$ and $p$ belong to the small piece $P(u)$.

The compact lookup table approach immediately extends to the directed setting, allowing us to use only $O(\lg r_{\ell})$ bits to store $p$ with $u$ as well as $2$-approximations of $d_{P(u)}(u,p)$ and $d_{\overline{P(u)}}(u,p)$. The latter does not suffice here though since we need an an additive approximation error of at most $\epsilon\alpha$. To deal with this, let us recall how we represented a multiplicative $2$-approximation of, say, $d_{P(u)}(u,p)$. Pick a maximum-weight edge $e_{up}$ on a shortest path from $u$ to $p$ in $d_{P(u)}$. Letting $\rho_{up} = \lceil d_{P(u)}(u,p)/w(e_{up})\rceil$, the product $w(e_{up})\rho_{up}$ is a $2$-approximation since
\[
d_{P(u)}(u,p)\leq w(e_{up})\rho_{up} < w(e_{up})\left(\frac{d_{P(u)}(u,p)}{w(e_{up})} + 1\right) = d_{P(u)}(u,p) + w(e_{up})\leq 2d_{P(u)}(u,p).
\]
We get an additive error of at most $\epsilon\alpha$ by instead letting $\rho_{up}$ be the smallest multiple of $\epsilon$ of value at least $d_{P(u)}(u,p)/w(e_{up})$ since then
\[
d_{P(u)}(u,p)\leq w(e_{up})\rho_{up} < w(e_{up})\left(\frac{d_{P(u)}(u,p)}{w(e_{up})} + \epsilon\right) = d_{P(u)}(u,p) + \epsilon w(e_{up})\leq d_{P(u)}(u,p) + \epsilon\alpha.
\]

Since $\rho_{up}$ is divisible by $\epsilon$, $\rho_{up}/\epsilon$ is an integer of value at most $d_{P(u)}(u,p)/(w(e_{up})\epsilon) + 1\leq |E(P(u))|/\epsilon + 1$ and can thus be represented using $O(\lg(|E(P(u))|/\epsilon)) = O(\lg(r_{\ell}/\epsilon))$ bits.

Note that the above approximation is only used when $\ell > 0$; when $\ell = 0$, we store the exact distances to portals, i.e., in this case, $\hat d_{P(u)} = d_{P(u)}$ and $\hat d_{\overline{P(u)}} = d_{\overline{P(u)}}$.

We need to address the following issue with local portal sets: when \textproc{Dist} makes a recursive call for a vertex pair $(u_1,u_2)$, $u_1$ and $u_2$ are represented using $o(\lg n)$ bit local labels, meaning that two distinct vertices may have the same local labels. We need some way of obtaining the unique global $\Theta(\lg n)$-bit label in $G$ from a local label.

Recall that Lemma~\ref{lemma:DigraphsOneScale} permits us to use $O(n)$ space that does not depend on the choice of $\alpha$ and $H$. Since the recursive decomposition of $r$-divisions is a decomposition of $G$, we can thus afford to keep a pointer from each local label to its global label in $G$ using a total of $O(\sum_{i = 0}^{k-1}n/\sqrt{r_{i+1}}) = O(n)$ space where we used that $r_{i+1}\leq \frac 1 2 r_i$ for $i = 1,2,\ldots,k-1$ and that each pointer can be stored in one word. Recall that $\mathcal O$ also needs $G$ to be stored explicitly which can be done using $O(n)$ space as well.

From the above, the space used by $\mathcal O$ which is \emph{dependent} on $\alpha$ and $H$ is
\[
O(\frac 1 \epsilon \sum_{i = 0}^{k-1}|\partial_{i+1}H|\lg(r_i)),
\]
where we converted from number of bits to number of words. Since $\sum_{H\in\mathcal G^\alpha}|\partial_{i+1} H| = O(n/\sqrt{r_{i+1}})$ for $i = 0,1,\ldots,k-1$, we are within the space bound of Lemma~\ref{lemma:DigraphsOneScale}.

\subsection{Proof of Lemma~\ref{lemma:SuccinctLca}}\label{subsec:ProofSuccinctLca}
We will rely on the data structure in~\cite{JanssonSS07} which builds on~\cite{BenoitDMRRR05}. It uses $O(n_H)$ \emph{bits}, i.e., $O(n_H/\lg n)$ words of space, and can report $\lcadepth_{\mathcal T}(w_1,w_2)$ in $O(1)$ time from the preorder numbers $\pi(w_1)$ and $\pi(w_2)$ of any two nodes $w_1$ and $w_2$ in $\mathcal T$. It thus suffices for us to give a data structure that can translate the label of a query vertex $u$ in $H$ to the preorder number $\pi(L_u)$ in $O(1)$ time within the space bound of the lemma.

Keep a pointer from each vertex $u$ of $G$ to the its triangle $\Delta(u)$. Associated with each piece $P$ of an $r_i$-division in the recursive decomposition of $G$ is an $O(r_i)$-length array $A_P^i$ containing the $O(r_i)$ triangles $t$ of $G_{\Delta}$ such that all three vertices of $t$ are contained in $P$; the ordering of triangles in $A_P^i$ is arbitrary. Note that the space for storing pointers and arrays is independent of $H$ and $\alpha$.

For $i = 1,2,\ldots,k$ and for each node $w$ located $\lceil\lg r_i\rceil$ levels from the bottom in $\mathcal T$ (assuming $w$ exists), we pick a single arbitrary triangle $t_i(w)$ of $G_{\Delta}$ contained in a leaf of the subtree of $\mathcal T$ rooted at $w$. We refer to $t_i(w)$ as \emph{$i$-special}. A triangle of $G_{\Delta}$ contained in a leaf of $\mathcal T$ is also $i$-special if its vertex set intersects $\partial_i H$. A triangle which is not $i$-special for any $i = 1,2,\ldots,k$ is said to be $(k+1)$-special.

Associated with every $1$-special triangle $t$ is the preorder number $\pi(L(t))$ of the leaf $L(t)$ of $\mathcal T$ containing $t$.

For $i = 2,3,\ldots,k+1$ and for every $i$-special triangle $t$ with associated leaf $L(t)$, let $a_t$ be the ancestor $\lceil\lg r_{i-1} \rceil$ levels above $L(t)$ in $\mathcal T$. Let $P$ be the piece of an $r_{i-1}$-division containing $t$. Associated with $t$ is an index $j_t$ such that $A_P^{i-1}[j_t]$ is an $(i-1)$-special triangle $t'$ belonging to a leaf $L(t')$ below $a_t$ in $\mathcal T$; this index must exist since either all triangles in leaves below $a_t$ in $\mathcal T$ belong to $P$ (and one of them is $(i-1)$-special) or at least one of them must intersect $\partial_{i-1} P$ (and such a triangle is $(i-1)$-special). We also associate with $t$ the value $\delta_t^{i-1} = \pi(L(t)) - \pi(L(t'))$. Hence, $t$ is associated with the pair $(A_P^{i-1}[j_t],\delta_t^{i-1})$.

\paragraph{Dealing with a query:}
We can now answer a query for a vertex $u$ as follows. First, obtain $t_u = \Delta(u)$ using the pointer from $u$. Pick the smallest $i$ such that $t_u$ is $i$-special. If $i = 1$ then we are done since $\pi(L(t_u))$ is stored with $t_u$ and can thus be obtained in $O(1)$ time. If $i > 1$ then an $(i-1)$-special triangle $t'$ can be obtained as $A_P^{i-1}[j_t]$; a recursive algorithm then obtains $\pi(L(t'))$ from which $\pi(L(t_u))$ can be obtained as $\pi(L(t_u)) = \pi(L(t')) + \delta_t^{i-1}$. Query time is $O(k) = O(1)$.

\paragraph{Bounding space:}
It remains to bound space. For $i = 1,2,\ldots,k$, the number of $i$-special triangles incident to boundary vertices of pieces is $O(|\partial_i H|)$ and the number of $i$-special triangles associated with nodes $\lceil\lg r_i\rceil$ levels from the bottom of $\mathcal T$ is $O(n_H/2^{\lceil\lg r_i\rceil}) = O(n_H/r_i)$. The total space for storing all preorder numbers with $1$-special triangles is thus $O(n_H/r_1 + |\partial_1H|)$.

For $i = 2,3,\ldots,k+1$ and for each $i$-special triangle $t$ contained in a piece $P$, storing the index $j_t$ requires only $O(\lg|A_P^{i-1}|) = O(\lg r_{i-1})$ bits. Letting $t'$ be the special triangle $A_P^{i-1}[j_t]$, both $t$ and $t'$ are contained in leaves of the subtree $\mathcal T_{a_t}$ of $\mathcal T$ rooted at the ancestor node $a_t$, and hence $|\delta_t| = |\pi(L(t)) - \pi(L(t'))|\leq |\mathcal T_{a_t}| = O(2^{\lceil\lg r_{i-1}\rceil}) = O(r_{i-1})$ so storing $\delta_t$ requires only $O(\lg r_{i-1})$ bits.

We conclude that the total space is $O(n)$ independent of $H$ and $\alpha$ plus space $O(n_H\lg(r_k)/\lg n + \frac 1{\lg n}\sum_{i = 1}^{k}(|\partial_i H| + n_H/r_i)\lg r_{i-1})$. Since $n_H$ summed over all $H$ is $n$ and since $|\partial_i H| + n_H/r_i$ summed over all $H$ is $O(\sum_{H\in\mathcal G^\alpha}|\partial_i H|) = O(n/\sqrt{r_i})$, this space bound is within that of Lemma~\ref{lemma:DigraphsOneScale}.

\subsection{Proof of Lemma~\ref{lemma:ApproxTreeDist}}\label{subsec:ProofApproxTreeDist}
First, let us introduce some notation. For a vertex $u\in T$, let $Q(u)$ denote the dipath containing $u$ and let $r(u)$ be the endpoint of $Q(u)$, i.e., $Q(u)$ is of the form $a\leadsto u\leadsto r(u)$. Let $n(u)\in\{0,1,2\}$ be the number of dipaths above $Q(u)$ in $T$ (recall that $T$ is $(3,\alpha)$-layered). With Thorup's construction of $T$, $Q(u)$ is directed towards the root of $T$ iff $n(u)$ is odd. $\mathcal O_T$ stores $n(u)$ with each vertex $u$ of $H$, requiring $O(n_H)$ bits of space. Additionally, $\mathcal O_T$ stores an $\lcads$-data structure for $T$ (ignoring edge orientations of $T$); we describe it at the end of the proof. Its query time is $O(1)$ and its space requirement is within the bounds of Lemma~\ref{lemma:DigraphsOneScale}.

Given a query pair $(u,v)$, $\mathcal O_T$ can determine if $u$ and $v$ are on the same dipath, i.e., if $Q(u) = Q(v)$, in $O(1)$ time: first use the $\lcads$-structure to check if $u$ and $v$ are on the same root-to-leaf (disoriented) path in $T$ and if so, check if $n(u) = n(v)$. For $v\in Q(u)$, the $\lcads$ structure can also be used to check in $O(1)$ time if $v$ is closer to $r(u)$ than $u$ (where $\mathcal O_T$ uses the parity of $n(u) = n(v)$ to determine if $Q(u)$ is directed towards the root of $T$). We may thus make the assumption that $v$ is on the subpath of $Q(u)$ from $u$ to $r(u)$ since otherwise, $\mathcal O_T$ may simply output $\hat d_T(u,v) = \infty$.

\paragraph{Reduction to restricted queries:}
We claim that it suffices to handle a query for an approximate distance $\hat d_T(u,r(u))$ and to ensure that $d_T(u,r(u))\leq\hat d_T(u,r(u))\leq d_T(u,r(u)) + \epsilon\alpha/2$. For suppose we have shown this. Then for any $v\in Q(u)$ with $d_T(v,r(u))\leq d_T(u,r(u))$ (i.e., $v$ is on the path $u\leadsto r(u)$ in $T$),
\[
d_T(u,r(u)) - d_T(v,r(u)) - \epsilon\alpha/2 \leq \hat d_T(u,r(u)) - \hat d_T(v,r(u)) \leq d_T(u,r(u)) + \epsilon\alpha/2 - d_T(v,r(u)).
\]
Since $d_T(u,r(u)) - d_T(v,r(u)) = d_T(u,v)$, we have
\[
d_T(u,v) - \epsilon\alpha/2\leq \hat d_T(u,r(u)) - \hat d_T(v,r(u))\leq d_T(u,v) + \epsilon\alpha/2
\]
and hence
\[
d_T(u,v)\leq \hat d_T(u,r(u)) - \hat d_T(v,r(u)) + \epsilon\alpha/2\leq d_T(u,v) + \epsilon\alpha.
\]

It follows that $\mathcal O_T$ can output $\hat d_T(u,v) = \hat d_T(v,r(u)) - \hat d_T(u,r(u)) + \epsilon\alpha/2$ for query pair $(u,v)$ to satisfy the requirement of Lemma~\ref{lemma:ApproxTreeDist}.

\paragraph{Dealing with a restricted query:}
It remains to handle a query for a vertex pair $(u,r(u))$. We use the same recursive decomposition with $r$-divisions as in Section~\ref{subsec:OracleSingleDistScale}. Stored with $u$ is the first boundary vertex $b$ of $P(u)$ on the $u$-to-$r(u)$ path in $T$; if such a boundary vertex does not exist, $b = r(u)$. In addition, $\mathcal O_T$ stores an approximate distance $\hat d_T(u,b)$ such that $d_T(u,b)\leq \hat d_T(u,b)\leq d_T(u,b) + \epsilon/(2k) = d_T(u,b) + O(\epsilon)$. The output $\hat d_T(u,r(u))$ is then $\hat d_T(u,b)$ plus a recursive query for the pair $(b,r(u))$.

Since the number of recursive calls is at most $k$, the additive error for $\hat d_T(u,r(u))$ is at most $k(\epsilon/(2k)) = \epsilon/2$ and the query time is $O(k) = O(1)$. Using a similar analysis as in Section~\ref{subsec:OracleSingleDistScale}, the space required for $\mathcal O_T$ (excluding the space for the $\lcads$-structure below) is $O(n_H\lg(r_k)/\lg n + \frac 1{\lg n}\sum_{i = 0}^{k-1}|\partial_{i+1}H|\lg r_i)$ which is within the space bound of Lemma~\ref{lemma:ApproxTreeDist}.

\paragraph{The $\lcads$-structure for $T$:}
To complete the proof, it remains to present the $\lcads$-structure for $T$ and bound its space and query time. The data structure is somewhat similar to that of Lemma~\ref{lemma:SuccinctLca}. As argued in the proof of that lemma, it suffices to give a data structure with the desired space bound which, given the label of a vertex $u$ of $G$ such that $u$ belongs to $T$, outputs the preorder number $\pi(u)$ of $u$ in $T$ in $O(1)$ time.

Consider the sequence of vertices of $T$ ordered by increasing preorder numbers. For $i = 1,2,\ldots,k$, let $\mathcal S_i$ be the set of maximal subsequences each of which is fully contained in a piece in an $r_i$-division in the recursive decomposition of $G$. Thus, each subsequence in $\mathcal S_i$ has length at most $r_i$. We claim that $|\mathcal S_i| = O(|\partial_i H|)$. To show this, it suffices to give an $O(|\partial_i H|)$ bound on the number of maximal subwalks of the full walk of $T$ inducing the preorder walk such that each subwalk is fully contained in a piece of an $r_i$-division; this suffices since the preorder walk is obtained from the full walk by removing a subset of the vertices.

Every subwalk except possibly the first and last starts and ends in boundary vertices of $\partial_i H$. Since these subwalks are pairwise non-crossing in a planar embedding of $G$, replacing each such walk $b_1\leadsto b_2$ between boundary vertices $b_1$ and $b_2$ with the directed edge $(b_1,b_2)$ gives a simple planar graph on the vertex set $\partial_i H$. Since simple planar graphs are sparse, the number of subwalks is $O(|\partial_i H|)$, as desired.

For $i = 1,2,\ldots,k$, we say that a vertex is \emph{$i$-special} if it is the start or endpoint of a sequence in $\mathcal S_i$. A vertex which is not $i$-special for any $i = 1,2,\ldots,k$ is said to be $(k+1)$-special.

The following is stored with each vertex $u\in T$. If $u$ is $1$-special then store $\pi(u)$ with $u$. Otherwise, let $i$ be the smallest index in $\{2,3,\ldots,k+1\}$ such that $u$ is $i$-special and store with $u$ a pointer to the $(i-1)$-special start vertex $s_{i-1}(u)$ of the subsequence of $\mathcal S_{i-1}$ containing $u$ and store also $\delta_{i-1}(u) = \pi(u) - \pi(s_{i-1}(u))$.

To obtain $\pi(u)$ from a query vertex $u\in T$, do as follows. Pick the smallest $i$ such that $u$ is $i$-special. If $i = 1$, the query is trivial since $\pi(u)$ is stored with $u$. Otherwise, obtain $s_{i-1}(u)$ and $\delta_{i-1}(u)$ from $u$, recursively obtain $\pi(s_{i-1}(u))$, and output $\pi(s_{i-1}(u)) + \delta_{i-1}(u)$.

Correctness is clear and query time is $O(k) = O(1)$. It remains to bound space.

The total space for storing preorder numbers with $1$-special vertices is $O(|\partial_1 H|) = O(\frac 1{\lg n}|\partial_1 H|\lg r_0)$. Let $u$ be a vertex in $T$ and let $i$ be the smallest index such that $u$ is $i$-special; assume $i > 1$. The sequence in $\mathcal S_{i-1}$ containing $u$ has length at most $r_{i-1}$ and hence only $O(\lg r_{i-1})$ bits are needed to store $s_{i-1}(u)$ and $\delta_{i-1}(u)$. Since the number of $i$-special vertices is at most $|\partial_i H|$, the total space required is $O(n_H/\lg n)$ for the data structure of~\cite{JanssonSS07} plus $O(n_H\lg(r_k)/\lg n + \frac{1}{\lg n}\sum_{i = 1}^{k}|\partial_iH|\lg r_{i-1})$. Summing over all $H\in\mathcal G^\alpha$, this is within the space bound of Lemma~\ref{lemma:ApproxTreeDist} since for $j = 1,2,\ldots,k$, $\sum_{H\in\mathcal G^\alpha}|\partial_{j} H| = O(n/\sqrt{r_{j}})$ and since $\sum_{H\in\mathcal G^\alpha}n_H = n$.

\section{Weak Nets in Linear Time}\label{section:weak-net}

\paragraph{Preprocessing} First, we delete from $G$ every edge of weight strictly larger than $r$ in $O(n)$ time. If the deletion disconnects $G(V,E,w)$, we apply the algorithm in this section to each connected component of $G$. Henceforth, we assume that $G(V,E,w)$ is connected and every edge has weight at most $r$. 

Our algorithm, described in Algorithm~\ref{alg:decompose}, is inspired by the algorithm to find a sparse cover for $K_{h,h}$-minor-free graphs by Abraham et al.~\cite{AGMW10}. The main idea is to find a \emph{well separate decomposition} that we formally define in Definition~\ref{def:WSD}. Here a \emph{$(\gamma\cdot r)$-assignment} $\mathcal{A}_i$ associated with a set $K_i$ is simply a family of $|K_i|$ sets such that for each $x\in K_i$, there is a set $\mathcal{A}_i[x]$ containing $x$ such that for every vertex $y\in \mathcal{A}_i[x]$, $d_G(x,y)\leq \gamma \cdot r$. Note that $\mathcal{A}_i[x]$ may contain vertices that are not in $K_i$.

\begin{definition}[Well Separate Decomposition]\label{def:WSD}A well $(r,\gamma)$-separate decomposition of $G$ for $K$ is a set of triples $$\mathcal{H} = \{(H_1,K_1,\mathcal{A}_1), \ldots, (H_s, K_s,\mathcal{A}_s)\}$$ such that:
	
	\begin{itemize}[noitemsep]
		\item[(1)]each $H_i$ is a subgraph of $G$, $K_i \subseteq K\cap H_i$ is a subset of terminals in $H_i$, and $\mathcal{A}_i$ is an $(\gamma\cdot r)$-assignment associated with $K_i$.
		\item[(2)] $\mathcal{C} = \{H_1,H_2,\ldots, H_s\}$ is a $(\gamma, \gamma, r)$-sparse cover of $G$.  
		\item[(3)]  $d_G(K_i,K_j)\geq r$ for every $i\not= j$.
		\item[(4)] $\cup_{i\in [s]}\mathcal{A}_j$ is a partition of $K$.
	\end{itemize}
\end{definition}

The technical bulk of this section is showing that the decomposition $\mathcal{H}$ returned in line~\ref{line:WSD} by Algorithm~\ref{alg:decompose} is a well separate decomposition with $\gamma = O(1)$.

\begin{lemma}\label{lm:WSD}The decomposition $\mathcal{H}$ returned in line~\ref{line:WSD} of Algorithm~\ref{alg:decompose} is a well $(r,O(1)$-separate decomposition $\mathcal{H}$ of $G(V,E,w)$ for $K$. Furthermore, the algorithm can be  implemented in time  $O(n)$.
\end{lemma}

Assuming that $\mathcal{H}$ returned in line~\ref{line:WSD} of Algorithm~\ref{alg:decompose} is well $(r,O(1)$-separate decomposition, in lines~\ref{line:Net} to~\ref{line:Net-end}, we compute a net $N$ and its associted assignment $\mathcal{A}$. 

\begin{lemma}\label{lm:net-correctness}
	Algorithm~\ref{alg:decompose} currrectly returns a weak $(r,O(1))$-net $N$ and its assignment $\mathcal{A}$ for $K$ in time $O(n)$.
\end{lemma}
\begin{proof}
	Clearly, given well separate decomposition $\mathcal{H}$, lines~\ref{line:Net} and~\ref{line:Net-end} can be implemented in time $O(n)$.

	By line~\ref{line:Net-Assignment} and property (4) in Definition~\ref{def:WSD}, $\mathcal{A}$ is a cover for $K$. By property (2), for every $x$ added to $N$ and every $y\in K_i \setminus x$ in line~\ref{line:Add-Net}, $d_G(x,y)\geq d_{H_i}(x,y) \leq \gamma r$ where $\gamma = O(1)$. By property (1), for every $z \in \mathcal{A}_i[y]$, $d_G(y,z)\leq \gamma r$. Thus, by triangle inequality, $d_G(x,z)\leq 2\gamma r = O(1)r$. 
	
	It remains to show that the distance in $G$ between any two different vertices in $N$ is at least $r$. By construction, for each triple $(H_i,K_i,\mathcal{A}_i)\in \mathcal{H}$ where $K_i\not= \emptyset$, the algorithm picks exactly one vertex to $N$. Thus, for every two points $x\not= y \in K$, they must be in two terminal sets, say $K_a$ and $K_b$, of two different triples. By property (3), $d_G(x,y)\geq d_G(K_a,K_b)\geq r$, as desired.
\end{proof}

For the rest of this section, we focus on proving Lemma~\ref{lm:WSD} The decomposition procedure, \texttt{Decompose}($H_i,K_i,\mathcal{A}_i$), is recusive and has three levels of recursion; the input to each level is a subgraph $H_i$ of $G$, subset of terminals $K_i \subseteq H_i\cap K$, an assignment $\mathcal{A}_i$ associated with $K_i$, and a subset of vertices $C_i\subseteq V(H_i)$ called the \emph{core} of $H_i$. $H_i$ is assumed to be a connected graph; lines~\ref{line:CC} to~\ref{line:CC-end} guarantee this assumption. The algorithm guarantees that $K_i \subseteq C_i\cap K$; we call $K_i$ a set of \emph{core terminals}. 

The algorithm starts by picking a root vertex $r_T\in V(H_i)$ (line~\ref{line:root-pick}) and slicing the graphs into \emph{slices} $\{S_1,S_2,\ldots, S_{J}\}$ each of \emph{width} $jr$: $S_j$ contains all vertices of distance at least $jr$ and at most $(j+1)r$ from $x$ (line~\ref{line:Slice}). We then compute the subset of vertices, $\bar{C}_{j,i+1}$, in the core in each slide $S_j$ -- we call $\bar{C}_{j,i+1}$ a \emph{subcore} -- and  construct $G_{j,i+1}$ as a subgraph induced by vertices in the ball of radius $(i+1)r$ from the subcore (line~\ref{line:Gij}).  (See Figure~\ref{fig:Net-Const}(a).)

The family of $\{X_i\}^J_{0}$ is a partition of core terminals $K_i$; $X_0$ is set to be $\emptyset$ for notatinal convenience.    Set $X_j$ (line~\ref{line:Xj}) is the subset of terminals within distance at most $(i+1)r$ from $\bar{K}_{i+1}$, which is a set of core terminals in the subcore $\bar{C}_{i+1}$.  In  constructing  $\bar{K}_{i+1}$ (line~\ref{line:Kji}), we ignore all the core terminals that was included in the previous sets $X_{\leq j-1}$; this will guarnatee that $\{X_i\}^J_{0}$ is a partition of $K_i$. To form assignment $\mathcal{A}_{i+1}$, we simply assign to each terminal $x\in \bar{K}_{i+1}$ nearest terminals in $X_j$ (line~\ref{line:Ai}); this will guarantee that:

\begin{observation}\label{obs:term-partition} Every terminal  $z \in  K_i\setminus \cup_{j=0}^{J_i}\bar{K}_{j,i+1}$, there is a terminal $x \in \bar{K}_{j,i+1}$ for some $j$ such that $z\in \mathcal{A}_{i+1}[x]$.
\end{observation}

Let $\tau$ be the recursion tree representing the execution of the decomposition algorithm; $\tau$ has depth $3$. Each node $\nu \in \tau$ is associated with argument $(H_i,K_i,\mathcal{A}),C_i,i)$ to the recursive call corresponding to $\nu$.

\begin{figure}[!htb]
	\center{\includegraphics[width=0.9\textwidth]
		{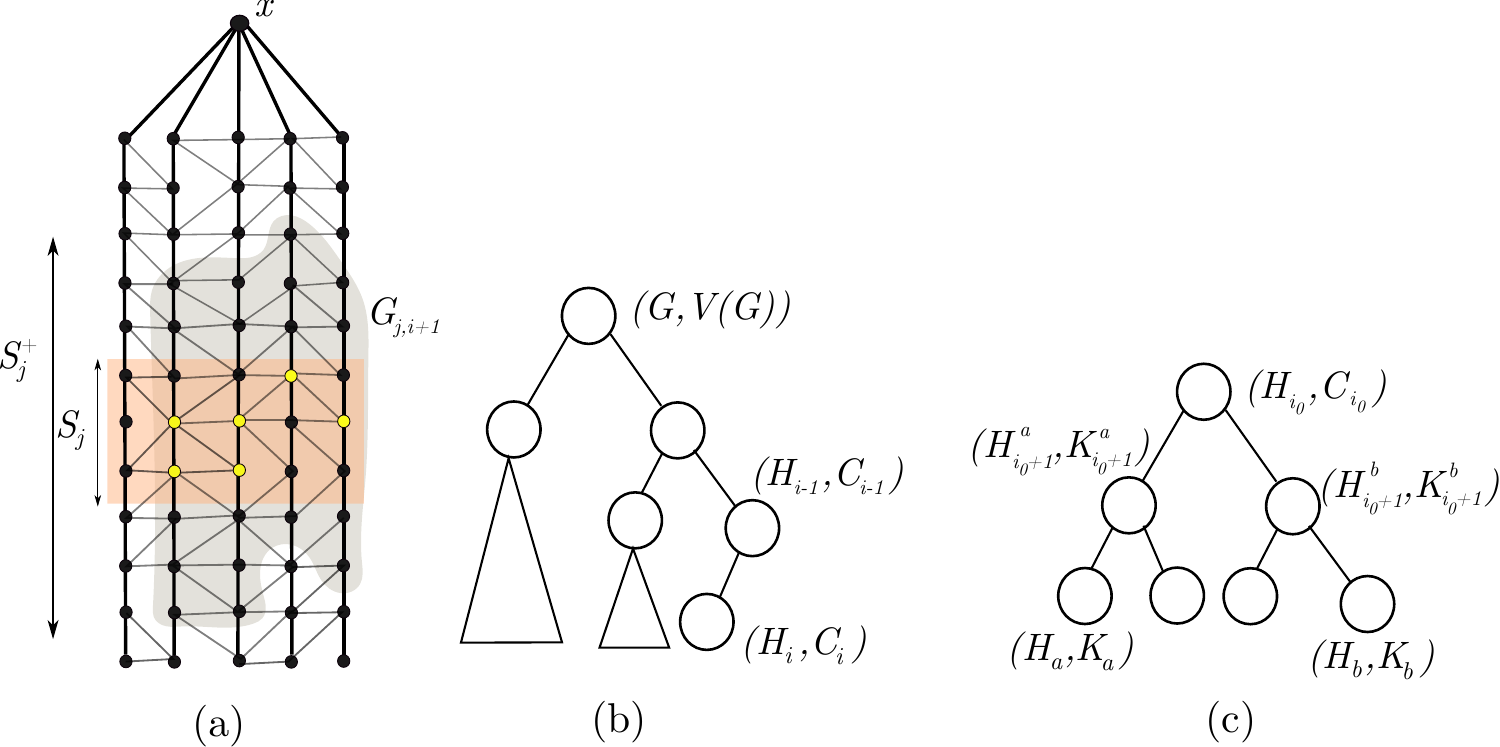}}
	\caption{(a) A slice $S_j$, the subcore $\bar{C}_{j,i+1}$ which consists of yellow vertices, and $G_{j,i+1}$ which contains vertices in the shaded region. (b) Illustration for the proof of Lemma~\ref{lm:ball-prop}. (c) Illustration for the proof of Lemma~\ref{lm:core-term-dist}}
	\label{fig:Net-Const}
\end{figure}

\begin{lemma}\label{lm:ball-prop} Let $C_i$ be a core at level $i$ of subgraph $H_i$ associated with a node $\nu \in \tau$, then for every $v\in C_i$:
	\begin{equation*}
		B_{G}(v,r) = B_{H_0}(v,r) \subseteq B_{H_1}(v,2r) \ldots \subseteq B_{H_i}(v,(i+1)r).
	\end{equation*}
	where $H_a$, $a \leq i-1$, is a graph associated with the ancestor at level $a$ of $\nu$ in $\tau$. In particular, this implies that $B_{G}(v, r)\subseteq V(H_i)$.	
\end{lemma}
\begin{proof}
	See Figure~\ref{fig:Net-Const}(b) for an illustration.	Note that by construction, $C_{i}\subseteq C_{i-1} \ldots\subseteq C_{0} = V(G)$.   We prove the lemma by induction on $i$.  By the construction in line~\ref{line:Gij} and the fact that $v\in C_0$, $B_{H_0}(v,r) = B_{G}(v,r)\subseteq V(G_{j,1})$ for some $j \in [0,J]$. Since $H_1$ is a connected component of $G_{1,j}$ containing $v$, it holds that $B_G(v,r)\subseteq V(H_1)$.
	
	Inductively, we assume that $ B_{H_a}(v,(a+1)r) \subseteq V(H_{i})$ for all $a \leq i-1$.  That means $B_{H_a}(v, (a+1)r)\subseteq B_{H_{i}}(v, i \cdot r)$ as $a+1\leq i$. By the construction in line~\ref{line:Gij}, $B_{H_{i}}(v,(i+1)r)\subseteq G_{j,i+1}$ for some $j \in [0,J]$. Note that $H_{i+1}$ is the connected component of $G_{j,i+1}$ containing $v$. Thus, $B_{H_{i}}(v,(i+1)r) \subseteq H_{i+1}$ and this implies $ B_{H_a}(v, (a+1)r) \subseteq B_{H_{i+1}}(v,(i+2)r)$ for all $a \leq i$.
\end{proof}

\begin{lemma}\label{lm:core-term-dist} Let $K_a,K_b$ be sets of core terminals associated with two leaves of $\tau$, then $$d_G(K_a,K_b)\geq r.$$
\end{lemma}
\begin{proof}
	Let $\mu_a,\mu_b$ be two leaves of $\tau$ corresponding to $K_a,K_b$, respectively. Let $\nu$ be the lowest common ancestor of $\mu_a$ and $\mu_b$ and $i_0$ be the level of $\nu$ (see Figure~\ref{fig:Net-Const}(c)). Let $H_{i_0}, K_{i_0}$ be the subgraph and core terminal set associated with $\nu$. 	By construction, $u,v \in K_{i_0}$. Let $K^a_{i_0+1}$ and $K^b_{i_0+1}$ be core terminal sets associated with two children of $\nu$ that contain $u$ and $v$, respectively. 
	
	Observe by construction that either (a) $u$ and $v$ are in different slices of $H_{i_0}$ or (b) $u$ and $v$ are in different connected components of graph $G_{j,i_0+1}$ for some $j \in [0,J_i]$. We consider each case separately.
	\begin{itemize}[noitemsep]
		\item \textbf{Case a.~} In this case, let $S_{j_1}$ and $S_{j_2}$ where $j_1 < j_2$ be two slices containing $u$ and $v$. Observe that $u\in \bar{C}_{i_0+1,j_1}$ and  $v\in \bar{C}_{i_0+1,j_2}$ since $\{u,v\}\subseteq C_{i_0}$. By construction in line~\ref{line:Kji}, it must be that $v\not \in B_{H_{i_0}}(u, (i_0+1)r)$. By Lemma~\ref{lm:ball-prop}, $v\not \in B_{G}(u,r)$ and hence $d_G(u,v)\geq r$ as desired.
		\item \textbf{Case b.~} In this case, $v\not\in B_{H_{i_0+1}}(u, (i_0+2)r)$ since $u$ and $v$ are in different connected components. By Lemma~\ref{lm:ball-prop}, $v\not\in B_{G}(u,r)$ and hence $d_G(u,v)\geq r$.
	\end{itemize}
	Thus, in both cases, $d_G(u,v)\geq r$.
\end{proof}

\noindent We obtain the following corollary of  Lemma~\ref{lm:core-term-dist}.
\begin{corollary}\label{cor:prop-(3)} $\mathcal{H}$ satisfies property (3) in Definition~\ref{def:WSD}.
\end{corollary}

\begin{lemma}\label{lm:Sparse-cover} $\mathcal{C} \stackrel{\mbox{\tiny{def.}}}{=} \{H : (H,K,\mathcal{A})\in \mathcal{H}\}$ is a $(O(1),O(1),r)$-sparse cover of $G$. 
\end{lemma}
\begin{proof}
	$\mathcal{C}$ is exacly the family of clusters found by the algorithm of Abraham et al.~\cite{AGMW10} when the excluded minor is $K_{3,3}$, instead of $K_{h,h}$. Since clusters in Abraham et al.~\cite{AGMW10} is a $(O(h^2), 2^{O(h)}h!, r)$-sparse cover, with $h = 3$, $\mathcal{C}$ is a $(O(1),O(1), r)$-sparse cover of $G$. We remark the proof of Abraham et al.~\cite{AGMW10} implicitely use the fac that edges of $G$ have weight at most $r$, that we achieve in the preprocessing step.
\end{proof}

\begin{algorithm}[!ht]
	\hspace*{\algorithmicindent} \textbf{Input:} $ G(V,E,w),r$ \\
	\hspace*{\algorithmicindent} \textbf{Output:}  $N,\mathcal{A}$ \Comment{Net $N$ and its associated assignment $\mathcal{A}$.}
	\begin{algorithmic}[1]
		\Procedure{WeakNet}{$G(V,E,w),r$}
		\State $\mathcal{A} \leftarrow \{\{z\} : z\in K\}$
		\State $\mathcal{H} \leftarrow$ \texttt{Decompose}($G,K,\mathcal{A}, V(G),0$) \label{line:WSD}
		\State $N\leftarrow \emptyset$
		\For{each triple $(H_i,K_i,\mathcal{A}_i)\in \mathcal{H}$ where $K_i\not= \emptyset$} \label{line:Net}
			\State pick an arbitrary vertex $x\in K_i$ and add to $N$\label{line:Add-Net}
			\State set $\mathcal{A}[x] \leftarrow \cup_{y\in K_i\setminus x}   \mathcal{A}[y]$\label{line:Net-Assignment}
		\EndFor \label{line:Net-end}
		\State return $(N, \mathcal{A})$
		\EndProcedure
		\Procedure{Decompose}{$H_i,K_i,\mathcal{A}_i, C_i,i$}
			\If{$i = 3$}
				\State	return $\{(H_i, K_i,\mathcal{A}_i)\}$
			\EndIf
			\State Pick a vertex $r_T \in V(H_i)$\label{line:root-pick}
			\State $\mathcal{H}\leftarrow\emptyset$; $X_0 = \emptyset$
			\State $J_i = \lceil  \frac{\max_{v\in V(H_i)}d_{H_i}(r_T,v)}{r} \rceil$
			\For{$ j \leftarrow 0$ to $J_i$} 
				\State $S_{j} \leftarrow  \{u : j r \leq d_{H_i}(u,r_T) < (j+1)r\}$\label{line:Slice}
				\State $\bar{C}_{j,i+1} \leftarrow S_{j}\cap C_{i}$
				\State $G_{j,i+1} \leftarrow H_i[B_{H_i}(\bar{C}_{j,i+1},(i+1)r)]$\label{line:Gij}
				\State $X_{\leq j-1} \leftarrow \bigcup_{0\leq t\leq j-1}X_t$
				\State $\bar{K}_{j,i+1}\leftarrow (K_i\cap \bar{C}_{j,i+1})\setminus X_{\leq j-1}$\label{line:Kji}
				\State $X_{j} \leftarrow B_{H_i}(\bar{K}_{i+1}, (i+1)r)\setminus X_{\leq j-1}$\label{line:Xj}
				\State $\bar{\mathcal{A}}_{i+1}\leftarrow \emptyset$
				\For{each terminal $x \in \bar{K}_{i+1}$}
					\State $\mathcal{A}_{i+1}[x] \leftarrow \mathcal{A}_i[x]\cup \{\mbox{terminals in $X_j$ closest to $x$, ties are broken arbitrarily}\}$ \label{line:Ai}
				\EndFor
				\For{each connected component $H_{i+1}$ of $G_{j,i+1}$}\label{line:CC}
					\State $\mathcal{A}_{i+1} \leftarrow \{\bar{\mathcal{A}}_{i+1}[x] : x\in \bar{K}_{i+1}\cap V(H_{i+1})\}$
					\State $K_{i+1}\leftarrow \bar{K}_{j,i+1}\cap V(H_{i+1})$
					\State $C_{i+1}\leftarrow V(H_{i+1})\cap \bar{C}_{j,i+1}$
					\State $\mathcal{H} \leftarrow \mathcal{H}\cup$\texttt{Decompose}($H_{i+1},K_{i+1},  \mathcal{A}_{i+1}, C_{i+1},i+1$)
				\EndFor \label{line:CC-end}
			\EndFor
			\textbf{return} $ \mathcal{H}$
		\EndProcedure
		\caption{Computing a weak $(r,O(1))$-net of a planar graph $G(V,E,w)$.}
		\label{alg:decompose}
	\end{algorithmic}
\end{algorithm}

\begin{lemma}\label{lm:Prop-H} $\mathcal{H}$ is a well $(r,O(1))$-separate decomposition of $G$.
\end{lemma}
\begin{proof}
	Property (3) follows from Corollary~\ref{cor:prop-(3)}. Property (2) follows directly from Lemma~\ref{lm:Sparse-cover}. It remains to show (1) and (4).
	
	Observe that for any $x \in K_j$ for some $(H_j,K_i,\mathcal{A}_j)\in \mathcal{H}$, by construction, every $v\in \mathcal{A}_j[x]$ is with in distance $(i+1)r \leq 4r$ from $x$ by the construction in line 16. Thus, $\mathcal{A}_j$ is a $4r$-assignment associated with $K_i$; this implies (1). The fact that $\cup_{i\in [s]}A_j$ is a partition of $K$ follows by induction and Observation~\ref{obs:term-partition}; this implies (4). 
\end{proof}

To complete the proof of Lemma~\ref{lm:WSD}, we show that Algorithm~\ref{alg:decompose} can be efficiently implemented.

\begin{lemma}Algorithm~\ref{alg:decompose} can be implemented in $O(n)$ time.
\end{lemma}
\begin{proof}
	We only need to show that the procedure decompose applied to $H_i$	 can be implemented in $O(|V(H_i)|\log (|V(H_i)|))$ time. This is because, the recursion tree $\tau$ only has 3 levels, and the total number of vertices in all graphs $H_i$ associated with nodes in each level is bounded by $\sum_{(H_j,K_i,\mathcal{A}_j) \in \mathcal{H}} |V(H_j)| = O(n)$ by property (2) in Definition~\ref{def:WSD}.
	
	We now focusing on implementing one decomposition step on $H_i$; note that $i\leq 1$. Let $n_i = |V(H_i)|$. First, we find a shortest path tree $T$ rooted at vertex $r_T$ chosen in line~\ref{line:root-pick} in $O(n_i)$ using the algorithm of Henzinger et al.~\cite{HKRS97}. Given $T$, all slices $\{S_j\}_{j=1}^J$ can be found in $O(n_i)$ time.  Let $S_j^{+} = \cup_{j-(i+1) \leq a \leq j + (i+1)} S_{a}$. Since $i+1\leq 4$, we have:
	\begin{equation}\label{eq:Slices-size}
		\sum_{j\in [J]}|S_j^{+}| \leq 9 \sum_{j\in [J]}|S_j| = O(n_i)
	\end{equation}
	
	Note that each $S_j^{+}$ is induces a planar graph in $H_i$.  Observe that $G_{j,i+1}$ is a subgraph of $H_i[S_j^{+}]$ and hence, line~\ref{line:Gij} can be implemented in $O(|S_j^{+}| )$ by adding a dummy vertex $y$,  connecting to every vertices in $\bar{C}_{j,i+1}$ by edges of length $+\infty$ and applying the algorithm of Henzinger et al.~\cite{HKRS97} from $y$. We note that single source shortest path algorithm of Henzinger et al.~\cite{HKRS97} applies not only to planar graphs but also to graphs whose subgraphs have sublinear separators which can be found in linear time; clearly a planar graph plus a single dummy vertex connected to all other vertices belongs to this class. 
	
	Using the same idea, we can find $X_j$ in line~\ref{line:Xj} in time $O(|S_j^{+}|)$ by finding a shortest path tree rooted at a dummy vertex that is connected to every vertex in $\bar{K}_{j,i+1}$ by an edge of weight $+\infty$. Once the shortest path tree (of $H_i[S_j^{+}]$) rooted at $y$ is given, finding terminals in $X_j$ closes to a terminal $x\in\bar{K}_{j,i+1}$ in line 19 can be done in $O(|S_j^{+}|)$ time. Thus, the total running time is:
	\begin{equation}
		\sum_{j\in [J]}O(|S_j^{+}|) \stackrel{\mbox{\tiny{Eq.~\ref{eq:Slices-size}}}}{=} O(n_i),
	\end{equation} 
	as desired.
\end{proof}

\paragraph{Acknowledgement.~} Hung Le is supported by the National Science Foundation under Grant No. CCF-2121952. We thank Arnold Filtser for informing us about~\cite{BFN19B}. 

	\bibliographystyle{alphaurlinit}
	\bibliography{spanner}
	
	\pagebreak
	\appendix
\section{Omitted Proofs}

\subsection{Proof  of the Decomposition Lemma}\label{subsec:proof-decomp-lemma}

We start by restating the lemma.

\DecompositionLemma*
\begin{proof}
	Assume that $G$ is given as a planar embedded graph. Let $T$ be a shortest-path tree of $G$ rooted at a vertex $r$. Let $G_{\Delta}$ is a triangulation of $G$. The root $R$ of $\mathcal{T}$ has $\chi(R) = G$ and the graph associated with $R$ is $G_{\Delta}$. 
	
	Let $X$ be a node of $\mathcal{T}$ in an intermediate step and $H_X$ be the subgraph of $G_\Delta$ associated with $X$; initially, $X = R$ and $H_X = G_{\Delta}$. We mark a constant number of special faces of $H_X$ that we call \emph{holes}. Each hole consists of two shortest paths rooted at the same vertex. If $X = R$, then $H_X$ has no hole.  We maintain the following invariant over the course of the construction.
	
	\begin{quote}
		\textbf{Separation Invariant} $H_X$ has at most 4 holes and $T_X \stackrel{\mathrm{def}}{=} T\cap H_X$ is a shortest-path tree of $H_X$. Each non-triangular face $f$ of $H_X$ is either a hole or an outer face of $H_X$, and consists of two monotone shortest paths of $T_X$ rooted at the same vertex $s$ and a (pseudo-)edge. 
	\end{quote}
	Recall that an $e$ is a pseudo-edge if $e\in E(G_{\Delta})\setminus E(G)$. The separation trivially holds when $X = R$ as $T_X = T$ and there is no hole in $G$.
	
	For each non-triangular face $f$ of $H_X$, we call the root $s$ of the two monotone shortest paths guaranteed by the Separation Invariant \emph{the root} of $f$. We triangulate $f$ by adding edges between $s$ and other vertices in $f$. Let $H_{\Delta}$ be the triangulation of $H$ obtained by triangulating every non-triangular face. We now design a weight function $\omega: V(H_X)\rightarrow \{0,1\}$ as follows:
	
	\begin{enumerate}
		\item \textbf{$H_X$ has less than $4$ holes} Let $\omega$ be the weight function that assign $w(v)= 1$  to every vertex $v \in \chi(X)$ and $0$ otherwise. 
		\item \textbf{$H$ has exactly $4$ holes} Let $\{F_i\}i=1^4$ be the holes of $H$. We choose four vertices $\{v\}_{i=1}^4$ such that $v_i$ is on the boundary of $F_i$ but not on the boundary of any other hole.  Let $\omega$ be the weight function that assigns $\omega(v)= 1$  if $v\in \{v_i\}_{i=1}^4$ and $\omega(v) = 0$ otherwise.  
	\end{enumerate}
	Let $C_X$ be the shortest path separator of $H_{\Delta}$ w.r.t $\omega$ and $T_X$ as guaranteed by~\Cref{thm:shortest-path-sep}.  Let $\Pi(X)$ to be the set of shortest paths on the boundary of the holes of $H$ and the two shortest paths constituting $C_X$. Again, when $X = R$, $\Pi(X)$ only contains two shortest paths from $C_X$.
	
	\begin{claim}\label{clm:non-intersecting} $C_X\subseteq H_X$. 
	\end{claim}
	\begin{proof}
		Suppose that $C_X$ contains an edge $e \in E(H_{\Delta})\setminus E(H_X)$. Let $F$ be the hole where $e$ is added. Let $s$ be the root of $F$. By the construction of $H_{\Delta}$, it must hold that $e = (s,v)$ for some vertex $v$ on the boundary of $F$.  Since $C_X$ is a fundamental cycle of $E(H_{\Delta})$ w.r.t the shortest path tree $T_X$, it must be that $C = T[s,v] \cup \{(s,v)\}$. Hence, $\omega(V(H_X)\cap \int(C)) = 0$ and $\omega(V(H_X)\cap \ext(T)) = W$ where $W = \sum_{v\in V(H_X)}\omega(v)$; this contradicts~\Cref{thm:shortest-path-sep}.
	\end{proof}
	
	Let $H_X^{in} = \left(H_{\Delta}\cap \int(C_X))\cup C_X\right)\setminus E(H_{\Delta})$ and  $H_X^{out} = \left(H_{\Delta}\cap \ext(C_X))\cup C_X\right)\setminus E(H_{\Delta})$. Note that, we do not keep the edges of $E(H_{\Delta})\setminus E(H_X)$ in $H_X^{in}$ and $H_X^{out}$. This will guarantee that every hole of $H$ will be a hole in $H^{in}$ or in $H^{out}$ by Claim~\ref{clm:non-intersecting}. $C_X$ will be a new hole of $H_X^{out}$.  
	
	We create two children $X_1$ and $X_2$ of $X$, and set $\chi(X_1) = \chi(X)\cap \int(C_X)$ and $\chi(X_2) = \chi(X)\cap \ext(C_X)$. We associate $H_X^{in}$ and $H_X^{out}$ with $X_1$  and $X_2$, respectively and recurse on $X_1$ and $X_2$. 
	
	If a node $X$ has $\chi(X)\leq \tau$, we \emph{stop} the recursive decomposition and $X$ will be a leaf of $\mathcal{T}$, and set $\Pi(X) = \emptyset$.	  This completes the description of the recursive decomposition.
	
	\begin{claim}\label{clm:sep-invariant} $H_X^{in}$ and $H_X^{out}$ satisfy Separation Invariant.
	\end{claim}
	\begin{proof}
		We assume inductively that $H_X$ satisfies Separation Invariant. Observe that the total  number of holes of $H_X^{in}$ and $H_X^{out}$ is at most the number of holes of $H_X$ plus $1$. Thus, if $H_X$ has at most $3$ holes, then the number of holes of $H_X^{in}$ and $H_X^{out}$ is at most $4$. 	If $H$ has exactly 4 holes, then the design of weight function $\omega$ guarantees that the number of holes of $H_X^{in}$ and $H_X^{out}$ is at most $\lfloor\frac{2}{3} \cdot 4 \rfloor  + 1 \leq 3$. Thus, the number of holes in Separation Invariant is guaranteed. In the remaining, we show other properties in Separation Invariant for $H_X^{in}$; the argument for $H_X^{out}$ is exactly the same. 
		
		By Claim~\ref{clm:non-intersecting}, any hole of $H_X^{in}$ is also a hole of $H_X$, and that $C_X$ is the outer face of $H_X^{in}$. By definition of shortest path separators, $C_X$  consists of two paths of $T_X$ rooted at a vertex $s$ and a (pseudo-)edge, say $e$, between the paths' endpoints. Let $F_{in}$ be the set of edges of $T_X$ strictly enclosed by $C_X$. By the Jordan Curve Theorem, $F_{in}\cup C_X$ is a spanning subgraph of $H_X^{in}$.  Hence $(F_{in}\cup C_X)\setminus \{e\}$ is a spanning tree of $H_X^{in}$, and it is a shortest path tree since it is a subtree of $T_X$. Thus, $H_X^{in}$ satisfies Separation Invariant. 
	\end{proof}

	We now show that $\mathcal{T}$ has all claimed properties (1)  - (4). Since  $\Pi(X)$ contains two shortest paths for each hole and two shortest paths of $C_X$, by Separation Invariant, $|\Pi(X)| \leq 10$. Thus, property (1) follows.
	
	For property (2), we observe that an internal node $X$ must have $\chi(X) > \tau$ by construction; otherwise, $X$ will be a leaf. Properties (2b)-(2c) follows directly from the construction. For (2d), we observe that any vertex $v\in \chi(X)\setminus (\chi(X_1)\cup\chi(X_2))$ must belong to the cycle separator $C_X$, hence belong to one of the two shortest paths constituting $C_X$. 
	
	Property (3) follows directly from the observation that any path from a vertex $u$ inside $H_X$ to a vertex $v$ outside $H_X$ must cross at least one of the holes of $H_X$ and that boundaries of these holes are shortest paths associated with $X$'s parent. 
	
	To show (4), we observe that when going up the tree $\mathcal{T}$ from a node $X$ to its \emph{grandparent} $Y$, $\chi(Y) \geq \frac{3}{2} \chi(X)$. It remains to bound the number of leaves of $\mathcal{T}$. 
	We will assign vertices of $S$ to each leaf of $\mathcal{T}$ in such a way that each leaf is assigned at least $\Omega(\tau)$ vertices and no vertex is assigned to more than one leaf; this will implies $O(\frac{|S|}{\tau})$ bound on the number of leaves of $\mathcal{T}$. Let $X$ be a leaf of $\mathcal{T}$, and $Y$ be $X$'s parent. By construction, $|\chi(Y)| > \tau$ and no vertex of $\chi(Y)$ belongs to any other nodes of $\mathcal{T}$ in on the same level with $Y$. Note that it could be that all vertices in $\chi(X) = \emptyset $; in this case at least $\frac{|\chi(Y)|}{3}$ active vertices belongs to shortest path separator $C_Y$ of subgraph $H_Y$  associated with $Y$. If $|\chi(X)| < \tau/3$, we assign an arbitrary subset of vertices, say $S_X$, of $\chi(Y) cap C_Y$ to $X$ such that $|S_X\cup \chi(X)|  = \tau/3$. If both children of $Y$ are leaves and are associated with less than $\tau/3$ vertices of $S$, we can always assign vertices of $|\chi(Y)|\cap C_Y$ to $Y$'s children in a way that each vertex is assigned to exactly one leaf. By this assignment scheme, no vertex of $S$ is assigned to more than one leaf, and each leaf is assigned at least $\frac{\tau}{3}$ vertices as desired.	
\end{proof}

\end{document}